\documentclass[11pt]{article}

\usepackage{amsmath}
\usepackage{epsfig, amssymb, amsfonts}
\usepackage{amsthm}
\usepackage{amstext}
\usepackage{amsopn}
\usepackage{enumerate}
\usepackage{mathrsfs}
\usepackage{subfigure}
\usepackage{graphicx}
\usepackage[left=2.3cm,top=2cm,right=2.3cm,bottom=2cm, nohead,foot=1cm]{geometry}
\numberwithin{equation}{section}

\newtheorem{theorem}{Theorem}[section]
\newtheorem{lemma}[theorem]{Lemma}
\newtheorem{assumptions}[theorem]{Rate assumptions}

\newtheorem{proposition}[theorem]{Proposition}

\newcommand{\R}{{\mathbb R}}

\newcommand{\Z}{{\mathbb Z}}

\newcommand{\Tr}{{\textup{Tr}}}

\newcommand{\C}{{\mathbb C}}

\title{A ballistic motion disrupted by  quantum reflections  }

\author{\textbf{Jeremy Thane Clark}\footnote{jtclark@math.msu.edu} \\  Department of Mathematics, Michigan State University  \\ East Lansing, MI 48824, USA }

\begin{document}
\maketitle

\begin{abstract}
 I study a Lindblad dynamics modeling a quantum test particle in a Dirac comb that collides with particles from a background gas.  The main result is a homogenization theorem  in an adiabatic  limiting regime involving  large initial momentum for the test particle.    
Over the time interval considered, the particle would exhibit essentially ballistic motion if either the singular periodic potential or the kicks from the gas were removed.  However, the particle behaves diffusively when both sources of forcing are present.  The conversion of the motion from ballistic to diffusive is generated by occasional quantum reflections that result when the test particle's momentum is driven through a collision near to an element of the half-spaced reciprocal lattice of the Dirac comb.

\end{abstract}

\section{Introduction}\label{SectIntro}

  In this article I demonstrate that a quantum  particle in a  one-dimensional, periodic potential field with singular peaks can exhibit a strong diffraction-generated effect in the presence of a noise.  In fact, the noise has a cooperative role in driving the test particle into temporary states in which coherence is developed through the Hamiltonian dynamics.   Coherent superpositions enter the dynamics primarily between certain pairs of transmitted and reflected waves that  quickly collapse  into classical superpositions through the noise.  After the quantum superposition collapses into a classical superposition, the particle has a non-negligible chance of moving with roughly the negative of its original momentum, i.e., having been reflected.   These coherent superpositions of transmitted and reflected waves occur on a slower time scale than the noise, at instances in which the test particle's momentum is randomly kicked close to a value such that the de Broglie wavelength is an integral fraction of $\pi^{-1}$ multiplied by the potential's period, i.e., momenta satisfying the Bragg condition for reflection.     Nevertheless, the test particle undergoes  reflections frequently enough to restrain its motion, yielding diffusive transport  in contrast to the ballistic transport dominating when either the singular periodic potential or the noise is removed.  If the  periodic potential is smooth rather than singular, then diffractive effects will be negligible.

First I will present a classical analog of the quantum model that I investigate.   Consider the following  forward Kolmogorov equation  for  probability densities $\mathcal{P}_{t}(x,p)\in L^{1}(\R^{2})$ depending on $t\in \R^{+}$:
\begin{align}\label{GinForTin}
\frac{d}{dt}\mathcal{P}_{t}(x,p)= -2p\frac{\partial}{\partial x}\mathcal{P}_{t}(x,p)+\frac{dV}{dx}(x)\frac{\partial}{\partial p}\mathcal{P}_{t}(x,p)+\int_{\R}dv\,j(v)\Big(\mathcal{P}_{t}(x,p-v) -\mathcal{P}_{t}(x,p) \Big),
\end{align}
where $V:\R\rightarrow \R^{+}$ is smooth and periodic with period $2\pi$, and  $j\in L^{1}(\R,\R^{+})$ with $j(v)=j(-v)$ and $\int_{\R}dv\, j(v)v^{2}<\infty $.   The master equation~(\ref{GinForTin}) describes the time evolution of the phase space densities for a particle in dimension one that is governed by the Hamiltonian $H(x,p)=p^{2}+V(x)$ and receives random momentum kicks of size $v\in \R$ with rate density $j(v)>0$.   Let $X_{t}$ and $P_{t}$ be the position and momentum processes, respectively, whose densities are determined by~(\ref{GinForTin}) and have initial values $X_{0}=0$ and $P_{0}=\frac{\mathbf{p}_{0}}{\lambda}$ for $\mathbf{p}_{0}>0$ and scale parameter $ 0<  \lambda  \ll 1  $.     In other words, the particle starts from the origin and moves to the right  with a high speed.   The displacement in momentum, $ P_{t}-P_{0}$, is the sum of the forcing contribution  $D_{t}=-\int_{0}^{t}dr\,\frac{dV}{dx}(X_{r})$ and a noise contribution given by the sum of the momentum jumps to have occurred by time $t$.    As long as the momentum remains high, the forcing term $D_{t}$ will yield only a small contribution due to  self-cancellation in the integral of  $\frac{dV}{dx}(X_{r})$   as the particle passes quickly through the period cells of $V(x)$.   Thus, since the jump rates have the symmetry $j(v)=j(-v)$, the momentum process $P_{t}$ is behaving roughly as a symmetric random walk starting from $ P_{0}=\frac{\mathbf{p}_{0}}{\lambda}$.   Based on these observations, it follows that for fixed $T>0$ and $\eta \in (0,2)$ there is a weak law of large numbers as $\lambda\searrow 0$ given by
\begin{align}\label{LawOfLarge}
 \lambda^{\eta+1}   X_{\frac{T}{\lambda^{\eta } }}   \Longrightarrow   2\mathbf{p}_{0}T.
\end{align}
The above limit describes ballistic motion with velocity $2\mathbf{p}_{0}$ for the particle on length scales  $\propto \lambda^{-\eta-1}$ and  time scales $\propto \lambda^{-\eta}$ since the displacement in the particle's position is deterministic and proportional to the time parameter $T$.  The convergence~(\ref{LawOfLarge}) follows since $X_{\frac{T}{\lambda^{\eta}  }}=2\int_{0}^{\frac{T}{\lambda^{\eta} }}dr\,P_{r} $, and the momentum $P_{r}$ will typically not deviate much from the initial value $P_{0}$ over the interval $r\in [0,\frac{T}{\lambda^{\eta} }]$  in the sense that $\sup_{ r\in [0,\frac{T}{\lambda}]}\frac{|P_{r}-P_{0}|   }{  |P_{0}|}  $  will typically be small for $\lambda \ll 1$.   In summary, the particle behaves ballistically over time scales on the order of $ \frac{T}{\lambda^{\eta } }$ for $\lambda\ll 1$ and $\eta\in (0,2)$ when starting from a momentum $\propto \lambda^{-1}$.   This  holds simply because the periodic force and random kicking do not have time to generate a  shift in momentum on the scale of its initial value.

Next I present a quantum analog of~(\ref{GinForTin}).    In the quantum version, the state $\rho_{t}$ of the particle at time $t\in \R^{+}$ is a density matrix in the space of trace class operators 
 $\mathcal{B}_{1}(\mathcal{H})$ over the Hilbert space $\mathcal{H}=L^{2}(\R)$ whose time evolution is determined by the Lindblad equation~(\ref{TheModel}) below.    I take the initial state  $\check{\rho}_{\lambda}$ to be a right-traveling wave with momentum concentrated around the value $\mathbf{p}=\frac{\mathbf{p}_{0} }{\lambda} $ for $\mathbf{p}_{0}$ and $\lambda$ as before:   More specifically,  $\check{\rho}_{\lambda}:=|\frak{h}\rangle\langle \frak{h}|$ for a  wave function $\frak{h}\in \mathcal{H}$ having the  form $ \frak{h}(x)=  e^{\textup{i}x\mathbf{p}}\frak{h}_{0}(x)$ in the position representation for  $\frak{h}_{0}\in \mathcal{H}$ satisfying
\begin{align}\label{Gauss}  
\big\|  X^{2} \frak{h}_{0}\big\|_{2}<\infty \quad \text{and}\quad   \big\|  P^{2} \frak{h}_{0}\big\|_{2}<\infty,
\end{align}
where $X$ is the position operator and $P:=-\textup{i}\frac{d}{dx}$ is the momentum operator.    For the Hamiltonian $H$ and the completely positive map $\Psi:\mathcal{B}_{1}(\mathcal{H})$ defined below, the dynamics in the Schr\"odinger  representation is determined by the quantum Kolmogorov equation
\begin{align}\label{TheModel}
\frac{d}{dt}\rho_{t}= -\textup{i} \big[H,\rho_{t}\big]+\Psi(\rho_{t})-\frac{1}{2}\big\{ \Psi^{*}(I),\rho_{t} \big\}.
 \end{align}
 The Hamiltonian is a Schr\"odinger operator  $H= P^{2}+V(X)$, where $V:\R\rightarrow \R^{+}$ is smooth and has period $2\pi$.    The map  $\Psi$ describes the noise acting on the particle and has the continuous Kraus form
\begin{align}\label{TheNoise}
 \Psi (\rho)=   \int_{\R} dv\,j(v)e^{\textup{i} vX}\, \rho\, e^{-\textup{i} v X},
 \end{align}
where  $ \rho\in \mathcal{B}_{1}(\mathcal{H})$ and $j(v)$ is defined as before.   In~(\ref{TheModel}) $\Psi^{*}(I)$ is the adjoint map $\Psi^{*}$ evaluated for the identity operator $I$ on $\mathcal{H}$, and it is easy to compute that $\Psi^{*}(I)= \mathcal{R}\, I  $, where $\mathcal{R}:=\int_{\R}dv\, j(v)  $.

Periodic potential fields for atoms and molecules have been produced in laboratory settings using lasers for the examination of several fundamental quantum phenomena, which include  Bose-Einstein condensation, Bloch oscillations, Zener tunneling, and Bragg scattering~\cite{Adams,Feldmann,AndersonJr}.   The physics literature on decoherence and matter-wave optics includes many experimental~\cite{Lucia,Klaus} and theoretical~\cite{Ghirardi, ESL, Sipe,Exper,Vacchini} studies in which a quantum noise of the form~(\ref{TheNoise}) appears.  The noise map $\Psi$ satisfies the Weyl covariance relation 
\begin{align}\label{Friction}
 \hspace{2.5cm}   \Psi\big(e^{\textup{i}aP+\textup{i}bX  }   \rho  e^{-\textup{i}aP-\textup{i}bX }\big)= e^{\textup{i}aP+\textup{i}bX  } \Psi(\rho)e^{-\textup{i}aP-\textup{i}bX  }, \hspace{1.5cm} a,b\in \R . 
 \end{align}
This means that the rate of collisions is invariant of both the position and the momentum for the test particle.  In particular, the noise does not generate dissipation in energy since it does not include any frictional contribution  that would systematically drag a ``high momentum" down to lower values.  The noise predicts a gradual stochastic acceleration to higher momenta over time, which is apparent in the linear mean energy growth found in the model:  
$$ \Tr[H\rho_{\lambda,t}]= \Tr[H\rho_{\lambda,0} ]+ \sigma t.  $$
Nevertheless, the model can be a useful description for a massive particle interacting with a gas of light particles over a limited time period.  A three-dimensional version is derived from a quantum linear Boltzmann equation in~\cite[Sect.7.1]{VaccHorn}.  Also, there is  a mathematical derivation from a singular coupling limit in~\cite{Hellmich}.   
Attention to the structure of completely positive maps and quantum dynamical semigroups satisfying the symmetry~(\ref{Friction}) and other classes of symmetries can be found in the work of Holevo~\cite{Holevo}.

The position density for the quantum particle at time $t$, denoted $D_{\lambda,t}$, is given by the diagonal of the integral kernel of the density matrix $\rho_{t}$ in the position representation:  $D_{\lambda,t}(x):=\rho_{t}(x,x)$.  The subscript $\lambda>0$ is a reminder that the initial density matrix $\check{\rho}_{\lambda}$ depends on $\lambda$ through the momentum $\mathbf{p}=\frac{\mathbf{p}_{0}}{\lambda} $.   Let $\mu_{\lambda, T}$ be the probability  measure with density   $\lambda^{-\eta-1} D_{\lambda,\frac{T}{\lambda^{\eta}} }(\lambda^{-\eta-1} x)$.   The quantum analog of the weak law of  large numbers in~(\ref{LawOfLarge})  is  the weak convergence as $\lambda\searrow 0$ given by
\begin{align}\label{LawOfLargeTwo}
  \mu_{\lambda,T}      \Longrightarrow   \delta_{T\mathbf{p}_{0}}   ,     
\end{align}
where $\delta_{x}$ is the delta distribution at $x\in \R$.   The convergence~(\ref{LawOfLargeTwo}) holds when $V(x)$ is smooth, however, my focus here is on the situation in which the periodic potential $V(x)$ has singular peaks.  

  To study the situation of a potential with singular peaks, I will take the Hamiltonian $H$ to be, very specifically, a Dirac comb of strength $\alpha>0$, which is  formally expressed by
 $$H= P^{2}+\alpha\sum_{N\in \Z}\delta\big(X-2\pi N\big).  $$ 
The advantage of the Dirac comb Hamiltonian is that there are closed forms available for the eigenvalues and eigenkets; see~\cite[Sect.III.2.3]{Solve}.    I will  make an additional technical assumption that the noise term operates on a longer time scale than the Hamiltonian dynamics by replacing $\Psi$ in~(\ref{TheModel}) by $\Psi_{\lambda}:=\lambda^{\varrho}\Psi$ for  $\varrho > 0   $ and $0<\lambda \ll  1$.    Notice that a slow-acting noise only strengthens the heuristic reasons given above~(\ref{LawOfLarge}) for why ballistic motion should be expected.   Let $\rho_{\lambda,t}$ for $t\geq 0$ be the solution to~(\ref{TheModel}) with rescaled noise term $\Psi_{\lambda}$ and   Dirac comb Hamiltonian.    Also define $\hat{\mu}_{\lambda, T}$ to be the probability measure with density  $ \lambda^{-\frac{\eta+\varrho+3}{2}    } \hat{D}_{\lambda,\frac{T}{\lambda^{\eta}  } }\big( \lambda^{-\frac{\eta+\varrho+3}{2}    } x \big)$, where $\hat{D}_{\lambda,t}(x)=\rho_{\lambda,t}(x,x) $ is the position density at time $t$.  I will prove that for $   \varrho >\frac{\eta}{2}     $ and $ 1<\eta-\varrho<2$  that  there is weak convergence as $\lambda\searrow 0$ given by
\begin{align}\label{CentralLimit} 
\hat{\mu}_{\lambda, T}\Longrightarrow \mathcal{N}\big(0,T\vartheta\big),
\end{align}
where  $\vartheta := \frac{16\mathbf{p}_{0}^{3}}{\alpha \mathcal{R}}$ for $\mathcal{R}:=\int_{\R}dv\,j(v)$ and $\mathcal{N}(a,b)$ is the normal distribution with mean $a$ and variance $b$.   Notice that the spatial scale $ \lambda^{-\frac{\eta+\varrho+3}{2}    } $ for the convergence in~(\ref{CentralLimit}) is smaller than the spatial scale $\lambda^{-\eta-1}$ for the convergence in~(\ref{LawOfLargeTwo}) by our restriction $\eta-\varrho >1$.   The central limit convergence~(\ref{CentralLimit}) describes diffusive behavior for the particle rather than ballistic motion.   The change from ballistic to diffusive behavior is induced by quantum reflections that occur when the momentum of the particle is kicked near the reciprocal  lattice $\Z$ of the Dirac comb; see the semi-classical heuristics below.    I conjecture that the  adiabatic assumption built into the model by replacing $\Psi$ with the slow-acting noise $\Psi_{\lambda}$ is not necessary:  if $\varrho=0$ the central limit convergence~(\ref{CentralLimit}) will hold for some  diffusion constant $\vartheta$.   Although the adiabatic regime dilutes this diffusion result, it  also reinforces the reasons for why ballistic motion should be expected in the contrasting classical model described above.

Next I will describe the mechanism underlying the central limt theorem~(\ref{CentralLimit}).
 A one-dimensional plane wave $|\mathbf{p}\rangle$ with momentum $\mathbf{p}\in \R $ evolving through a  Schr\"odinger Hamiltonian $H=P^{2}+V(X)$ for a period-$2\pi$ potential $V:\R\rightarrow \R^{+}$ will develop into discrete superpositions of plane waves of the form
\begin{align}\label{Discrete}
   e^{-\textup{i}t H }  |\mathbf{p}\rangle= \sum_{n\in \Z}C_{t}(n,  \mathbf{p}) |\mathbf{p}+n\rangle    
\end{align}
for some coefficients  $C_{t}(n,  \mathbf{p})\in \C$ satisfying $\sum_{n}|C_{t}(n,  \mathbf{p})|^{2}=1$.  The discrete form of the superpositions is a consequence of Bloch theory.    If the initial momentum $\mathbf{p}$  is large in the sense that $\mathbf{p}^{2}\gg \sup_{x}V(x)$, then the wave will tend to transmit freely through the potential, i.e., $|C_{t}(0,  \mathbf{p})|^{2}\approx 1$,  unless $\mathbf{p}$ is ``very close" to an element of the lattice $\frac{1}{2}\Z$.  Momenta in $\frac{1}{2}\Z$ satisfy the Bragg condition for reflection  by the periodic potential, and a plane wave $ |\mathbf{p}\rangle  $  with momentum  close enough to satisfying the Bragg condition $\mathbf{p}\approx \frac{\mathbf{n}}{2}\in\frac{1}{2}\Z$ will  evolve nearly as a superposition of a transmitted and an essentially reflected wave: 
\begin{align}\label{Pendo}
e^{-\textup{i}t H }  |\mathbf{p}\rangle\approx  C_{t}(0,  \mathbf{p}) |\mathbf{p}\rangle + C_{t}(-\mathbf{n},  \mathbf{p}) |\mathbf{p}-\mathbf{n}\rangle   . 
\end{align}
The adjective ``reflected" is justified since $\mathbf{p}-\mathbf{n}\approx -\mathbf{p}$ when  $\mathbf{p}\approx \frac{\mathbf{n}}{2}$.   The superposition~(\ref{Pendo}) will evolve roughly periodically through different weights upon the transmitted and reflected waves, and this behavior is termed \textit{Pendell\"osung oscillations}.    I will refer to the neighborhood of momenta around  the lattice point $\frac{\mathbf{n}}{2}\in \frac{1}{2}\Z$  in which ``non-negligible" reflected waves appear as the \textit{reflection band}; see~\cite{Friedman}.     If the particle's momentum is knocked into a reflection band, then the particle will develop into a quantum  superposition of a transmitted and a reflected wave~(\ref{Pendo}).   Intuitively, this quantum superposition will collapse into a classical superposition after the particle receives another momentum kick, and there is a non-negligible chance that the particle will be heading in the opposite direction that it started with  after the full process of being kicked in and out of a reflection band.  Note that the momentum kicks are small compared to the initial momentum, and thus the  reflections are the most dramatic change occurring in this process.  

The difference in transport  behavior between a smooth potential~(\ref{LawOfLargeTwo}) and the Dirac comb potential~(\ref{CentralLimit}) involves the rate at which reflections occur.      The width of the reflection bands for an infinitely differentiable potential decay superpolynomially fast  as $|\frac{\mathbf{n}}{2}|\nearrow \infty$, whereas the width of the reflection bands for the Dirac comb decay as 
$\frac{\alpha }{8|\frac{\mathbf{n}}{2}|}+\mathit{O}\big( \frac{1}{|\frac{\mathbf{n}}{2}|^{2}}  \big) $.   Other types of periodic singular potentials will have different decay rates.   The comparatively  large width of the reflection bands for the Dirac comb make it much more likely that a particle with high momentum $\mathbf{p}=\frac{\mathbf{p}_{0}}{\lambda}\gg \mathbf{p}_{0}$ will land in a reflection band after receiving a random momentum kick.  Since the width of the reflection bands near $\mathbf{p}$ are approximately $\frac{\alpha }{8\mathbf{p}}$ and are located around the lattice points $\frac{1}{2}\Z$, the probability of being kicked randomly from momentum $\mathbf{p}$ into a nearby reflection band  is approximately $2\times \frac{\alpha }{8\mathbf{p}}=\frac{\alpha }{4\mathbf{p}}$.    With the rescaled noise term  $\Psi_{\lambda}:=\lambda^{\varrho}\Psi$, the collisions occur with Poisson rate $\lambda^{\varrho}\mathcal{R} $ for $\mathcal{R}:=\int_{\R}dv\, j(v)$, and   the length of the random time intervals of uninterrupted ballistic motion are essentially exponentially distributed with mean $\frac{4\mathbf{p}}{\alpha\lambda^{\varrho} \mathcal{R}  }   $.     Thus, over a time interval of length $ \frac{T}{\lambda^{\eta}}  $,  there will be on the order of   
$$ \frac{   \frac{T}{\lambda^{\eta}}  }{ \frac{4\mathbf{p}}{\alpha\lambda^{\varrho} \mathcal{R}  }   } =\lambda^{\varrho-\eta+1}\frac{ \alpha\mathcal{R}T  }{4 \mathbf{p}_{0}   }   $$
 reflections, which is $\gg 1$ for $\lambda\ll 1$ by my constraint  $ 1<\eta-\varrho  $.  Since many reflections occur over the time interval considered, the particle is not allowed to move in a single direction for a long time period, and there will be much cancellation between the particle's rightward and leftward movements.   For independent mean-one exponentials $(\mathbf{e}_{n})_{n\geq 1}$, I think of  the position of the particle as behaving in a similar  way to the following sum:
\begin{align}\label{Projectile}
2\mathbf{p} \sum_{n=1}^{\lambda^{\varrho-\eta+1}\frac{ \alpha\mathcal{R}T  }{ 4 \mathbf{p}_{0}   }   }  \frac{4\mathbf{p}}{\alpha\lambda^{\varrho} \mathcal{R}  }   (-1)^{n-1}\mathbf{e}_{n}\approx  \frac{8\mathbf{p}_{0}^{2} }{\alpha \mathcal{R}  \lambda^{\varrho+2} }\sum_{n=1}^{\lambda^{\varrho-\eta+1} \frac{   \alpha\mathcal{R}T  }{  \mathbf{p}_{0}   }  }     (-1)^{n-1}(\mathbf{e}_{n}-1) ,  
\end{align}
where the approximation occurs by throwing away $\frac{8\mathbf{p}_{0}^{2} }{\alpha \mathcal{R}  \lambda^{\varrho+2} }$ when the sum has an odd number of terms.   The expression~(\ref{Projectile}) treats the particle as alternating between velocities $\pm 2\mathbf{p}$ for  time periods of length   
$\frac{4\mathbf{p}}{\alpha\lambda^{\varrho} \mathcal{R}  } \mathbf{e}_{n}$.  When multiplied by $ \lambda^{\frac{\eta+\varrho+3}{2}    } $, the sum above  is close in distribution to  $\mathcal{N}\big(0, T\vartheta\big)$ for $\lambda\ll 1$ and $\vartheta:=\frac{16\mathbf{p}_{0}^{3} }{\alpha \mathcal{R}  }$  by the central limit theorem.

The mathematical results of this article extend those from~\cite{Dispersion}, which considered the same quantum model under different parameter scalings.      The previous work focused, firstly, on the derivation of a classical Markovian dynamics in an adiabatic limit and, secondly,  on a central limit theorem for the time integral of the limiting classical process, which intuitively corresponds to the  position of the particle.   Because the central limit theorem in~\cite{Dispersion} concerns a classical process arising in a limit of the quantum model, it is not  clear from that work whether the  the anomalous transport behavior found in the classical model actually holds in some form for the original quantum model.   The current article addresses this issue by integrating the adiabatic limit into a central limit theorem for the quantum model.   I believe that the adiabatic limit is not required for this result.

This article is organized as follows:  Section~\ref{SecConv} adds  more commentary on the dynamics and introduces some  notational conventions.  In Sect.~\ref{SecThmMain}  I  present the main results of this article and sketch their proofs.  Sections~\ref{SecBlochFiber} and~\ref{SecPseudoPoisson} discuss important symmetries and decompositions for the dynamics.    Sections~\ref{SecQuasiMomentum}-\ref{SecTransition} bound the differences between a series of intermediary dynamics between the original quantum dynamics and a limiting classical dynamics studied in Sect.~\ref{SecClassical}.

\section{Preliminary discussion} \label{SecConv}

\subsection{The Dirac comb Hamiltonian}\label{SecConvComb}

 The Schr\"odinger Hamiltonian $H=P^{2}+\alpha\sum_{n\in\Z}\delta(X-2\pi n)$, $\alpha>0$ is defined mathematically as a particular self-adjoint extension of the symmetric operator $-\frac{d^{2}}{dx^{2}}$ with domain consisting of all $L^{2}$-functions on $\R$ with two weak derivatives in $L^{2}(\R)$ and that take the value $0$ on the lattice $2\pi\Z$, i.e., $f\in\mathbf{H}^{2,2}(\R)\cap \{g\,|\,g(2\pi n)=0,\,n\in\mathbb{Z} \}$.  The operator domain of the self-adjoint extension is the  space of functions  that have one weak $L^{2}$-derivative in the domain $\R$ and two weak $L^{2}$-derivatives in $\R-2\pi\Z$, and that satisfy the boundary condition
$$  \alpha f(2\pi n)= \frac{df}{dx}( 2\pi n+ )-\frac{df}{dx}(2\pi n-)   $$
for each $n\in \Z$.  

The spectrum of the Hamiltonian $H$ is continuous, and there are a continuum of kets $|p\rangle_{\scriptscriptstyle{Q} }$, $p\in \R$ and a dispersion relation $E:\R\rightarrow \R^{+}$ such that the Hamiltonian can be formally written as
$$H=\int_{\R}dp\,E(p)|p\rangle_{\scriptscriptstyle{Q} }\,{ }_{\scriptscriptstyle{Q} } \langle  p|.$$
This representation is related to standard Bloch theory through the extended-zone scheme, and  I will discuss the connection in Sect.~\ref{SecBlochFiber}.  For the Dirac comb, the dispersion relation and eigenkets have closed forms.   The dispersion relation is given by $E(p)=\mathbf{q}^{2}(p)$ for the anti-symmetric, increasing function $\mathbf{q}:\R\rightarrow \R$ determined by the convention $\mathbf{q}(\frac{n }{ 2})=\frac{ n}{ 2}$ for $n\in \Z-\{0\}$, the Kr\"onig-Penney relation for  $p\in \R-\frac{  1 }{ 2}\Z$:
\begin{align}\label{KronigPenney}
\hspace{2cm}\cos\big(2\pi  p\big)=\cos\big(2\pi  \mathbf{q}(p) \big)+\frac{\alpha }{2\mathbf{q}(p) }\sin\big(2\pi \mathbf{q}(p)\big), 
\end{align}
and continuity at zero.  The dispersion relation $E(p)$ has a roughly parabolic shape $E(p)\approx p^2+\frac{\alpha}{2\pi }$ except for momenta near  the lattice values $p=\frac{ n  }{ 2}$,  $n\in \Z-\{0\}$ where there are gaps $g_{n}:=\big|E(\frac{ n  }{ 2}+)-E(\frac{ n}{ 2}-)\big|$.  The kets $|p\rangle_{\scriptscriptstyle{Q} }$ can be constructed as discrete combinations of the standard momentum kets:  
\begin{align}\label{HairCut}
 |p\rangle_{\scriptscriptstyle{Q} }=\sum_{m\in \Z} \eta(p,m)\big|p+m\big\rangle,  
 \end{align}
 where the coefficients $\eta(p,m)$ satisfy $\sum_{m\in \Z}|\eta(p,m)|^{2}=1$ and have the form
$$ \eta(p,m):=-\textup{i} N_{p}^{-\frac{1}{2}}\big(e^{ \textup{i}2\pi (\mathbf{q}(p)-p)}-1 \big) \Big(\frac{1}{\mathbf{q}(p)+p+m}+\frac{1}{\mathbf{q}(p)-p-m} \Big)$$  
for a normalization constant $N_{p}>0$.

\subsection{The dynamics and rescaling time  }

In the introduction I defined the state $\rho_{\lambda, t}$ to be the solution of the master equation~(\ref{TheModel}) with $H$ being  the Dirac comb Hamiltonian and the noise term $\Psi$ replaced by  $\lambda^{\varrho}\Psi$ for $0<\lambda\ll 1$.   However, it will be convenient to reset this notation and rescale time so that the $\lambda$-dependent scale factor  appears in front of the Hamiltonian term:   let $\rho_{\lambda, t}$ be the solution to the quantum Kolmogorov equation
  \begin{align}\label{RETheModel}
\frac{d}{dt}\rho_{\lambda, t}= -\frac{\textup{i} }{\lambda^{\varrho}}\big[H,\rho_{\lambda, t}\big]+\Psi(\rho_{\lambda, t})-\frac{1}{2}\big\{ \Psi^{*}(I),\rho_{\lambda, t} \big\},
 \end{align}   
where $\rho_{\lambda, 0}:=\rho$ for some density matrix $\rho\in \mathcal{B}_{1}(\mathcal{H})$.   There exists   a unique strongly continuous semigroup of completely positive maps $\Phi_{\lambda, t}:\mathcal{B}_{1}(\mathcal{H})$ that preserve trace, i.e., $\Tr[\Phi_{\lambda, t}(\rho)]=\Tr[\rho]$, such that $ \rho_{\lambda, t}:=\Phi_{\lambda, t}(\rho)$ solves~(\ref{RETheModel}).    Technical questions regarding the construction and uniqueness of the semigroup  $\Phi_{\lambda, t}$ do not follow from Lindblad's basic result~\cite{Lindblad} since the Hamiltonian part of the generator is unbounded, however,  these technical issues are trivial due to the existence of the unitary group $U_{t}=e^{-\textup{i}tH}$ and the exceptionally simple form of the noise map $\Psi$: see~\cite[Appx.A]{Dispersion}.

In future $\rho_{\lambda, t}$ will denote the solution to~(\ref{RETheModel}) with initial matrix $\rho_{\lambda, 0}:=\check{\rho}_{\lambda}$ defined above~(\ref{Gauss}).     The rescaling of time in~(\ref{RETheModel}) makes a slight change in the statement of the main result~(\ref{CentralLimit}) since the final time becomes $\frac{T}{\lambda^{\gamma}}$ for  $\gamma:=\eta-\varrho  $.   Note that the exponent $\gamma$ satisfies  $\gamma<\varrho$ and  $1<  \gamma  <2$ by the constraints above~(\ref{CentralLimit}).  Define  $D_{\lambda,t}(x):=\rho_{\lambda,t}(x,x) $ to be the position density at time $t$.  
  Under the new conventions, the weak convergence~(\ref{CentralLimit}) becomes $\mu_{\lambda, T}\Longrightarrow \mathcal{N}\big(0,T\vartheta\big)$ where  $\mu_{\lambda, T}$ is defined as the probability measure with density  $ \lambda^{-\frac{\eta+2\varrho+3}{2}    } D_{\lambda,\frac{T}{\lambda^{\gamma}  } }\big( \lambda^{-\frac{\eta+2\varrho+3}{2}    } x \big)$.

\subsection{A classical dynamics}\label{SecPreClassical}

The analysis in future sections reduces the proof of the weak convergence  $\mu_{\lambda, T}\Longrightarrow \mathcal{N}\big(0,T\vartheta\big)$  to the study of a classical Markovian process $(Y_{t},K_{t})\in \R^{2}$, where $Y_{t}$ and $K_{t}$ are defined below and have units of position and momentum, respectively.   This classical process was my primary focus in~\cite{Dispersion}, and Sect.~\ref{SecClassical} revisits some of the proofs that simplify due to the simpler setup of this article.  The relative simplicity here is a consequence of the short time scales considered over which the position process does not have variable-rate diffusive behavior.  For instance, if the time scale $\frac{T}{\lambda^{\gamma}}$  for $\gamma \in (1,2) $  were replaced by $\frac{T}{\lambda^{2} }$, then there would be fluctuations in the momentum on the order of the initial momentum $\mathbf{p}=\frac{\mathbf{p}_{0}}{\lambda}$.  However, this clearly complicates a characterization of the position distribution since the diffusion rate $\vartheta=\frac{16 \mathbf{p}_{0}^{3}  }{\alpha \mathcal{R}}$ appearing in the limit distribution $\mathcal{N}\big(0,T\vartheta\big)$   depends on the proportionality constant $\mathbf{p}_{0}$.

 I   will take the initial distribution of the Markov process $(Y_{t},K_{t})$ to be $\delta_{0}(y) |\frak{h}(p)|^{2}$.    In words, the particle begins at the origin and has a probability density in  momentum   given by $|\frak{h}(p)|^{2}$, which is closely concentrated around the value $\mathbf{p}>0$ by our assumptions~(\ref{Gauss}) on $\frak{h}\in L^{2}(\R)$.   The  position component is a time integral of the momentum: $Y_{t}=\frac{2}{\lambda^{\varrho}}\int_{0}^{t}dr\, K_{r}$, and  the momentum component is a Markovian jump process for which the rate of jumps $J(p,p')$  from $p'\in \R$ to $p\in \R$ is given by  
$$  J(p,p'): =  \sum_{n\in \Z}j\big(p-p'-n\big)\left|\kappa_{p-p'-n }\big(p',\,n\big) \right|^{2},  $$
where the values $ \kappa_{v }(p,n)\in \C $ are defined through the coefficients appearing in~(\ref{HairCut}) as
\begin{align}\label{TheKappas}
\kappa_{v }(p,n):= \sum_{m\in \Z }\overline{\eta}\big(p+v+n,m-n \big) \eta(p,m).   
 \end{align}
The values $\left|\kappa_{v}(p,\,n) \right|^{2}$ satisfy the normalization
$\sum_{n\in \Z} \left|\kappa_{v}(p,\,n) \right|^{2} =1 $.  It follows that the escape rate $\mathcal{R}=\int_{\R}dp\,  J(p,p')$ is invariant of the current momentum $p'$.  A jump for the process $K_{t}$ from a momentum $p'\in \R$  can be understood as composed of a sum $v+n$ in which $v\in \R$ has density $\frac{j(v)}{\mathcal{R}}$ and $n\in \Z$ has conditional probabilities $\left|\kappa_{v}(p',\,n) \right|^{2}$ given $p'$ and $v$.  The component $v$ is the momentum transfered to the test particle through a collision with a gas particle, and  $n$ is the momentum transfered through diffractive scattering.

\subsection{Some notational conventions}

The following is a partial summary of notation.  Let $\mathbf{h}$ be a Hilbert space, $\mathbf{b}$ be a Banach space, $\frak{m}$ be a measure space. 
\begin{eqnarray*}
\textbf{Sets:} \quad & \mathbb{T}  & \text{Torus identified with the Brillouin zone  $\big[-\frac{1}{2},\frac{1}{2}\big)$ }\\
&  \mathcal{I}   & \text{Spatial interval $[-\pi,\pi )$ }\\
&       &   \\
 \textbf{Spaces:} \quad & \mathcal{B}(\mathbf{h}) & \text{Bounded operators over the Hilbert space $\mathbf{h}$ }\\
&  \mathcal{B}_{1}(\mathbf{h})  & \text{Trace class operators over the Hilbert space $\mathbf{h}$ }\\  
&  \mathcal{B}_{2}(\mathbf{h})  & \text{Hilbert-Schmidt operators over the Hilbert space $\mathbf{h}$}\\  &   &  \\
 \textbf{Norms:} \quad & \|f\|_{p} & \text{$L^{p}$-norm for $p\in[1,\infty]$ and a measurable function $ f:\frak{m}\rightarrow \mathbf{b} $ }\\  
&  \|G\|     & \text{Operator norm for a linear map $G:\mathbf{b}\rightarrow \mathbf{b}$  }\\
&   \|\rho\|_{\mathbf{1}}      & \text{Trace norm for an element $\rho\in \mathcal{B}_{1}(\mathbf{h})$ }
\end{eqnarray*}
Note that the trace norm has a boldface subscript.

\section{Statements of the main results with proof sketches}\label{SecThmMain}

This section outlines the proofs for the main results of this article.   Section~\ref{SecDiff} is devoted to the situation in which both the Dirac comb and the random kicking contribute to the dynamics and the resulting behavior of the particle is diffusive.  Section~\ref{SecBall} discusses the simpler cases in which  ballistic motion prevails because either the Dirac comb or the random kicking is not present.

\subsection{The diffusive case}\label{SecDiff}

  I will have the following technical assumptions on the jump rates $j(v)$:
\begin{assumptions}\label{Assumptions} \text{ }
There is a $\varpi>0$ such that
\begin{enumerate} 

\item $  \int_{\R}dv\,j(v)e^{a|v|}<\varpi$ for some $a>0$,

\item $\sup_{-\frac{1}{4}\leq \theta \leq \frac{1}{4}} \sum_{n\in \Z} j\big(\theta+\frac{n}{2}\big)<\varpi $,    

\item $\inf_{v\in [-1,1]   } j(v)\geq \varpi^{-1}$.

\end{enumerate}

\end{assumptions}

Recall that $\vartheta:=\frac{16\mathbf{p}_{0}^{3}}{\alpha \mathcal{R} }$.  The theorem below states that the position distribution for the test particle  at time $\frac{T}{\lambda^{\gamma}}$ for $\lambda\ll 1$ and $\gamma \in (1,2)$ is approximately a Gaussian distribution with variance $T\vartheta$ when normalized by the factor $ \lambda^{-\frac{\gamma+2\varrho+3  }{2} }$.  

 \begin{theorem}\label{ThmMain}  Pick  $1<\gamma<2$ and $\varrho>\gamma$.
Let $\rho_{\lambda, t}$ be the solution to~(\ref{RETheModel})  and $ D_{\lambda,t}\in L^{1}(\R) $ be the density determined by $\rho_{\lambda, t}$ for the position distribution.  Define the measure $\mu_{\lambda,T}$ to have the density: 
   $$\frac{d\mu_{\lambda,T}}{dx}(x)= \lambda^{-\frac{\gamma+2\varrho+3  }{2} }D_{\lambda,\frac{T}{\lambda^{\gamma} }}\big( \lambda^{-\frac{\gamma+2\varrho+3  }{2} } x   \big).$$   For each $T\in \R^{+}$, the measures $\mu_{\lambda,T}$ converge in law as $\lambda\searrow 0$ to a mean zero Gaussian with variance $T\vartheta$.

 \end{theorem}

I will go through the main part of the proof for Thm.~\ref{ThmMain} after  defining notation and stating the main technical results that the argument depends upon.  The goal of the analysis in the quantum setting is to approximate the relevant quantities by  analogs  corresponding to the classical process $(K_{t},Y_{t})$ discussed in Sect.~\ref{SecPreClassical}.  To reach the classical quantities, there are three intermediary approximations that are roughly characterized by the following:       
\begin{enumerate}[(I).]

\item  \textbf{Momentum to extended-zone scheme momentum approximation:} The integral kernel $\big\langle p_{1}  \big|\rho_{\lambda,\frac{T}{\lambda^{\gamma} } } \big|  p_{2}\big\rangle $ for the state $\rho_{\lambda,\frac{T}{\lambda^{\gamma} } }$  in the momentum representation  encodes   the limiting spatial diffusive behavior    in the region near the diagonal $p_{1}=p_{2}$; see the proof of Thm.~\ref{ThmMain} below.  However, the state $\rho_{\lambda,\frac{T}{\lambda^{\gamma} } }$ has nearly the same integral kernel in the momentum representation $\big\langle p_{1}  \big|\rho_{\lambda,\frac{T}{\lambda^{\gamma} } } \big|  p_{2}\big\rangle $ and the extended-zone scheme representation ${   }_{\scriptscriptstyle{Q}}\big\langle p_{1}  \big|\rho_{\lambda,\frac{T}{\lambda^{\gamma} } } \big|  p_{2}\big\rangle_{\scriptscriptstyle{Q}}  $.  Propagation for the position variable can effectively be treated as if it were generated by the  extended-zone scheme momentum rather than the standard momentum.

\item  \textbf{Adiabatic approximation:} The evolution of the density matrix $\rho_{\lambda,t }$  over the time interval $r\in [0,\frac{T}{\lambda^{\gamma} }]$  has an approximate decomposition in the extended-zone scheme representation that emerges for $\lambda\ll 1$.  The limiting decomposition is such that  the off-diagonal functions $\big[\rho_{\lambda,t} \big]^{(k)}_{\scriptscriptstyle{Q}}\in L^{1}(\R)$, $k\in \R$  that are formally defined through
 \begin{align}\label{Jabber}
   \big[\rho \big]^{(k)}_{\scriptscriptstyle{Q}}(p):={   }_{\scriptscriptstyle{Q}}\big\langle p- k  \big|\,\rho\, \big|  p+ k \big\rangle_{\scriptscriptstyle{Q}} 
\end{align}
are approximately given by $\Phi_{\lambda,t}^{(k)}\big[\check{\rho}_{\lambda}  \big]^{(k)}_{\scriptscriptstyle{Q}}$ for a contractive semigroup of maps $\Phi_{\lambda,t}^{(k)}:L^{1}(\R)$.  The dynamics along the off-diagonal fibers of the density matrix are thus approximately autonomous.  This feature is a consequence of the relatively fast time scale of the Hamiltonian dynamics compared to the noise in the adiabatic regime.

\item  \textbf{Classical approximation:}  Fiber decompositions such as in the idealized form remarked upon in (II) are characteristic of translation covariant dynamics, i.e., the rate of momentum kicks does not depend on the position of the particle.  The Markovian process $(K_{t},Y_{t})$ discussed in Sect.~\ref{SecPreClassical} has a translation covariant law, so an analogous decomposition applies.   Let  $\mathcal{P}_{\lambda, t}\in L^{1}(\R^{2})$ be the joint density for the random variable $(K_{t},Y_{t})$.  For each $k\in \R$ there is a contractive semigroup  $\Upsilon_{\lambda, t}^{(k)}:L^{1}(\R)$ such that 
$ \Upsilon_{\lambda, t}^{(k)}\mathcal{P}_{\lambda,0}^{(k)}=\mathcal{P}_{\lambda,t}^{(k)}$ for all times $t\in \R^{+}$, where  $\mathcal{P}_{\lambda,t}^{(k)}: = \int_{\R}dx\, \mathcal{P}_{\lambda, t}(x,p)e^{\textup{i}2xk}$.  The semigroups $\Phi_{\lambda,t}^{(k)}$ and $\Upsilon_{\lambda, t}^{(k)}$ are close for small $k$, and   through this connection, the problem reduces to the classical case.

\end{enumerate}

Let  $\Phi_{\lambda,t}^{(k)}:L^{1}(\R)$ be the semigroup with generator $\mathcal{L}_{\lambda, k}$ that acts on elements  $f\in \mathcal{T}:= \big\{ f\in L^{1}(\R)\,\big|\,\int_{\R}dp\, |p|\, |f(p)|<\infty \big\}    $ as
\begin{align*}
\big(\mathcal{L}_{\lambda, k}f\big)(p)= &-\frac{\textup{i}}{\lambda^{\varrho}} \big(E\big(p- k\big)- E\big(p+ k \big)   \big)f(p) \\   & -\mathcal{R}f(p)+\int_{\R}dp'\,J_{k}(p,p')f(p')    ,
\end{align*}
where the kernel $ J_{k}$ is defined as
\begin{align}\label{OffJays}   J_{k}(p,p'):= \sum_{n\in \Z}j(p-p'-n)\kappa_{p-p'-n }\big(p'- k,\,n\big) \overline{\kappa}_{p-p'-n }\big(p'+k,\,n\big)   
\end{align}
for $\kappa_{v }(p,\,n)$ given by~(\ref{TheKappas}).  The operator $\mathcal{L}_{\lambda, k}$ is closed when assigned the domain $\mathcal{T}$, and the semigroup $\Phi_{\lambda,t}^{(k)}=e^{t\mathcal{L}_{\lambda, k} }$ is contractive. 
The semigroup of maps $\Upsilon_{\lambda, t}^{(k)}:L^{1}(\R)$ is defined to have generator $\mathcal{L}_{\lambda, k}'$ with domain $\mathcal{T}$ and operation given by 
\begin{align}\label{FlowerBasket}
\big(\mathcal{L}_{\lambda, k}' f\big)(p)=\frac{\textup{i}}{\lambda^{\varrho}} 4kp f(p)+\int_{\R}dp'\,\Big(J(p,p')f(p')-J(p',p)f(p) \Big)  .
\end{align}

For $k\in \R$ and $\rho \in \mathcal{B}_{1}(\mathcal{H})$, let $[\rho]^{(k)}\in L^{1}(\R)$  be defined  analogously to $[\rho]^{(k)}_{\scriptscriptstyle{Q}}$ as
$$ [\rho]^{(k)}(p):= \big\langle p-k  \big|\,\rho\,\big|  p+k\big\rangle  .   $$
The functions $[\rho]^{(k)}$ and $[\rho]^{(k)}_{\scriptscriptstyle{Q}}$ are extractions of the off-diagonals from the kernels for $\rho$ in the momentum and extended-zone scheme momentum representations, respectively, and  these objects are discussed more rigorously in Sect.~\ref{SecBlochFiber}.  The descriptions (I), (II), and (III) correspond respectively to   Lem.~\ref{StanToQuasi}, Thm~\ref{ThmSemiClassical}, and Lem.~\ref{SemiToFull} below. 

\begin{lemma}\label{StanToQuasi}
   For  $T\in \R^{+}$, $\gamma\in (1, 2)$, and  $\iota= \frac{2-\gamma}{3}$, there is a $C>0$ such that for all $\lambda<1$ and $|k|\leq  \lambda^{\iota}\mathbf{p}$,
$$\left\|   \big[\rho_{\lambda,\frac{T}{\lambda^{\gamma} } }  \big]^{(k)}-  \big[\rho_{\lambda,\frac{T}{\lambda^{\gamma} } } \big]^{(k)}_{\scriptscriptstyle{Q}}\right\|_{1}\leq C\lambda^{\iota}.    $$

\end{lemma}

\begin{theorem}\label{ThmSemiClassical}
For $T\in \R^{+}$,  $\gamma\in (1,2)$, and $\iota=\min\big(\varrho-\gamma,\frac{2-\gamma}{2}\big)$,     there is a $C>0$ such that for all  $\lambda<1$ and $|k|\leq \frac{1}{2}\lambda^{2}$,
$$
\left\|  \big[\rho_{\lambda, \frac{T}{\lambda^{\gamma}} } \big]^{(k)}_{\scriptscriptstyle{Q}} -\Phi_{\lambda,\frac{T}{\lambda^{\gamma}}}^{(k)}\big[\check{\rho}_{\lambda} \big]^{(k)}_{\scriptscriptstyle{Q}} \right\|_{1} \leq C\lambda^{ \iota}.    $$

\end{theorem}

\begin{lemma}\label{SemiToFull}
  There is a $C>0$ such that  for all $\lambda<1$, $t\in \R^{+}$, and $|k|\leq \lambda^{\frac{\gamma}{2}+\frac{5}{4}+\varrho }$,
$$ \left\| \Phi_{\lambda,t }^{(k)}\big[\check{\rho}_{\lambda} \big]^{(k)}_{\scriptscriptstyle{Q}}- \mathcal{P}_{\lambda, t}^{(k)}\right\|_{1} \leq C\big(\lambda^{\frac{1}{2}} + t\lambda^{\frac{\gamma}{2}+1} \big) .  $$

\end{lemma}

\begin{theorem}\label{ThmClassical}
For $T\in \R^{+}$, the processes $Y_{\lambda}=\big(  \lambda^{\frac{\gamma+2\varrho +3}{2} }Y_{\frac{s}{\lambda^{\gamma} } },\, s\in [0,T]\big)$ converge in law as $\lambda\searrow 0$ to a Brownian motion with diffusion constant $\vartheta $.  In particular,   the characteristic functions of the densities $\mathcal{P}_{\lambda, \frac{T}{\lambda^{\gamma}} }\in L^{1}(\R^2 )$ satisfy the   pointwise convergence  as $\lambda\searrow 0$ given by
$$\int_{\R^{2}} dx dp\, \mathcal{P}_{\lambda, \frac{T}{\lambda^{\gamma}}}(x,p)e^{\textup{i}xu   } \longrightarrow e^{-\frac{T\vartheta}{2}k^{2}}$$  
for $u:=\frac{1}{2}\lambda^{\frac{\gamma+2\varrho +3}{2} }k$.  The convergence of the processes $ Y_{\lambda}$ is with respect to the uniform metric on paths.  

\end{theorem}

\vspace{.4cm}

\begin{proof}[Proof of Thm.~\ref{ThmMain}]
 Let $\varphi_{\lambda, T}$ be the characteristic function for the probability measure $\mu_{\lambda,T}$.  To show that  $\mu_{\lambda,T}$ converges in distribution to a   zero mean Gaussian with variance $T\vartheta $, it is sufficient to prove that $\varphi_{\lambda, T}(k)$ converges pointwise to $e^{- \frac{T\vartheta}{2}  k^{2} }  $. The characteristic function $\varphi_{\lambda, T}(k)$ is equal to the following:
\begin{align}\label{Alt0}
\varphi_{\lambda, T}(k) := &\int_{\R} d\mu_{\lambda, T}(x)\,e^{\textup{i}k x}\nonumber \\ =& \int_{\R} dx\, D_{\lambda,  \frac{T }{\lambda^{\gamma}}  }(x)\,e^{\textup{i}2u x}\nonumber   \\
=&
\Tr\big[\rho_{\lambda, \frac{T }{\lambda^{\gamma}}}e^{\textup{i}2uX }\big]\nonumber    =\int_{\R} dp\,\big\langle p\big|\rho_{\lambda, \frac{T }{\lambda^{\gamma}} }\big|p+2 u\big\rangle  \\ =& \int_{\R}dp\, [\rho_{\lambda,\frac{T }{\lambda^{\gamma}}}]^{(u)}(p), 
\end{align}
where $u:=\frac{1}{2} \lambda^{ \frac{\gamma+2\varrho+3}{2}  }k $.   The fourth equality uses the formal kernel relation $ \big\langle p\big|\rho_{\lambda, \frac{T }{\lambda^{\gamma}} } e^{\textup{i}2uX }\big|p\big\rangle=    \big\langle p\big|\rho_{\lambda, \frac{T }{\lambda^{\gamma}} }\big|p+2 u\big\rangle $; more rigorously, the equality between the  first expression on the third line  of~(\ref{Alt0}) and the expression on the fourth line will follow from  Part (2) of Prop.~\ref{PropFiberII}.   The results of  Lem.~\ref{StanToQuasi}, Thm.~\ref{ThmSemiClassical}, and Lem.~\ref{SemiToFull} yield that the telescoping differences
$$
   \big[\rho_{\lambda,\frac{T}{\lambda^{\gamma} } }  \big]^{(u)}-  \big[\rho_{\lambda,\frac{T}{\lambda^{\gamma} } } \big]^{(u)}_{\scriptscriptstyle{Q}},\hspace{.7cm}  \big[\rho_{\lambda, \frac{T}{\lambda^{\gamma}} } \big]^{(u)}_{\scriptscriptstyle{Q}} -\Phi_{\lambda,\frac{T}{\lambda^{\gamma}}}^{(u)}\big[\check{\rho}_{\lambda}\big]^{(u)}_{\scriptscriptstyle{Q}} , \hspace{.7cm}\Phi_{\lambda,\frac{T}{\lambda^{\gamma} } }^{(u)}\big[\check{\rho}_{\lambda}\big]^{(u)}_{\scriptscriptstyle{Q}}- \mathcal{P}_{\lambda, \frac{T}{\lambda^{\gamma} }}^{(u)}  $$
decay in the $L^{1}$-norm on the order $\mathit{O}(\lambda^{\iota})$ for $\lambda\ll 1$ and small enough choice of the exponent $\iota>0$.  Thus,
  $$\Big|\int_{\R}dp\, [\rho_{\lambda,  \frac{T}{\lambda^{\gamma} }}]^{(u)}(p)-\int_{\R}dp\,  \mathcal{P}_{\lambda, \frac{T}{\lambda^{\gamma} }}^{(u)}(p)\Big|=\mathit{O}(\lambda^{\iota}).      $$
Finally, by Thm.~\ref{ThmClassical} there is convergence as $\lambda\searrow 0$ for each $k\in \R$:
 $$\int_{\R}dp\,  \mathcal{P}_{\lambda, \frac{T}{\lambda^{\gamma} }}^{(u)}(p)=\int_{\R^{2}} dx dp \,\mathcal{P}_{\lambda, \frac{T}{\lambda^{\gamma}}}(x,p)e^{\textup{i}xu  } \longrightarrow e^{-\frac{T\vartheta}{2}k^{2}}.$$

\end{proof}

\subsection{The ballistic cases}\label{SecBall}

Theorem~\ref{ThmEcho} characterizes the test particle's motion without the Dirac comb or the random kicking from the gas.  In those situations, the test particle retains an effectively ballistic motion with speed $\frac{2\mathbf{p}}{\lambda^{\varrho}}$, and the position distribution at time $\frac{T}{\lambda^{\gamma}}$ is centered around the location $ \frac{T}{\lambda^{\gamma}} \frac{2\mathbf{p}}{\lambda^{\varrho}}= 2\mathbf{p}_{0}T\lambda^{-\gamma-\varrho-1 }   $.  In particular, diffractive effects do not appear without the aid of random kicks from the environment.

\begin{theorem}\label{ThmEcho}
Let $\gamma$, $\varrho$,  $D_{\lambda,t}$, and $\rho_{\lambda, t}$ be as in Thm.~\ref{ThmMain} except with $\alpha=0$ (no Dirac comb) or $\Psi=0$ (no noise).  Define the measure $\mu_{\lambda,T}$ to have the density: 
$$\frac{d\mu_{\lambda,T}}{dx}(x)=\lambda^{-\gamma-\varrho-1 }D_{\lambda,\frac{T}{\lambda^{\gamma} }}(\lambda^{-\gamma-\varrho-1 } x   ).$$  
 For each $T\in \R^{+}$, the measures $\mu_{\lambda,T}$ converge in law as $\lambda\searrow 0$ to a $\delta$-distribution at $2\mathbf{p}_{0}T$.
\end{theorem}

It is natural to apply different techniques to prove Thm.~\ref{ThmEcho} for the cases in which either the Dirac comb or the noise is set to zero.  The scenario without the comb is mathematically trivial since the dynamics has a well-known closed form when viewed through the quantum characteristic function representation; see Lem.~\ref{LemExplicit}.  When only the Dirac comb is present, the proof follows by a reduction of the strategy  applied in the proof of Thm.~\ref{ThmMain}.

For $k\in \R$ and $t,\lambda \in \R^{+}$, define the maps $U_{\lambda,t}^{(k)},\widetilde{U}_{\lambda,t}^{(k)}:L^{1}(\R)$ to act as multiplication by the functions
\begin{align}\label{Jigsaw}
U_{\lambda,t}^{(k)}(p):=e^{-\frac{\textup{i}t}{\lambda^{\varrho}}\big(E(p- k )-E(p+ k ) \big)}\quad \text{and} \quad  \widetilde{U}_{\lambda,t}^{(k)}(p):=e^{\frac{\textup{i}t}{\lambda^{\varrho}} 4pk}.
\end{align}

  The following lemma holds for a more general class of quantum dynamical  semigroups~\cite{Holevo} satisfying a symmetry known as \textit{Galilean covariance}.
\begin{lemma}\label{LemExplicit}
Let $\rho_{\lambda,t}$ satisfy the Lindblad equation~(\ref{RETheModel}) with $\alpha=0$ and beginning from a density matrix $\rho\in \mathcal{B}_{1}( \mathcal{H} )$.  The quantum characteristic function for $\rho_{\lambda,t}$ has the closed form
\begin{align*}
\Tr\big[  \rho_{\lambda,t}e^{\textup{i}uX+\textup{i}qP}  \big]=e^{\int_{0}^{t}dr\,\big( \widehat{\varphi}(q+\frac{2u}{\lambda^{\varrho}}(t-r))- \widehat{\varphi}(0)  \big)}\Tr\big[  \rho e^{\textup{i}uX+\textup{i}(q+\frac{2u t}{\lambda^{\varrho}}  )P}  \big],
\end{align*}
where $\widehat{\varphi}(q):=\int_{\R}dv\, j(v) e^{\textup{i}qv} $.

\end{lemma}

\begin{proof}[Proof of Thm.~\ref{ThmEcho}]
 Let  $u:= \frac{1}{2}\lambda^{ \gamma+\varrho+1 }k $ and denote the characteristic function for the  measure $\mu_{\lambda,T}$ by $\varphi_{\lambda, T}$.  By the proof of Thm.~\ref{ThmMain}, the function $\varphi_{\lambda, T}$ can be written in the forms
\begin{align}
\varphi_{\lambda, T}(k) =  &\Tr\big[\rho_{\lambda, \frac{T }{\lambda^{\gamma}}}e^{\textup{i}2uX }\big] \label{ForNoComb} \\ = &\int_{\R}dp\, \big[\rho_{\lambda,\frac{T }{\lambda^{\gamma}}}\big]^{(u)}(p). \label{ForNoNoise} 
\end{align}
The cases without the Dirac comb and without the noise are handled in (i) and (ii) below, respectively.  It is sufficient in each case to show the pointwise convergence of $\varphi_{\lambda, T}(k)$ to the value $ e^{\textup{i}2T\mathbf{p}_{0}k   } $ as $\lambda\searrow 0$.
\vspace{.4cm}

\noindent (i).\hspace{.2cm} By~(\ref{ForNoComb}) and the formula for the quantum characteristic function from Lem.~\ref{LemExplicit}, I have the first equality below:
\begin{align}\label{Harbinger}
 \varphi_{\lambda, T}(k)=&e^{\int_{0}^{\frac{T}{\lambda^{\gamma}} }dr\,\big( \widehat{\varphi}(\frac{2u}{\lambda^{\varrho}}(\frac{T}{\lambda^{\gamma}}-r))- \widehat{\varphi}(0)  \big)}\Tr\Big[  \check{\rho}_{\lambda} e^{\textup{i}uX+\textup{i}\frac{2uT}{\lambda^{\gamma+\varrho}}P}  \Big]  \nonumber  \\  = & \Tr\Big[  \check{\rho}_{\lambda} e^{\textup{i}uX+\textup{i}\frac{2uT}{\lambda^{\gamma+\varrho}}P}  \Big]+\mathit{O}(\lambda)  \nonumber \\  
=& \Big\langle \frak{h}_{0}\Big| e^{\textup{i}uX+\textup{i}\frac{2uT}{\lambda^{\gamma+\varrho}}P}   \frak{h}_{0}    \Big\rangle  e^{\textup{i}\frac{2uT}{\lambda^{\gamma+\varrho}}\mathbf{p} }+\mathit{O}(\lambda)  \nonumber \\ 
=&    e^{\textup{i}\frac{2uT}{\lambda^{\gamma+\varrho}}\mathbf{p} }+\mathit{O}(\lambda)=  e^{\textup{i}2\mathbf{p}_{0}Tk   }+\mathit{O}(\lambda).
\end{align}
In the above $\widehat{\varphi}$ is the  Fourier transform  of $j:\R\rightarrow \R^{+}$.  The second equality holds since  $\frac{2 uT}{\lambda^{\gamma+\varrho}}=2kT \lambda \ll 1$, the first derivative of $\widehat{\varphi}$ is zero at zero, and the second  derivative  of $\widehat{\varphi}$ is bounded by $\sigma=\int_{\R}dv\, j(v) v^{2}$.   The third equality uses that $\check{\rho}_{\lambda}$ can be written as  $e^{\textup{i}\mathbf{p}X  }\big| \frak{h}_{0}\big\rangle \big\langle \frak{h}_{0}\big|   e^{-\textup{i}  \mathbf{p}X  } $ and Weyl intertwining relations.

\vspace{.4cm}

\noindent (ii). By~(\ref{ForNoNoise}) and the triangle inequality,
\begin{align}\label{Beckham}
\Big|\varphi_{\lambda, T}(k)-\int_{\R}dp\,\widetilde{U}_{\lambda,\frac{T}{\lambda^{\gamma}}}^{(u)}(p)[\check{\rho}_{\lambda}]^{(u)}(p)\Big|\leq & \Big\| [\rho_{\lambda,\frac{T }{\lambda^{\gamma}}}]^{(u)}-[\rho_{\lambda,\frac{T }{\lambda^{\gamma}}}]^{(u)}_{\scriptscriptstyle{Q}}     \Big\|_{1} +\Big\| \big[\rho_{\lambda,\frac{T }{\lambda^{\gamma}}}\big]^{(u)}_{\scriptscriptstyle{Q}} -\widetilde{U}_{\lambda,\frac{T}{\lambda^{\gamma}}}^{(u)}\big[\check{\rho}_{\lambda}\big]^{(u)}    \Big\|_{1} . 
\end{align}
 The difference $\big[\rho_{\lambda,\frac{T }{\lambda^{\gamma}}}\big]^{(u)}-\big[\rho_{\lambda,\frac{T }{\lambda^{\gamma}}}\big]^{(u)}_{\scriptscriptstyle{Q}}$  converges to zero in the $L^{1}$-norm as $\lambda\searrow 0$ by Lem.~\ref{StanToQuasi}.  The density matrix evolved to time $t\in \R^{+}$ is given by $\rho_{\lambda,t}=e^{-\frac{\textup{i}t}{\lambda^{\gamma}} H}\check{\rho}_{\lambda}e^{\frac{\textup{i}t}{\lambda^{\gamma}} H}$,  and consequently $[\rho_{\lambda,t}]^{(k')}_{\scriptscriptstyle{Q}}=U_{\lambda,t}^{(k')}[\check{\rho}_{\lambda}]^{(k')}_{\scriptscriptstyle{Q}}$ for $k'\in \R$.  
A much simplified version of the proof for Lem.~\ref{SemiToFull} shows that there is a $C>0$ such that for all $\lambda<1$, $t\in \R^{+}$, and $|k'|\leq \lambda^{\frac{\gamma}{2}+\frac{5}{4}+\varrho}$,
\begin{align}\label{Hoover}
\Big\| U_{\lambda,t}^{(k')}[\check{\rho}_{\lambda}]^{(k')}_{\scriptscriptstyle{Q}}-\widetilde{U}_{\lambda,t}^{(k')}[\check{\rho}_{\lambda}]^{(k')}\Big\|_{1}\leq C\big(\lambda^{\frac{1}{2}}+ t\lambda^{\frac{\gamma}{2}+1}\big).
\end{align}
Since $\gamma< 2$ applying~(\ref{Hoover}) for $t=\frac{T}{\lambda^{\gamma}}$ implies that  the term on the right side of~(\ref{Beckham}) converges to zero for small $\lambda$.  

Finally, we have the equalities
\begin{align*}
\int_{\R}dp\,\widetilde{U}_{\lambda,\frac{T}{\lambda^{\gamma}}}^{(u)}(p)[\check{\rho}_{\lambda}]^{(u)}(p)= & \Tr\Big[e^{-\textup{i}\frac{T}{\lambda^{\gamma+\varrho  }   } P^{2}  }  \check{\rho}_{\lambda} e^{\textup{i}\frac{T}{\lambda^{\gamma+\varrho  }   } P^{2}  }   e^{\textup{i}uX }  \Big] \\ =&\Tr\Big[ \check{\rho}_{\lambda} e^{\textup{i}uX+\textup{i}\frac{Tu}{\lambda^{\gamma+\varrho  }   }P }  \Big] \\ = & e^{\textup{i}2\mathbf{p}_{0}Tk   }+\mathit{O}(\lambda),
\end{align*}
where the last equality holds by the analysis in (i).  

\end{proof}

\section{Bloch functions and the fiber decompositions}\label{SecBlochFiber}

\subsection{Fiber decomposition for the Hamiltonian}\label{SecHamFiber}
The invariance of the Hamiltonian $H$ under spatial shifts by $2\pi$ is characterized by the commutation relation
\begin{align}\label{Commy}
e^{\textup{i}2\pi P}H=H e^{\textup{i}2\pi P}.
\end{align}
It follows that $H$ acts invariantly on the eigenspaces of $ e^{\textup{i}2\pi P}$, which is the foundation for Bloch theory~\cite{Reed}.   The Hilbert space $\mathcal{H}=L^{2}(\R)$ has a canonical tensor product decomposition $\mathcal{H}=L^{2}(\mathbb{T})\otimes L^{2}(\mathcal{I} ) $ in which an element $f\in \mathcal{H}$ is related to an $L^{2}$-function $ \widehat{f}:\mathbb{T}\rightarrow L^{2}(\mathcal{I}) $ through the partial Fourier transform
 $$\hspace{3cm}\widehat{f}_{\phi}(x)= \sum_{n\in \Z}e^{-\textup{i}2\pi \phi n}f(x+2\pi n),\hspace{1cm} x\in \mathcal{I},      $$      
where the argument $\phi \in \mathbb{T}$ of  $ \widehat{f}$ is placed as a subscript.  The commutation relation~(\ref{Commy}) implies that there are self-adjoint operators $H_{\phi}$, $\phi\in \mathbb{T}$ defined on dense domains of $L^{2}(\mathcal{I}) $  such that  $ (\widehat{Hf})_{\phi}= H_{\phi}\widehat{f}_{\phi}$.  The operators $H_{\phi}$ have the form 
$$H_{\phi}= -\big(\frac{d^{2}}{dx^{2}}\big)_{\phi}^{(\alpha)},$$
where $\big(\frac{d^{2}}{dx^{2}}\big)_{\phi}^{(\alpha)}$ is the self-adjoint extension of the second derivative over the domain $(-\pi,0)\cup (0,\pi)$ with the boundary conditions
\begin{eqnarray*}
  \alpha g(0)&=&\frac{dg}{dx}(0+)-\frac{dg}{dx}(0-),   \\
 g(-\pi)&=&e^{-\textup{i} 2\pi \phi   }  g(\pi ), \\
   \frac{dg}{dx}(-\pi )&=&e^{-\textup{i} 2\pi \phi   }  \frac{dg}{dx}(\pi )  .
   \end{eqnarray*}

\subsection{Bloch functions}

The eigenfunctions for the operators $H_{\phi}$, $\phi\in \mathbb{T}$ have closed forms in the case of the Dirac comb, which  are   Bloch functions $\widetilde{\psi}_{p}\in L^{2}(\mathcal{I})$ for $p=\phi \text{ mod } 1$  given by 
\begin{align}\label{BlochFunctions}
\widetilde{\psi}_{p}(x)= N_{p}^{-\frac{1}{2}}\left\{  \begin{array}{cc} \frac{e^{\textup{i}2\pi (\mathbf{q}(p)-p)  }-1    }{ e^{\textup{i}2\pi (\mathbf{q}(p)+p)}  -1 } e^{-\textup{i} x \mathbf{q}(p) }+e^{\textup{i}2\pi ( \mathbf{q}(p)-p) }  e^{\textup{i} x \mathbf{q}(p)  } &  -\pi \leq x\leq 0  ,  \\  \quad & \quad \\ \frac{e^{\textup{i}2\pi(\mathbf{q}(p)-p)  }-1     }{1-e^{-\textup{i}2\pi(\mathbf{q}(p)+p)} }e^{-\textup{i} x \mathbf{q}(p)  }+   e^{\textup{i}x  \mathbf{q}(p)  }     & 0\leq x<\pi ,   \end{array} \right.  
\end{align}
where $N_{p}>0$ is a normalization, and $  \mathbf{q}:\R\rightarrow \R $ is defined  in the Kr\"onig-Penney relation~(\ref{KronigPenney}).   I  denote the Bloch functions for the momentum operator by $\psi_{p}(x):=(2\pi)^{-\frac{1}{2}}e^{\textup{i}x p}$.   When $|p|\gg 1 $, then  $\mathbf{q}(p)\approx p$ and the Bloch function $\widetilde{\psi}_{p}$ is approximately equal to $\psi_{p}$ except for $p$ near an element of the  lattice $\frac{1}{2}\Z$ (see Lem.~\ref{Ticks} below).  Under the usual  conventions, the eigenvalues $E_{N,\phi}$ and corresponding eigenvectors $ \widetilde{\psi}_{N,\phi}$ for the Hamiltonian $H_{\phi}$ are labeled progressively $E_{ N+1,\phi}\geq E_{ N,\phi}$ by the \textit{band index} $N\in \mathbb{N}$ and the \textit{quasimomentum} $\phi\in \mathbb{T}$.  For $\phi \neq -\frac{1 }{2},0 $, the extended-zone scheme parameter $p\in \R$ is determined by the pair $N,\phi$ through the relations 
\begin{align*}
p=\phi\,\textup{mod}\,1, \hspace{1cm}\text{and}\hspace{1cm}   N= \left\{  \begin{array}{cc}  \frac{1}{2}|p-\phi|    &  S(p)=S(\phi) ,  \\  \quad & \quad \\   \frac{1}{2} |p-\phi|-1   &  S(p)=-S(\phi),    \end{array} \right.  
\end{align*}
where $S:\R\rightarrow \{\pm 1\}$ is the sign function.  The assignment convention for the measure zero set  $\phi \in \{-\frac{1 }{2},0\}$ is not important for the purpose of this article.


    Let $\Theta:\R\rightarrow [-\frac{1}{4 },\frac{1}{4 })$ and $\mathbf{n}:\R\rightarrow \Z$ be defined such that
\begin{align}\label{Barrack}
\Theta(p)=p\,\text{mod}\,\frac{1}{2}, \quad \quad \mathbf{n}(p)=2\big(p-\Theta(p)\big).\quad 
\end{align}
Also, define $\beta:\R\rightarrow \R$ such that $\beta(p):=\frac{1}{2}\mathbf{n}(p)\Theta(p)$.  Lemma~\ref{Ticks} bounds the difference between the Bloch functions $\widetilde{\psi}_{p}$,  $\psi_{p}$, and the proof is contained in the proof of~\cite[Prop.4.2]{Dispersion}.
\begin{lemma}\label{Ticks}
There is a $C>0$ such that for all $p\in \R$,
$$  \big\|\psi_{p}-\widetilde{\psi}_{p}\big\|_{2}\leq \frac{C}{1+ |\beta(p)|  } .  $$ 
\end{lemma}

\subsection{Dissecting a density matrix}\label{SecDissect}

There are various substructures for a density matrix $\rho\in \mathcal{B}_{1}(\mathcal{H})$ that are useful to identify and define rigorously.  Recall that an element $f\in \mathcal{H}$ can be identified with an element $\widehat{f}\in L^{2}\big(\mathbb{T},L^{2}(\mathcal{I} )\big)$ through the tensor product decomposition $\mathcal{H}=L^{2}(\mathbb{T})\otimes L^{2}(\mathcal{I} ) $.  For an Hilbert-Schmidt operator $\rho\in \mathcal{B}_{2}(\mathcal{H})  $, there are operator coefficients $\ell_{\phi}^{(\kappa)}(\rho)\in \mathcal{B}_{2}\big(L^{2}(\mathcal{I} )\big) $  defined for a.e. $(\phi,\kappa)\in \mathbb{T}\times [-\frac{1}{4}, \frac{1}{4})  $ through the  relation 
\begin{align}\label{Kernel}
 \langle f |\rho g\rangle =  \int_{\mathbb{T}\times [-\frac{1}{4}, \frac{1}{4})}d\phi d\kappa \, \big\langle \widehat{f}_{\phi- \kappa }\big|\ell_{\phi}^{(\kappa)}(\rho)  \widehat{g}_{\phi+ \kappa  }\big\rangle     
 \end{align}
for all $f,g\in \mathcal{H}$.   The operators $\ell_{\phi}^{(\kappa)}(\rho):L^{2}(\mathcal{I} )$ are merely the blocks associated with the Hilbert space tensor product $L^{2}(\mathbb{T})\otimes L^{2}(\mathcal{I} ) $.  When $\rho $ is trace class, the operator  coefficients  $\ell_{\phi}^{(\kappa)}(\rho)$ can be taken to be in $ \mathcal{B}_{1}\big(L^{2}(\mathcal{I} )\big) $ and are determined in a stricter sense than a.e. $(\phi,\kappa)\in \mathbb{T}\times  [-\frac{1}{4}, \frac{1}{4}) $: For each $\kappa\in  [-\frac{1}{4}, \frac{1}{4})$, the operators $\ell_{\phi}^{(\kappa)}(\rho)\in \mathcal{B}_{1}\big(L^{2}(\mathcal{I}  )\big) $ are defined a.e. $\phi\in \mathbb{T}$.   In fact, for all $\kappa\in [-\frac{1}{4}, \frac{1}{4})$, the function $\ell^{(\kappa)}(\rho)$ that sends $\phi\in \mathbb{T}$ to $\ell^{(\kappa)}_{\phi}(\rho)$ can be regarded as an element in $L^{1}\left(\mathbb{T},\,\mathcal{B}_{1}(L^2(\mathcal{I}))\right)$.

 The function $\ell^{(\kappa)}(\rho):\mathbb{T}\rightarrow  \mathcal{B}_{1}\big(L^2(\mathcal{I})\big)$ is defined in the lemma below though the  Banach algebra $\mathcal{A}_{\mathbb{T}}\subset \mathcal{B}(\mathcal{H})$  of operators that commute with $e^{\textup{i}2\pi P}$.  The algebra  $\mathcal{A}_{\mathbb{T}}$ is isometrically isomorphic to $L^{\infty}\big(\mathbb{T}, \mathcal{B}(L^2(\mathcal{I}) )\big)$.  Elements  $G\in  \mathcal{A}_{\mathbb{T}}$ are identified with elements $\widetilde{G} \in L^{\infty}\big(\mathbb{T}, \mathcal{B}(L^2(\mathcal{I}))\big)$ through the equality
$$ \hspace{3cm}  \langle f| G g\rangle = \int_{\mathbb{T}}d\phi  \,\big\langle \widehat{f}_{\phi}\big| \widetilde{G}_{\phi}\,   \widehat{g}_{\phi} \big\rangle ,  \hspace{2cm} f,g\in \mathcal{H} .               $$

\begin{lemma}\label{Israelistine}
Let $\rho\in \mathcal{B}(\mathcal{H})$ and $\kappa\in [-\frac{1}{4}, \frac{1}{4})$.  There is a unique function $\ell^{(\kappa)}(\rho)\in L^{1}\left(\mathbb{T},\,\mathcal{B}_{1}\big(L^2(\mathcal{I})\big)\right)$ satisfying that for all $G\in \mathcal{A}_{\mathbb{T}}$, 
$$\Tr\big[\rho e^{\textup{i}\kappa X}G e^{\textup{i}\kappa X} \big]=\int_{\mathbb{T}}d\phi \, \Tr\big[ \ell^{(\kappa)}_{\phi}(\rho) e^{\textup{i}\kappa X_{\mathbb{T}}}\widetilde{G}_{\phi} e^{\textup{i}\kappa X_{\mathbb{T}}}  \big],     $$
where $X_{\mathbb{T}}\in \mathcal{B}\big(L^{2}(\mathcal{I})\big)$ acts as the multiplication operator $(X_{\mathbb{T}}f)=xf(x)$ for $f\in L^2(\mathcal{I})$.  Moreover, the norm for $\ell^{(\kappa)}(\rho)$ has the bound $\big\| \ell^{(\kappa)}(\rho) \big\|_{1}\leq \|\rho\|_{\mathbf{1}}   $.

\end{lemma}

\begin{proof}
The equality follows by expanding $\rho$ in terms of its singular value decomposition and using  that for $G \in \mathcal{A}_{\mathbb{T}}$ and  $f,g\in \mathcal{H}$, then $ \langle f| G g\rangle = \int_{\mathbb{T}}d\phi \, \big\langle \widehat{f}_{\phi}\big| \widetilde{G}_{\phi}\,   \widehat{g}_{\phi} \big\rangle $  by the definition of  $\widetilde{G}$.  The following relations bound the integral norm of $ \ell^{(\kappa)}(\rho)$:
\begin{align*}
\big\|  \ell^{(\kappa)}(\rho)  \big\|_{1}=\int_{\mathbb{T}}d\phi \, \big\|\ell^{(\kappa)}_{\phi}(\rho) \big\|_{\mathbf{1}} &=\sup_{\substack{\widetilde{G}\in L^{\infty}(\mathbb{T},\mathcal{B}(L^2(\mathcal{I}))), \\ \| \widetilde{G}   \|_{\infty}=1 } }\int_{\mathbb{T}}d\phi \, \Tr\big[ \ell^{(\kappa)}_{\phi}(\rho) e^{\textup{i}\kappa X_{\mathbb{T}} }\widetilde{G}_{\phi} e^{\textup{i}\kappa X_{\mathbb{T}} } \big]\\ &= \sup_{\substack{G\in \mathcal{A}_{\mathbb{T}}, \\ \| G   \|=1 } }\Tr\big[\rho e^{\textup{i}\kappa X}G e^{\textup{i}\kappa X} \big]\\ & \leq \|\rho\|_{\mathbf{1}},
\end{align*}
where the third equality above holds by the definition of $\widetilde{G}$.  The supremum on the first line is obtained as a maximum with $\widetilde{G}_{\phi}=e^{-\textup{i}\kappa X_{\mathbb{T}} }U_{\rho,\kappa,\phi}e^{-\textup{i}\kappa X_{\mathbb{T}} }$  for the unitary $U_{\rho,\kappa,\phi}\in \mathcal{B}\big(L^{2}(\mathcal{I})\big) $ in the polar decomposition of $ \ell^{(\kappa)}_{\phi}(\rho)  $, i.e.,   $\ell^{(\kappa)}_{\phi}(\rho)=U_{\rho,\kappa,\phi} |\ell^{(\kappa)}_{\phi}(\rho) | $.

\end{proof}

Recall that  $[\rho ]^{(k)}_{\scriptscriptstyle{Q}}$ and $[\rho ]^{(k)}$ are formally defined to be functions of $p\in \R$ given by $\text{ }_{\scriptscriptstyle{Q}}\langle p-k|\rho| p+ k\rangle_{\scriptscriptstyle{Q}}$ and $\langle p- k|\rho| p+ k\rangle$, respectively.  I interpret the  mathematical definitions for expressions involving kets as referring to analogous expressions formulated in terms of Bloch functions and the tensor product decomposition $\mathcal{H}=L^{2}(\mathbb{T})\otimes L^{2}(\mathcal{I} ) $; for instance,
\begin{align*}\text{ }_{\scriptscriptstyle{Q}}\big\langle p- k\big|\rho \big|\, p+k\big\rangle_{\scriptscriptstyle{Q}}:= &\big\langle \widetilde{\psi}_{p- k}\big| \ell^{(\kappa)}_{\phi}(\rho)\big| \widetilde{\psi}_{p+  k}\big\rangle \quad \text{and}\\ \quad \quad  \big\langle p-k\big|\rho \big|\, p+k\big\rangle := &\big\langle \psi_{p-  k}\big| \ell^{(\kappa)}_{\phi}(\rho)\big| \psi_{p+  k}\big\rangle 
\end{align*}
for $\phi \in \mathbb{T}$ and $\kappa\in [-\frac{1}{4},\frac{1}{4})$  equal, respectively, to $p \in \R$  modulo    $ 1$   and $k\in \R$ modulo $\frac{1}{2}$.  
Finally, I also define $\langle \rho \rangle^{(\kappa)}\in L^{1}(\mathbb{T})$ for $\kappa\in [-\frac{1}{4},\frac{1}{4})$ and $\rho\in \mathcal{B}_{1}(\mathcal{H})$ such that 
 $\langle \rho \rangle_{\phi}^{(\kappa)}: =\Tr[\ell^{(\kappa)}_{\phi}(\rho)] $. 
In future, the expressions   $\ell^{(k)}_{p}(\rho)$ and $\langle \rho \rangle_{p}^{(k)} $ for $k,p\in \R$ should be understood as  $\ell^{(\kappa)}_{\phi}(\rho)$ and $\langle \rho \rangle_{\phi}^{(\kappa)} $ for $k,p$ related to $\kappa,\phi$ as before.

The inequalities in Prop.~\ref{MiscFiber} and~\ref{PropFiberII} are all consequence of the Cauchy-Schwarz inequality. 

\begin{proposition}\label{MiscFiber}

 Let $\rho \in \mathcal{B}_{1}(\mathcal{H})$, and define  $|\rho|:=\sqrt{\rho^{*}\rho}$ and $|\rho|_{*}:=\sqrt{\rho \rho^{*}}$.

\begin{enumerate}
 \item Let $f_{\phi},g_{\phi}\in L^{2}(\mathcal{I})$ be square-integrable functions of the parameter $\phi\in \mathbb{T}$.  For each $\kappa \in  [-\frac{1}{4}, \frac{1}{4})$, the following inequality  holds for a.e. $\phi\in \mathbb{T}$:
\begin{align}\label{KungFu}
\big|\big\langle f_{\phi} \big| \ell_{\phi}^{(\kappa)} (\rho)g_{\phi}\big\rangle \big| \leq \big\langle f_{\phi} \big| \ell_{\phi- \kappa }^{(0)}\big(|\rho|_{*}\big)  f_{\phi}\big\rangle^{\frac{1}{2}}\big\langle g_{\phi} \big| \ell_{\phi+ \kappa }^{(0)}\big(|\rho|\big)  g_{\phi}\big\rangle^{\frac{1}{2}}  .
\end{align}

\item  For any $\kappa\in [-\frac{1}{4}, \frac{1}{4})$ and a.e.  $\phi\in \mathbb{T}$,
 $$\big|\langle \rho \rangle_{\phi}^{(\kappa)}\big|\leq \big(\big\langle |\rho|_{*} \big\rangle_{\phi- \kappa }^{(0)}\big)^{\frac{1}{2}} \big(\big\langle |\rho| \big\rangle_{\phi+\kappa }^{(0)}\big)^{\frac{1}{2}}.$$

\item  For any $k\in \mathbb{R}$ and a.e.  $p\in \mathbb{R}$,
$$\big| [\rho ]^{(k)}(p) \big|\leq  \Big( \big[|\rho|_{*} \big]^{(0)}\big(p-k\big) \Big)^{\frac{1}{2}} \Big( \big[|\rho| \big]^{(0)}\big(p+ k\big) \Big)^{\frac{1}{2}}  .  $$
The analogous equality holds with the $[\rho ]^{(k)}$'s replaced by the $[\rho ]^{(k)}_{\scriptscriptstyle{Q}}$'s.

\end{enumerate}

\end{proposition}

\begin{proof}\text{   }\\
\noindent Part (1):  For a measurable set $A\subset \mathbb{T}$, let $G^{(A)},G^{(A),\prime},G^{(A),\prime \prime}\in \mathcal{A}_{\mathbb{T}}$ have respective corresponding  elements in $L^{\infty}\big(\mathbb{T},\mathcal{B}(L^2(\mathcal{I}))\big)$ given by 
\begin{eqnarray*}
\widetilde{G}^{(A)}_{\phi}&:=&\overline{s}_{\phi }1_{A}(\phi)e^{-\textup{i} \kappa  X_{\mathbb{T}} } |g_{\phi}\rangle \langle f_{\phi}| e^{-\textup{i} \kappa  X_{\mathbb{T}} } ,
\\    \widetilde{G}^{(A), \prime }_{\phi}&:=& s_{\phi }  1_{A}(\phi)|\psi_{0}\rangle \langle g_{\phi}| e^{\textup{i}\kappa X_{\mathbb{T}} } ,
\\
\widetilde{G}^{(A),\prime \prime}_{\phi}&:=&1_{A}(\phi)|\psi_{0}\rangle \langle f_{\phi}|  e^{-\textup{i}\kappa X_{\mathbb{T}} } 
,
\end{eqnarray*}
where $s_{\phi }:=\frac{ \langle f_{\phi}| \ell^{(\kappa)}_{\phi}(\rho)|g_{\phi}\rangle  }{  | \langle f_{\phi}| \ell^{(\kappa)}_{\phi}(\rho)|g_{\phi}\rangle |  }    $ and $X_{\mathbb{T}}\in \mathcal{B}\big(L^{2}(\mathcal{I})\big)$ is defined as in the statement of  Lem.~\ref{Israelistine}.  The choice of the vector  $ \psi_{0}\in L^2(\mathcal{I})  $ is arbitrary, and I will only use that $ \langle  \psi_{0}\,| \psi_{0}\rangle =1$. Notice that for all $\phi\in \mathbb{T}$ $\widetilde{G}^{(A)}_{\phi}=  ( \widetilde{G}^{(A), \prime }_{\phi} )^{*}\widetilde{G}^{(A),\prime \prime}_{\phi}$, and thus $G^{(A)}=(G^{(A),\prime})^{*}G^{(A),\prime \prime}$.

The second and fifth equalities below  invoke the definition for  $\ell^{(\kappa)}(\rho)$:
\begin{align}\label{Axis}
\int_{A}d\phi \, \big| \big\langle f_{\phi}\big| \ell^{(\kappa)}_{\phi}(\rho)\big|g_{\phi}\big\rangle\big| &=\int_{\mathbb{T}}d\phi \, \Tr\big[ \ell^{(\kappa)}_{\phi}(\rho)e^{\textup{i}\kappa  X_{\mathbb{T}} } \widetilde{G}_{\phi}^{(A)} e^{\textup{i}\kappa X_{\mathbb{T}} }  \big]=\Tr\big[\rho e^{\textup{i}\kappa X}G^{(A)} e^{\textup{i}\kappa X} \big]\nonumber  \\ &=\Tr\Big[\big(  G^{(A),\prime}e^{-\textup{i}\kappa X}|\rho|^{\frac{1}{2}}\big)^{*}\big(  G^{(A),\prime \prime}   e^{\textup{i}\kappa X}U_{\rho}|\rho|^{\frac{1}{2}}\big)      \Big] \nonumber \\  &\leq    \Tr\Big[\big|  G^{(A),\prime}e^{-\textup{i}\kappa X}|\rho|^{\frac{1}{2}}\big|^{2}\Big]^{\frac{1}{2}} \Tr\Big[\big|  G^{(A),\prime \prime}e^{\textup{i}\kappa X}U_{\rho}|\rho|^{\frac{1}{2}}\big|^{2}\Big]^{\frac{1}{2}} \nonumber  \\ &=    \Tr\Big[|\rho| \big(e^{\textup{i}\kappa X} |  G^{(A),\prime}|^2 e^{-\textup{i}\kappa X}\big)\Big]^{\frac{1}{2}}  \Tr\Big[|\rho|_{*} \big(e^{-\textup{i}\kappa X} |  G^{(A),\prime \prime}|^2 e^{\textup{i}\kappa X}\big)\Big]^{\frac{1}{2}}\nonumber \\ & = \Big(  \int_{A}d\phi  \,\big\langle g_{\phi }\big| \ell^{(0)}_{\phi +\kappa}(|\rho|)\big| g_{\phi}\big\rangle    \Big)^{\frac{1}{2}} \Big(  \int_{A}d\phi  \, \big\langle f_{\phi }\big| \ell^{(0)}_{\phi -\kappa}(|\rho|_{*})\big| f_{\phi}\big\rangle  \Big)^{\frac{1}{2}},
\end{align}
where  $U_{\rho}\in \mathcal{B}(\mathcal{H})$ is the unitary operator in the polar decomposition of $\rho$.   The third and fourth equalities hold by the cyclicity of trace, and the fourth also uses that $|\rho|_{*}= U_{\rho}|\rho|U_{\rho}^{*}$.   The inequality is the Cauchy-Schwarz inequality $|\Tr[Y^{*}Z]|^2\leq \Tr[|Y|^2]\Tr[|Z|^2]$ for Hilbert-Schmidt operators $Y,Z$.   For the fifth equality, the operators $e^{\textup{i}\kappa X} |  G^{(A),\prime}|^2 e^{-\textup{i}\kappa X}$ and $e^{-\textup{i}\kappa X} |  G^{(A),\prime \prime}|^2 e^{\textup{i}\kappa X}$ are in $\mathcal{A}_{\mathbb{T}}$, and the corresponding elements in $L^{\infty}\big(\mathbb{T},\mathcal{B}(L^2(\mathcal{I}) )\big)$ are respectively $1_{A}(\phi-\kappa)|g_{\phi-\kappa}\rangle \langle g_{\phi-\kappa}|$ and $1_{A}(\phi+\kappa)|f_{\phi+\kappa}\rangle \langle f_{\phi+\kappa}|$.  

Since~(\ref{Axis}) holds for all measurable sets $A\subset \mathbb{T}$, it follows that for a.e. $\phi\in \mathbb{T}$, 
$$   \big| \big\langle f_{\phi}\big| \ell^{(\kappa)}_{\phi}(\rho)\big|g_{\phi}\big\rangle\big|\leq  \big\langle f_{\phi}\big| \ell^{(0)}_{\phi-\kappa }(|\rho|)\big| f_{\phi} \big\rangle^{\frac{1}{2}}  \big\langle g_{\phi }\big| \ell^{(0)}_{\phi +\kappa}(|\rho|_{*})\big| g_{\phi}\big\rangle^{\frac{1}{2}} .   $$

\vspace{.4cm}

\noindent Part (2):  By definition $\langle \rho \rangle_{\phi}^{(\kappa)}= \Tr[\ell^{(\kappa)}_{\phi}(\rho)]$.  The result follows by choosing an orthonormal basis for $L^{2}(\mathcal{I})$ in which to compute the trace:
\begin{align*}
\big|\Tr[\ell^{(\kappa)}_{\phi}(\rho)]\big|= &\Big|\sum_{n\in \Z}\big\langle \psi_{n }\big| \ell^{(\kappa)}_{\phi}(\rho)\big| \psi_{n }  \big\rangle    \Big|  \\ \leq & \sum_{n\in \Z} \big\langle \psi_{n }\big| \ell^{(0)}_{\phi-\kappa }(|\rho|)\big| \psi_{n }  \big\rangle^{\frac{1}{2}} \big\langle \psi_{n }\big| \ell^{(0)}_{\phi+\kappa }(|\rho|_{*})\big| \psi_{n }  \big\rangle^{\frac{1}{2}}
\\ \leq & \Big( \sum_{n\in \Z} \big\langle \psi_{n }\big| \ell^{(0)}_{\phi-\kappa }(|\rho|)\big| \psi_{n }  \big\rangle   \Big)^{\frac{1}{2}}\Big( \sum_{n\in \Z} \big\langle \psi_{n }\big| \ell^{(0)}_{\phi+\kappa }(|\rho|_{*})\big| \psi_{n }  \big\rangle   \Big)^{\frac{1}{2}}\\ = & \Tr\big[\ell^{(0)}_{\phi-\kappa}(|\rho|)\big]^{\frac{1}{2}}\Tr\big[\ell^{(0)}_{\phi+\kappa}(|\rho|_{*})\big]^{\frac{1}{2}}=\big(\big\langle |\rho| \big\rangle_{\phi-\kappa}^{(0)}\big)^{\frac{1}{2}}\big(\big\langle |\rho|_{*} \big\rangle_{\phi+\kappa}^{(0)}\big)^{\frac{1}{2}}.
\end{align*}
The first inequality above is by Part (1) and the second is by Cauchy-Schwarz. 

\vspace{.4cm}

\noindent Part (3):  By definition $[\rho]^{(k)}(p):=\langle \psi_{p-k}|  \ell^{(\kappa)}_{\phi}(\rho) \psi_{p+k}\rangle  $ for $p=\phi\,\textup{mod}\,1$ and  $k=\kappa\,\textup{mod}\,\frac{1}{2}$.  The result follows directly from Part (1).  The same argument applies for $[\rho]^{(k)}_{\scriptscriptstyle{Q}}$.

\end{proof}

\begin{proposition}\label{PropFiberII}

Let $\rho \in \mathcal{B}_{1}(\mathcal{H})$. 
\begin{enumerate}

\item The following equalities hold:
$$
\int_{\R}dp\, [\rho ]^{(0)}_{\scriptscriptstyle{Q}}(p)=\int_{\R}dp\, [\rho ]^{(0)}(p)= \int_{\mathbb{T}}d\phi \,\big\langle \rho \big\rangle_{\phi}^{(0)}=\Tr[\rho].    
$$

\item For any $k\in \R $,
$$\Tr\big[\rho e^{\textup{i}2k X}    \big]=\int_{\R}dp \, [\rho ]^{(k)}(p).  $$

\item  For any  $\kappa\in [-\frac{1}{4},\, \frac{1}{4})$ and  $k\in \R $, 
$$ \big\|\langle \rho \rangle^{(\kappa)} \big\|_{1}, \big\| [\rho ]^{(k)}\big\|_{1}, \, \big\| [\rho ]^{(k)}_{\scriptscriptstyle{Q}}\big\|_{1}\leq \|\rho\|_{\mathbf{1}}. $$

\end{enumerate}

\end{proposition}

\begin{proof}[Proof of Prop.~\ref{PropFiberII}]\text{ }\\
\noindent Part (1):   The first and fourth equalities below hold by the definitions for $[\rho]^{(0)}$ and $\ell^{(0)}_{\phi}(\rho)$, respectively:
\begin{align*}
\int_{\R}dp\, [\rho]^{(0)}(p)= &  \int_{\R}dp \,\big\langle \psi_{p}\big|  \ell^{(0)}_{\phi}( \rho) \psi_{p}\big\rangle \\ = &  \int_{\mathbb{T}}d\phi \, \sum_{n\in \Z}  \big\langle \psi_{\phi+n }\big|  \ell^{(0)}_{\phi}( \rho) \psi_{\phi+n}\big\rangle\\ =&\int_{\mathbb{T}}d\phi \,\Tr[  \ell^{(0)}_{\phi}( \rho) ]=\Tr[\rho],  
\end{align*}
where, on the first line, $\phi \in [-\frac{1}{2},\frac{1}{2})$ with   $\phi =p\,\textup{mod}\,1$.  The third equality uses that $\psi_{\phi+n}$, $n\in \Z$ is an orthonormal basis of $L^{2}(\mathcal{I})$ for each $\phi \in \mathbb{T}$.  I also have the equality $\int_{\mathbb{T}}d\phi \,\langle \rho \rangle_{\phi}^{(0)}=\Tr[\rho]$,  since  $\langle \rho \rangle_{\phi}^{(0)}=\Tr[  \ell^{(0)}_{\phi}( \rho) ]$.   The  argument is analogous for $[\rho]^{(0)}$ replaced by  $[\rho]^{(0)}_{\scriptscriptstyle{Q}}$.

\vspace{.4cm}

\noindent Part (2): Let $\phi\in \mathbb{T}$ be equal to  $p$ modulo $1$ and $ \kappa\in [-\frac{1 }{4},\frac{1}{4})$ be equal to $k $ modulo $\frac{1}{2}$.   Also let $n:=2k-2\kappa $.  Since  $e^{\textup{i}nX}      \in \mathcal{A}_{\mathbb{T}}$ for $n\in \Z$, the definition for $\ell^{(\kappa)}(\rho)$ yields the second equality below:
\begin{align*}
 \Tr\big[\rho e^{\textup{i}2k X}    \big] & = \Tr\big[\rho e^{\textup{i}\kappa X} e^{\textup{i}n X}       e^{\textup{i}\kappa X}\big]=\int_{\mathbb{T}}d\phi \, \Tr\big[\ell_{\phi}^{(\kappa)}(\rho)e^{\textup{i}\kappa X_{\mathbb{T} }} e^{\textup{i}nX_{\mathbb{T}} } e^{\textup{i}\kappa X_{\mathbb{T} }} \big]\\ &= \int_{\mathbb{T}}d\phi \,\sum_{m\in \Z}\big\langle   \psi_{m}\,\big|\,e^{\textup{i}k X_{\mathbb{T} }}\ell_{\phi}^{(\kappa)}(\rho)e^{\textup{i}k X_{\mathbb{T} }}  \psi_{m}\big\rangle = \int_{\mathbb{T}}d\phi \,\sum_{m\in \Z} \big\langle   \psi_{m-k}  \,\big|\,\ell_{\phi}^{(\kappa)}(\rho) \psi_{m+ k }\big\rangle  \\ & = \int_{\R}dp \,\big\langle   \psi_{p- k}  \,\big|\,\ell_{\phi}^{(\kappa)}(\rho) \psi_{p+ k }\big\rangle = \int_{\R}dp\, [\rho]^{(k)}(p) .
 \end{align*}
 The fourth equality uses that $e^{\textup{i}v X_{\mathbb{T} }} \psi_{p} =\psi_{p+ v} $ for $p,v\in \R$.

\vspace{.4cm}

\noindent Part (3): This is a consequence of Part (1) above and Parts (2) and (3) of Prop.~\ref{MiscFiber}.

\end{proof}

\subsection{Fiber decomposition for the Lindblad dynamics}

Since the Hamiltonian is spatially periodic and the noise~(\ref{TheNoise}) is invariant under all spatial shifts,  the Lindblad dynamics~(\ref{TheModel}) is invariant under  shifts by the period $2\pi $ of the Dirac comb.  In terms of the dynamical semigroup $\Phi_{\lambda,t}:\mathcal{B}_{1}(\mathcal{H})$, the spatial  symmetry translates to the covariance:
\begin{align*}
\hspace{3.5cm}\Phi_{\lambda,t}\big(e^{\textup{i}2\pi P}\rho e^{-\textup{i}2\pi P}\big)= e^{\textup{i}2\pi P}\Phi_{\lambda,t}(\rho) e^{-\textup{i}2\pi P}, \hspace{1.5cm} \rho\in\mathcal{B}_{1}( \mathcal{H} ).
\end{align*}
As a consequence, the dynamics decomposes into fibers as stated in Part (1) of Prop.~\ref{PropMoreFiber}.  To be mathematically rigorous, the differential equation in Part (1) of Prop.~\ref{PropMoreFiber} should be phrased as an integral equation and evaluated against an appropriate class of test functions.  The constant $\varpi >0 $ in Parts (3) and (4) of the proposition below is from the jump rate assumptions in List~\ref{Assumptions}.

\begin{proposition}\label{PropMoreFiber} \text{   }
\begin{enumerate}
\item For each $\kappa\in  [-\frac{1}{4}, \frac{1}{4})$, the functions $ \ell^{(\kappa)} \big( \Phi_{\lambda,t}(\rho) \big)\in L^{1}\left(\mathbb{T},\,\mathcal{B}_{1}(L^{2}(\mathcal{I}))\right)$ satisfy the differential equation
\begin{align*} 
\frac{d}{dt}\ell^{(\kappa)}_{\phi} \big( \Phi_{\lambda,t}(\rho) \big)=&-\frac{\textup{i}}{\lambda^{\varrho}}\Big(H_{\phi- \kappa } \ell^{(\kappa)}_{\phi} \big( \Phi_{\lambda,t}(\rho) \big)-  \ell^{(\kappa)}_{\phi} \big( \Phi_{\lambda,t}(\rho) \big) H_{\phi+\kappa}\Big)-\mathcal{R}\ell^{(\kappa)}_{\phi} \big( \Phi_{\lambda,t}(\rho) \big)\\ &+\int_{\mathbb{T}}d\phi'\, \sum_{n\in \Z}j(\phi-\phi'+n) e^{\textup{i}(\phi-\phi'+ n)  X  }\ell^{(\kappa)}_{\phi'} \big( \Phi_{\lambda,t}(\rho) \big) e^{-\textup{i}(\phi-\phi'+n)  X   } .
\end{align*}
In particular,  there is a contractive semigroup $\Gamma_{\lambda,t}^{(\kappa)}:L^{1}\left(\mathbb{T},\,\mathcal{B}_{1}(L^{2}(\mathcal{I}) )\right)   $   such that
for all $\rho \in \mathcal{B}_{1}(\mathcal{H})$,
$$ \Gamma_{\lambda,t}^{(\kappa)}\big(\ell^{(\kappa)} (\rho)\big)=  \ell^{(\kappa)} \big( \Phi_{\lambda,t}(\rho) \big).    $$

\item  For all $\kappa\in  [-\frac{1}{4}, \frac{1}{4})$ and $\lambda>0 $, the functions $\big\langle \Phi_{\lambda,t}(\rho)\big\rangle^{(\kappa)}\in L^{1}(\mathbb{T})$ satisfy the Kolmogorov equation 
$$ \frac{d}{dt}\big\langle \Phi_{\lambda,t}(\rho)\big\rangle^{(\kappa)}_{\phi}=-\mathcal{R}\big\langle\Phi_{\lambda,t}(\rho)\big\rangle^{(\kappa)}_{\phi}+\int_{\mathbb{T}}d\phi'\, J_{\mathbb{T}}(\phi,\phi')\big\langle \Phi_{\lambda,t}(\rho)\big\rangle^{(\kappa)}_{\phi'},  $$
where  $J_{\mathbb{T}}(\phi,\phi'):=\sum_{n\in \Z}j(\phi-\phi'+n)  $.

\item  For all  $\kappa\in  [-\frac{1}{4}, \frac{1}{4})$ and $\rho \in \mathcal{B}_{1}(\mathcal{H})  $, the following inequality holds: $ \big\|\big\langle \Psi(\rho) \big\rangle^{(\kappa)}\big\|_{\infty} \leq \varpi\|\rho\|_{\mathbf{1}}.$

\item For all $\lambda>0$, $\kappa\in  [-\frac{1}{4}, \frac{1}{4})$,  $t\in \R^{+}$,  and $\rho \in \mathcal{B}_{1}(\mathcal{H}) $, $$  \big\|\big\langle \Phi_{\lambda,t}(\rho) \big\rangle^{(\kappa)}\big\|_{\infty} \leq e^{-\mathcal{R}t} \big\|\langle \rho \rangle^{(\kappa)}\big\|_{\infty}+ \frac{\varpi}{\mathcal{R}}\|\rho\|_{\mathbf{1}}.$$

\end{enumerate}

\end{proposition}

\begin{proof}\text{   }\\
\noindent Part (1):   The maps  $\Phi_{\lambda,t}:\mathcal{B}_{1}(\mathcal{H}   )$ satisfy the  Duhamel equation
\begin{align}\label{Mary}
\Phi_{\lambda,t}(\rho)= e^{-\mathcal{R}t}e^{-\frac{\textup{i}t}{\lambda^{\varrho}}   H }\rho  e^{\frac{\textup{i}t}{\lambda^{\varrho}}    H }+\int_{0}^{t}dr\, e^{-\mathcal{R}(t-r)}e^{-\frac{\textup{i}(t-r ) }{\lambda^{\varrho}} H }\Psi\big(\Phi_{\lambda,r}(\rho)\big)e^{\frac{\textup{i}(t-r) }{\lambda^{\varrho}} H } .  
\end{align}
Thus, I have the integral equation
\begin{align*}
\ell^{(\kappa)}_{\phi} \big( \Phi_{\lambda,t}(\rho) \big)=  e^{-\mathcal{R}t}\ell^{(\kappa)}_{\phi}\big(e^{-\frac{\textup{i}t}{\lambda^{\varrho}}  H }\rho  e^{\frac{\textup{i}t}{\lambda^{\varrho}} H }\big)+\int_{0}^{t}dr \,e^{-\mathcal{R}(t-r)}\ell^{(\kappa)}_{\phi} \Big( e^{-\frac{\textup{i}(t-r  ) }{\lambda^{\varrho}}H }\Psi\big(\Phi_{\lambda,r}(\rho)\big)e^{\frac{\textup{i}(t-r ) }{\lambda^{\varrho}} H }    \Big).
\end{align*}
From the definition of $\ell^{(\kappa)}(\rho)$, it  can be shown  that
\begin{align*}
\ell^{(\kappa)}_{\phi}\big(e^{-\textup{i}t  H }\rho  e^{\textup{i}t   H }\big)=& e^{-\textup{i}t H_{\phi-\kappa} } \ell^{(\kappa)}_{\phi}(\rho  )e^{\textup{i}t H_{\phi+\kappa} }\quad \text{and} \quad  \ell^{(\kappa)}_{\phi}\big(\Psi(\rho) \big)=\int_{\mathbb{T}}d\phi'\,\widehat{\Psi}_{\phi-\phi'}\big(\ell^{(\kappa)}_{\phi '}(\rho)\big) ,
\end{align*}
where $\widehat{\Psi}_{\phi}:\mathcal{B}_{1}\big(L^2(\mathcal{I})   \big)$ is defined for $\phi \in \mathbb{T}$ as
$$\hspace{3.5cm}\widehat{\Psi}_{\phi}(h)=\sum_{n\in \Z}j\big(\phi +n\big) e^{\textup{i}(\phi+n ) X_{\mathbb{T}}  }h e^{-\textup{i}(\phi+n)  X_{\mathbb{T}}  },\hspace{1cm}h\in \mathcal{B}_{1}\big(L^2(\mathcal{I})   \big). $$
From the above equalities, it follows that $\ell^{(\kappa)} \big( \Phi_{\lambda,t}(\rho) \big)$ satisfies an integral equation of its own:
\begin{align}\label{Jubee}
\ell^{(\kappa)}_{\phi} \big( \Phi_{\lambda,t}(\rho) \big)= & e^{-\mathcal{R}t}e^{-\frac{\textup{i}t}{\lambda^{\varrho}} H_{\phi-\kappa } } \ell^{(\kappa)}_{\phi}(\rho  )e^{\frac{\textup{i}t}{\lambda^{\varrho}} H_{\phi+\kappa} }\nonumber \\ & +\int_{0}^{t}dr \, e^{-\mathcal{R}(t-r)} e^{-\frac{\textup{i}(t-r)}{\lambda^{\varrho}}H_{\phi-\kappa} }\Big(\int_{\mathbb{T}}d\phi' \,\widehat{\Psi}_{\phi-\phi'}\big(\ell^{(\kappa)}_{\phi'} (\Phi_{\lambda,r}(\rho) )\big)\Big)e^{\frac{\textup{i}(t-r)}{\lambda^{\varrho}}H_{\phi+\kappa} }.  
\end{align}
Since convolution with $\widehat{\Psi}_{\phi}$ is a bounded map on $L^{1}\big(\mathbb{T}, \mathcal{B}_{1}\big(L^2(\mathcal{I}) \big)  \big)$ and  the operators $e^{-\textup{i}t H_{\phi} }$, $\phi\in \mathbb{T}$ are unitary, a semigroup $\Gamma_{\lambda,t}^{(\kappa)}:L^{1}\big(\mathbb{T}, \mathcal{B}_{1}(L^2(\mathcal{I}) )  \big)$ can be constructed  through the Dyson series corresponding to the integral equation~(\ref{Jubee}) that satisfies $\Gamma_{\lambda,t}^{(\kappa)}\big(\ell^{(\kappa)}_{\phi}(\rho)\big)=\ell^{(\kappa)}_{\phi} \big( \Phi_{\lambda,t}(\rho) \big)$.  The semigroup $\Gamma_{\lambda,t}^{(\kappa)}$ is contractive since the noise term conforms to the bound
$$ \int_{\mathbb{T}}d\phi\,\Big\| \int_{\mathbb{T}}d\phi' \widehat{\Psi}_{\phi-\phi'}\big(\ell^{(\kappa)}_{\phi'} (\rho)\big)    \Big\|_{\mathbf{1}} \leq \int_{\mathbb{T}}d\phi \,\int_{\mathbb{T}}d\phi'\, J_{\mathbb{T}}(\phi, \phi')\big\|\ell^{(\kappa)}_{\phi'} (\rho)\big\|_{\mathbf{1}}=\mathcal{R}\big\| \ell^{(\kappa)} (\rho)  \big\|_{1}.  $$

\vspace{.4cm}

\noindent Part (2): 
By taking the trace of both sides of~(\ref{Jubee}), I obtain the integral equation
\begin{align}\label{Hum}
\big\langle \Phi_{\lambda,t}(\rho) \big\rangle^{(\kappa)}_{\phi} =  e^{-\mathcal{R}t}\big\langle \rho \big\rangle^{(\kappa)}_{\phi} +\int_{0}^{t}dr\, e^{-\mathcal{R} (t-r)}\int_{\mathbb{T}}d\phi' J_{\mathbb{T}}(\phi,\phi') \big\langle \Phi_{\lambda,r}(\rho) \big\rangle^{(\kappa)}_{\phi'} , 
\end{align}
where I have used that $\Tr\big[\widehat{\Psi}_{\phi-\phi '}(h)\big]=J_{\mathbb{T}}(\phi , \phi')\Tr[h]    $ for  $h\in \mathcal{B}_{1}\big(L^2(\mathcal{I})   \big)$.  Differentiating~(\ref{Hum}) yields the Kolmogorov equation.

\vspace{.4cm}

\noindent Part (3):  For all $\kappa\in [-\frac{1}{4},\frac{1}{4})$ and $\rho\in \mathcal{B}_{1}(\mathcal{H})$, I have the closed formula
$$ \big \langle \Psi(\rho) \big\rangle^{(\kappa)}_{\phi}=\int_{\mathbb{T}}d\phi'\, J_{\mathbb{T}}(\phi,\phi')  \big \langle \rho \big\rangle^{(\kappa)}_{\phi '} . $$
Thus, taking the infemum norm of both sides yields the inequality 
$$\big\|  \big\langle \Psi(\rho) \big\rangle^{(\kappa)} \big\|_{\infty}\leq  \Big( \sup_{\phi,\phi'\in \mathbb{T} }J_{\mathbb{T}}(\phi,\phi')    \Big) \big\|   \big \langle \rho \big\rangle^{(\kappa)} \big\|_{1}\leq \varpi \|\rho\|_{\mathbf{1}}, $$
where the second inequality uses assumption (2) of List~\ref{Assumptions}  followed by Part (3) of Prop.~\ref{PropFiberII}.  

\vspace{.4cm}

\noindent Part (4): The integral equation from Part (2) implies that $ \big \langle \Phi_{\lambda,t}(\rho) \big\rangle^{(\kappa)} $ can be written as
$$\big \langle \Phi_{\lambda,t}(\rho) \big\rangle^{(\kappa)}= e^{-\mathcal{R}t}\big \langle \rho \big\rangle^{(\kappa)} +\frac{J_{\mathbb{T}}}{\mathcal{R}}\Big(e^{-\mathcal{R}t}\sum_{n=1}^{\infty} \frac{(\mathcal{R}t)^{n}}{n!} \frac{J_{\mathbb{T}}^{n-1}}{\mathcal{R}^{n-1}}  \Big) \langle \rho \big\rangle^{(\kappa)} . $$
The inequality $\| \big \langle \Phi_{\lambda,t}(\rho) \big\rangle^{(\kappa)}\|_{\infty}\leq e^{-\mathcal{R}t}  \| \langle \rho \big\rangle^{(\kappa)}\|_{\infty}  +\frac{\varpi}{\mathcal{R}}\|\rho\|_{\mathbf{1}}$ follows by $\| \sum_{n=1}^{\infty} \frac{(\mathcal{R}t)^{n}}{n!} \frac{J_{\mathbb{T}}^{n-1}}{\mathcal{R}^{n-1}}\langle \rho \big\rangle^{(\kappa)} \|_{1}\leq  \|\langle \rho \big\rangle^{(\kappa)} \|_{1}$ and the reasoning in Part (3).



\end{proof}

\section{Energy submartingales and unravelings of the dynamical maps  }\label{SecPseudoPoisson}

\subsection{Pseudo-Poisson and L\'evy  unravelings}\label{SecPandL}
As mentioned in Sect.~\ref{SectIntro}, the map $\Psi:\mathcal{B}_{1}(\mathcal{H})$  generating the noise for the quantum dynamical semigroup $\Phi_{\lambda,t}:\mathcal{B}_{1}(\mathcal{H})$ satisfies $\Psi^{*}(I)=\mathcal{R}I$.  This should be interpreted as meaning that the escape rates for the quantum Markovian process are invariant of the state.  This property, referred to as \textit{pseudo-Poisson} for classical processes, implies that the dynamical maps $\Phi_{\lambda,t} $ can be unraveled in terms of an underlying Poisson process with rate $\mathcal{R}$.  For each $k$, the semigroups $\Phi_{\lambda,t}^{(k)}:L^{1}(\R)$ also have the pseudo-Poisson property.  The proof of Part (1) from Lem.~\ref{Trivial} is contained in~\cite[Appx.A]{Dispersion}, and the proof of Part (2) is similar.

 Define the transition operator  $T_{k}:L^{1}(\R)$ to have kernel
$T_{k}(p,p')= \mathcal{R}^{-1}J_{k}(p,p') $ for $J_{k}$ defined in~(\ref{OffJays})  and recall that $U_{\lambda, t}^{(k)}:L^{1}(\R)$ acts as multiplication by the function $U_{\lambda, t}^{(k)}(p)=e^{-\frac{\textup{i}t}{\lambda^{\varrho}}\left(E(p - k )-E(p +k)\right)    }$.

\begin{lemma}\label{Trivial}
Let $\mathbb{E}$ denote the expectation with respect to a rate-$\mathcal{R}$ Poisson process $\mathcal{N}\equiv \mathcal{N}_{t}(\xi)$ with  realizations $\xi=(t_{1},t_{2},\dots)\in (\R^{+})^{\infty}$, $t_{j}\leq t_{j+1}$.    

\begin{enumerate}
\item  The map $\Phi_{\lambda,t}:\mathcal{B}_{1}(\mathcal{H})$ can be written as  $\Phi_{\lambda,t}=\mathbb{E}\big[ \Phi_{\lambda, \xi,t}  \big]$, where
$$\Phi_{\lambda,\xi, t}(\rho):=\mathcal{R}^{-\mathcal{N}}    e^{-\frac{\textup{i}(t-t_{\mathcal{N}})}{\lambda^{\varrho}}H}  \Psi(   \cdots  e^{-\frac{\textup{i}(t_{2}-t_{1})}{\lambda^{\varrho}}H}\Psi(e^{-\frac{\textup{i}t_{1} }{\lambda^{\varrho}}H}\rho e^{\frac{\textup{i}t_{1}}{\lambda^{\varrho}}H})e^{\frac{\textup{i}(t_{2}-t_{1})}{\lambda^{\varrho}}H}\cdots)e^{\frac{\textup{i}(t-t_{\mathcal{N} })}{\lambda^{\varrho}}H}.     $$ 

\item Similarly, the map $\Phi_{\lambda,t}^{(k)}:L^{1}(\R)$ can be written as $\Phi_{\lambda,\xi,t}^{(k)}=\mathbb{E}\big[ \Phi_{\lambda, \xi,t}^{(k)}  \big]$, where 
$$\Phi_{\lambda,\xi, t}^{(k)}(\rho):= U_{\lambda, t-t_{\mathcal{N}}}^{(k)} T_{k} \cdots  U_{\lambda, t_{2}-t_{1}}^{(k)}T_{k}U_{\lambda, t_{1}}^{(k)}.     $$

\end{enumerate}

\end{lemma}

  The semigroup $\Phi_{\lambda,t}$ has an even more restrictive property than being pseudo-Poisson, since the noise map $\Psi$ has a Kraus decomposition~(\ref{TheNoise}) comprised of an integral combination of unitary conjugations.  This allows the maps $\Phi_{\lambda,t}$ to be written as convex combinations of unitary conjugations using an underlying L\'evy process with jump rate density $j(v)$.  Given an element $\xi=(v_{1},t_{1};\, v_{2},t_{2};\dots )\in (\mathbb{R}\times \mathbb{R}^{+})^{\infty}$ with $t_{j}\leq t_{j+1}$, the unitary operator  $U_{\lambda,t}(\xi):\mathcal{H}$ is defined by the product
 \begin{align}\label{LaughableMan}
 U_{\lambda,t}(\xi):= e^{-\frac{\textup{i}(t-t_{n })}{\lambda^{\varrho}}   H }e^{\textup{i}v_{n }X}  \cdots  e^{-\frac{\textup{i}(t_{2}-t_{1}) }{\lambda^{\varrho}}H} e^{\textup{i}v_{1}X}  e^{-\frac{\textup{i} t_{1}}{\lambda^{\varrho}} H}.
 \end{align} 
Lemma~\ref{LemLevyUnravel} is from~\cite[Appx.A]{Dispersion}.

\begin{lemma}\label{LemLevyUnravel}
Let  $\xi=(v_{1},t_{1};\, v_{2},t_{2};\dots )\in (\mathbb{R}\times \mathbb{R}^{+})^{\infty}$ be the realization for a L\'evy process with jump rate density $j(v)$.  The dynamical maps $\Phi_{\lambda,t}$ can be written as
\begin{align}\label{Egret}
 \Phi_{\lambda,t}(\rho)=  \mathbb{E}\Big[  U_{\lambda,t}(\xi)\,   \rho \, U_{\lambda,t}^{*}(\xi) \Big], 
 \end{align}
where the expectation is with respect to the law of the L\'evy process.  
 \end{lemma}

\subsection{Energy submartingales}

Parts (1) and (2) of the proposition below are energy submartingale properties from~\cite[Prop.4.1]{Dispersion}.  Part (3) follows by a similar argument as Part (2).  Parts (1) and (2) carry over to analogous properties for the classical stochastic process $K_{t}$ discussed in Sect.~\ref{SecPreClassical};  see Prop.~\ref{SubMart}.

\begin{proposition}\label{LemSubMartBasic}
Let the unitary process $U_{\lambda,t}(\xi)\in \mathcal{B}\big(L^{2}(\mathcal{H})\big)$ be defined as in~(\ref{LaughableMan}) for times $t\in \R^{+}$ and a realization $\xi\in (\mathbb{R}\times \mathbb{R}^{+})^{\infty}$ of a L\'evy process with rate density $j(v)$.

\begin{enumerate}
\item For every $f\in \textup{D}(H)$, the stochastic process $ \big\langle U_{\lambda,t}(\xi)  f\big|\,H^{\frac{1}{2}} U_{\lambda,t}(\xi)  f      \big\rangle $
is an integrable submartingale with respect to the filtration of the L\'evy process.   

\item The evaluation of the Hamiltonian by the Heisenberg evolution maps $\Phi_{\lambda,t}^{*}$ has the explicit form:
$$\Phi_{\lambda,t}^{*}(H)=H+\sigma t   I .  $$

\item A similar formula holds for each map $\Phi_{\lambda,\xi, t_{n} }^{*}$ when acting on the Hamiltonian:
$$\Phi_{\lambda,\xi, t_{n}}^{*}(H)=H+\frac{\sigma n   }{\mathcal{R}  }I.$$

\end{enumerate}

\end{proposition}

\section{From the momentum to the extended-zone scheme representation}\label{SecQuasiMomentum}

In this section I focus on proving Lem.~\ref{StanToQuasi}.

\subsection{Control over the energy        }

The lemma below contains estimates for the square roots of the  dispersion relation $E(p)$ and the Hamiltonian $H$.

\begin{lemma}\label{LemLittleEnergy}
There is a $C>0$ such that for all $\lambda< 1 $ and $p\in \R$,
\begin{enumerate}

\item $\big| E^{\frac{1}{2}}(p)-|p|   \big| \leq C ,$

\item $\Tr\Big[ \check{\rho}_{\lambda}\big(H^{\frac{1}{2}}-\big(P^{2}   \big)^{\frac{1}{2}} \big) \Big]\leq C.$

\end{enumerate}

\end{lemma}

\begin{proof}\text{  }\\
\noindent Part (1):  The inequality holds with $C=\frac{1}{2}$  since $\big| E^{\frac{1}{2}}(p)-|p|   \big|=|\mathbf{q}(p)-p|$ and the values $\mathbf{q}(p) $ and $p$ can not be separated by more than $\frac{1}{2}$ by the Kr\"onig-Penney relation~(\ref{KronigPenney}).

\vspace{.4cm}

\noindent Part (2): Since $\check{\rho}_{\lambda}:=|\frak{h}\rangle \langle \frak{h}|$ for $\frak{h}\in \mathcal{H}$ of defined above,  I have the first equality below:   
\begin{align}\label{Jango}
\Tr\Big[ \check{\rho}_{\lambda}\Big(H^{\frac{1}{2}}-( P^{2}   )^{\frac{1}{2}} \Big) \Big]=\Big\langle \frak{h}\Big| H^{\frac{1}{2}}-( P^{2}   )^{\frac{1}{2}}\Big|\frak{h}\Big \rangle=\int_{\mathbb{T}}d\phi \, \Big\langle \widehat{\frak{h}}_{\phi}\Big| H_{\phi}^{\frac{1}{2}}-( P^{2}    )_{\phi}^{\frac{1}{2}}\Big|\widehat{\frak{h}}_{\phi}\Big \rangle  .
\end{align}
The second equality  invokes the fiber decomposition discussed in Sect.~\ref{SecHamFiber}.  The operators   $H_{\phi}$, $(P^2)_{\phi}$ for $\phi\in \mathbb{T}$ denote the operation of $H$ and $P^2$ on the $\phi$-fiber copy  of $L^{2}(\mathcal{I})$ in the tensor product decomposition $\mathcal{H}=L^{2}(\mathbb{T})\otimes L^{2}(\mathcal{I})$.

By using the formula $u^{\frac{1}{2}}=\frac{1}{\pi}\int_{0}^{\infty}d\epsilon\,\epsilon^{-\frac{1}{2}} \frac{u}{\epsilon+u}$  for $u\in \R^{+}$ and functional calculus~\cite[Ch.VIII.Ex.50]{Simon}, I can write the difference between the square roots of the Hamiltonians $H_{\phi}$ and $(P^{2}  )_{\phi}$ as
\begin{align}\label{Babel} H^{\frac{1}{2}}_{\phi}-(P^{2} )_{\phi}^{\frac{1}{2}} =&\frac{1}{\pi}\int_{0}^{\infty}\,\frac{d\epsilon}{\epsilon^{\frac{1}{2}}}\Big(\frac{H_{\phi}}{\epsilon+H_{\phi}}-\frac{  (P^{2})_{\phi}   }{ \epsilon+ (P^{2})_{\phi}   }   \Big)\nonumber \\=&\frac{1}{\pi} \int_{0}^{1} \frac{d\epsilon}{\epsilon^{\frac{1}{2}}}\Big(\frac{H_{\phi}}{\epsilon+H_{\phi} }-\frac{  (P^{2} )_{\phi}      }{ \epsilon+ (P^{2} )_{\phi}       }   \Big)     +  \frac{1}{\pi} \int_{1 }^{\infty}d\epsilon\,\epsilon^{\frac{1}{2}}\Big(\frac{  1    }{ \epsilon+ (P^{2} )_{\phi}      }  -\frac{1}{\epsilon+H_{\phi}}  \Big) .   
\end{align}
However, the operators in the integrands have the bounds 
\begin{align*}
(\textup{i}).\hspace{.2cm}\Big\|\frac{H_{\phi}}{\epsilon+H_{\phi}}-\frac{  (P^{2} )_{\phi}    }{ \epsilon+ (P^{2}  )_{\phi}  } \Big\|\leq 1 \quad \text{and} \quad (\textup{ii}).\hspace{.2cm}\Big\langle \widehat{\frak{h}}_{\phi}\Big| \frac{  1    }{ \epsilon+ (P^{2} )_{\phi}      }  -\frac{1}{\epsilon+H_{\phi} } \Big|\widehat{\frak{h}}_{\phi}\Big\rangle  \leq  \frac{c}{\epsilon^2  },  
\end{align*}
where the second inequality is for some  $c>0$.  The inequality (i) uses that the function $\frac{ x }{\epsilon+x   }$ is operator monotonically increasing for each $\epsilon\in \R^{+}$ and that $\frac{H}{\epsilon+H}\leq 1$.  I will prove (ii) below.   Applying ~(\ref{Jango}) and~(\ref{Babel}) with the inequalities (i) and (ii) yields the bound 
\begin{align*}
 \Tr\Big[ \check{\rho}_{\lambda}\Big(H^{\frac{1}{2}}-\big(P^{2}  \big)^{\frac{1}{2}} \Big) \Big] \leq   \frac{1}{\pi}+\frac{2 c   }{\pi   } ,
\end{align*}
 which would complete the proof.

 \vspace{.4cm}
\noindent (ii).  The operator $\epsilon+( P^{2})_{\phi}$ has Green function $G_{\phi,\epsilon}:\mathcal{I}\rightarrow \C$ with the closed form
$$  G_{\phi,\epsilon}(x)=\frac{1}{2\pi}\sum_{n\in \Z}\frac{1}{\epsilon+(\phi+n)^2    } e^{\textup{i}x(\phi+n)}.    $$
Let $A_{\phi,\epsilon}\in \mathcal{B}_{1}\big(L^2(\mathcal{I})\big)$  be defined as the rank one operator $A_{\phi,\epsilon}= |G_{\phi,\epsilon}\rangle \langle G_{\phi,\epsilon}|$.  By the general theory of Schr\"odinger operators with point potentials~\cite{Solve},  the difference between the resolvents of $H_{\phi}$ and $(P^{2})_{\phi}$  has the closed form 
\begin{align}\label{Flash}
  \frac{  1    }{ \epsilon+ (P^{2})_{\phi}      }  -\frac{1}{\epsilon+H_{\phi}}=\frac{  \alpha  }{ 1+\alpha G_{\phi,\epsilon}(0)  }A_{\phi,\epsilon }. 
\end{align}

Thus, I must bound
\begin{align}\label{Blithe}
 \Tr\Big[ \check{\rho}_{\lambda}\Big(H^{\frac{1}{2}}-\big(P^{2}  \big)^{\frac{1}{2}} \Big) \Big]= &  \Big\langle \widehat{\frak{h}}_{\phi}\Big| \frac{  1    }{ \epsilon+ \big(P^{2}\big)_{\phi}      }  -\frac{1}{\epsilon+H_{\phi}}\Big| \widehat{\frak{h}}_{\phi}\Big\rangle \nonumber \\   = & \frac{  \alpha  }{ 1+\alpha G_{\phi,\epsilon}(0)  }\big|\big\langle  G_{\phi,\epsilon} \big| \widehat{\frak{h}}_{\phi} \big\rangle\big|^{2} \nonumber \\   = &\frac{1}{4\pi^{2} } \frac{  \alpha  }{ 1+\alpha G_{\phi,\epsilon}(0)  }\Big| \sum_{n\in \Z}\frac{1}{\epsilon+(\phi+n)^2    }\int_{\mathcal{I} }dx \,\widehat{\frak{h}}_{\phi}(x)\frac{ e^{-\textup{i}x(\phi+n)} }{\sqrt{2\pi}}\Big|^{2} .
\end{align}
The Fourier coefficients of $\widehat{\frak{h}}_{\phi}\in L^{2}(\mathcal{I})$ have the form
\begin{align}\label{Phoebe}
\int_{\mathcal{I} }dx \,\widehat{\frak{h}}_{\phi}(x)\frac{ e^{-\textup{i}x(\phi+n)} }{\sqrt{2\pi}} =\frak{h}(\phi+n)=\frak{h}_{0}(\phi+n-\mathbf{p})  ,     
\end{align}
where in the above $\frak{h}$ and  $\frak{h}_{0}$ are evaluated in the momentum representation, and I have used that $\frak{h}:=e^{\textup{i}\mathbf{p}X}\frak{h}_{0}$.   With~(\ref{Phoebe}) the last line of~(\ref{Blithe}) is equal to 
\begin{align*}
 \frac{1}{4\pi^{2} } \frac{  \alpha  }{ 1+\alpha G_{\phi,\epsilon}(0)  }\Big| \sum_{n\in \Z}\frac{\frak{h}_{0}(\phi+n-\mathbf{p}) }{\epsilon+(\phi+n)^2    }\Big|^{2}  \leq &  \frac{\alpha }{4\pi^{2}  }\Big( \sup_{\phi \in \mathbb{T} }\sum_{n\in \Z} \frac{1}{\big(\epsilon+(\phi+n)^2 \big)^{2}   }    \Big)\Big(\sup_{\phi \in \mathbb{T} }\sum_{n\in \Z}\big| \frak{h}_{0}(\phi+n)  \big|^2 \Big)\\   \leq &  \frac{C }{ \epsilon^{2}  }\Big(\sup_{\phi \in \mathbb{T} }\sum_{n\in \Z}\big| \frak{h}_{0}(\phi+n)  \big|^2 \Big),
\end{align*}
where the first inequality follows by Cauchy-Schwarz and $G_{\phi,\epsilon}(0) >0$.   The second inequality above bounds the sum over $ \big(\epsilon+(\phi+n)^2 \big)^{-2}   $ by a constant multiple $C>0$ of $\frac{1}{\epsilon^{2}}$.    The summation on the last line is uniformly finite as a consequence of the assumption~(\ref{Gauss}) on $\frak{h}_{0}$ that  $\| X^{2} \frak{h}_{0}\|_{2}<\infty$.  To see this, notice that   for any $\phi \in \mathbb{T}$
\begin{align*}
 \| \frak{h}_{0}\|_{2}=&\Big(\int_{\R}dp\,\big| \frak{h}_{0}(p)\big|^{2}\Big)^{\frac{1}{2}} = \Big(\sum_{n\in \Z}\int_{[-\frac{1}{2},\frac{1}{2}) }d\phi' \,\big| \frak{h}_{0}(\phi+n+\phi')  \big|^{2}\Big)^{\frac{1}{2}}
\\  \geq &   \Big(\sum_{n\in \Z}\big| \frak{h}_{0}(\phi+n )  \big|^{2}\Big)^{\frac{1}{2}}-  \Big(\sum_{n\in \Z}\int_{[-\frac{1}{2},\frac{1}{2}) }d\phi' \,\Big| \int_{0}^{\phi'}d\phi'' \,\frak{h}_{0}'(\phi+n+\phi'')  \Big|^{2}\Big)^{\frac{1}{2}} \\  \geq &   \Big(\sum_{n\in \Z}\big| \frak{h}_{0}(\phi+n )  \big|^{2}\Big)^{\frac{1}{2}}- \| X  \frak{h}_{0}\|_{2}  ,
\end{align*}
where the first inequality uses calculus to write $\frak{h}_{0}(\phi+n+\phi') = \frak{h}_{0}(\phi+n)+\int_{0}^{\phi'}d\phi''\,\frak{h}_{0}'(\phi+n+\phi'')$ and applies the triangle inequality.   It follows that the supremum of $\sum_{n\in \Z}\big| \frak{h}_{0}(\phi+n )  \big|^{2}$ over $\phi \in \mathbb{T}$ is bounded by $( \| \frak{h}_{0}\|_{2}+\| X  \frak{h}_{0}\|_{2} )^{2}$.

\end{proof}

  Lemma~\ref{LemMomDist} states that the extend-zone scheme momentum for the state $\rho_{\lambda,\frac{T}{\lambda^{\gamma} } } \in \mathcal{B}_{1}(\mathcal{H})$  is concentrated ``near" the values $\pm\mathbf{p}$.

\begin{lemma}\label{LemMomDist}
Let $\gamma\in (1,2)$ and $\iota\in  (0, \frac{2-\gamma}{2})$.  For fixed $T>0$ there is a $C>0$ such that for all $\lambda<1$,
 $$ \int_{||p|-\mathbf{p}|  \geq \lambda^{\iota}\mathbf{p}  }\Big| \big[\rho_{\lambda,\frac{T}{\lambda^{\gamma} } } \big]^{(0)}_{\scriptscriptstyle{Q}}(p) \Big| \leq C\lambda^{2-\gamma-2\iota} .$$

\end{lemma}

\begin{proof} The integral $ \big[\rho_{\lambda,\frac{t}{\lambda^{\gamma} } } \big]^{(0)}_{\scriptscriptstyle{Q}}(p) $ over the domain $\big||p|-\mathbf{p}\big|\geq \lambda^{\iota}\mathbf{p} $ has the bound
\begin{align}\label{Tyrone}
 \int_{||p|-\mathbf{p}|\geq \lambda^{\iota}\mathbf{p}}\big[\rho_{\lambda,\frac{T}{\lambda^{\gamma} } } \big]^{(0)}_{\scriptscriptstyle{Q}}(p)\leq &\frac{\int_{\R} dp\,\big[\rho_{\lambda,\frac{T}{\lambda^{\gamma} } } \big]^{(0)}_{\scriptscriptstyle{Q}}(p)\big| E^{\frac{1}{2}}(p)-E^{\frac{1}{2}}(\mathbf{p})   \big|^{2}}{\inf_{||p|-\mathbf{p}|\geq \lambda^{\iota}   \mathbf{p} }\big| E^{\frac{1}{2}}(p)-E^{\frac{1}{2}}(\mathbf{p})   \big|^{2}   }\nonumber  \\   < & \frac{4}{\mathbf{p}^{2}\lambda^{2\iota}}  \int_{\R} dp\, \big[\rho_{\lambda,\frac{T}{\lambda^{\gamma} } } \big]^{(0)}_{\scriptscriptstyle{Q}}(p)\big| E^{\frac{1}{2}}(p)-E^{\frac{1}{2}}(\mathbf{p})   \big|^{2}.
\end{align}
The first inequality is Chebyshev's.  The second inequality holds for small $\lambda$ and  uses that $\big|E^{\frac{1}{2}}(p)- |p|\big|$ is bounded by Part (1) of Lem.~\ref{LemLittleEnergy}.  The analysis below shows that the integral on the second line of~(\ref{Tyrone}) is bounded by a constant multiple of $\lambda^{-\gamma}$ for $\lambda\ll 1$.  This would imply that $ \int_{||p|-\mathbf{p}|\leq \lambda^{\iota}\mathbf{p}}\big[\rho_{\lambda,\frac{T}{\lambda^{\gamma} } } \big]^{(0)}_{\scriptscriptstyle{Q}}(p) $ is bounded by a constant multiple of $\lambda^{2-\gamma-2\iota} $, which is the statement of the lemma.

Using the unraveling for the dynamical map $\Phi_{\lambda,t}:\mathcal{B}(\mathcal{H})$ from Lem.~\ref{LemLevyUnravel}, I have the following relations:
\begin{align}\label{Emma}
\int_{\R} dp\,\big[\rho_{\lambda,\frac{T}{\lambda^{\gamma} } } \big]^{(0)}_{\scriptscriptstyle{Q}}(p)\big| E^{\frac{1}{2}}(p)-E^{\frac{1}{2}}(\mathbf{p})   \big|^{2}= & \Tr\Big[ \Phi_{\lambda,\frac{T}{\lambda^{\gamma}}}(\check{\rho}_{\lambda}  )\Big(H^{\frac{1}{2}}-E^{\frac{1}{2}}(\mathbf{p})   \Big)^{2}    \Big] \nonumber  \\
=&\mathbb{E}\Big[ \Tr\Big[ \check{\rho}_{\lambda}  \Big( U_{\lambda,\frac{T}{\lambda^\gamma}}^{*}(\xi)H^{\frac{1}{2}}U_{\lambda,\frac{T}{\lambda^\gamma} }(\xi)-E^{\frac{1}{2}}(\mathbf{p})   \Big)^{2}    \Big] \Big]\nonumber 
\\
 \leq  & \Big| \mathbb{E}\Big[ \Tr\Big[ \check{\rho}_{\lambda} \Big( U_{\lambda,\frac{T}{\lambda^{\gamma} }}^{*}(\xi)H U_{\lambda,\frac{T}{\lambda^{\gamma}}}(\xi)-H   \Big)   \Big] \Big]\Big| \nonumber \\ & + 2E^{\frac{1}{2}}(\mathbf{p})\Big|\Tr\left[ \check{\rho}_{\lambda} \Big( H^{\frac{1}{2}}-E^{\frac{1}{2}}(\mathbf{p})   \Big)    \right]\Big|+\Big|\Tr\Big[ \check{\rho}_{\lambda} \Big( H-E(\mathbf{p})  \Big)    \Big]\Big|.
\end{align}
To obtain the inequality~(\ref{Emma}), I write
\begin{align}\label{Boast}
\Big( U_{\lambda,\frac{T}{\lambda^\gamma}}^{*}(\xi)H^{\frac{1}{2}}U_{\lambda,\frac{T}{\lambda^\gamma} }(\xi)-E^{\frac{1}{2}}(\mathbf{p})   \Big)^{2} = & \Big(U_{\lambda,\frac{T}{\lambda^\gamma}}^{*}(\xi)H U_{\lambda,\frac{T}{\lambda^\gamma}}(\xi)-H\Big) -2E^{\frac{1}{2}}(\mathbf{p})\Big(H^{\frac{1}{2}}-E^{\frac{1}{2}}(\mathbf{p})       \Big)\nonumber \\ &+\Big(H-E(\mathbf{p})\Big)-2E^{\frac{1}{2}}(\mathbf{p})\Big( U_{\lambda,\frac{T}{\lambda^\gamma}}^{*}(\xi)H^{\frac{1}{2}}U_{\lambda,\frac{T}{\lambda^\gamma} }(\xi)-H^{\frac{1}{2}} \Big),  
\end{align}
and use the triangle inequality for the first three terms.  The fourth term on the right side of~(\ref{Boast}) can be ignored since it makes a negative contribution by the inequality
 \begin{align}\label{Forsake}
  \mathbb{E}\Big[ \Tr\Big[\check{\rho}_{\lambda} \Big( U_{\lambda,\frac{T}{\lambda^\gamma}}^{*}(\xi)H^{\frac{1}{2}}U_{\lambda,\frac{T}{\lambda^\gamma} }(\xi)-H^{\frac{1}{2}}\Big) \Big]\Big]\geq 0.
  \end{align} 
The inequality~(\ref{Forsake}) holds since the process $U_{\lambda,t}^{*}(\xi)H^{\frac{1}{2}}U_{\lambda,t }(\xi)-H^{\frac{1}{2}}$ is an operator-valued submartingale by Part (1) of Prop.~\ref{LemSubMartBasic}.

 For the term in the third line of~(\ref{Emma}), the second equality below holds by Part (2) of Prop.~\ref{LemSubMartBasic}:
\begin{align*}
  \mathbb{E}\Big[ \Tr\Big[ \check{\rho}_{\lambda}  \Big( U_{\lambda,\frac{T}{\lambda^\gamma}}^{*}(\xi)H U_{\lambda,\frac{T}{\lambda^\gamma}}(\xi)-H  \Big)  \Big] \Big] = \Tr\Big[  \check{\rho}_{\lambda} \Big(\Phi_{\lambda,\frac{T}{\lambda^\gamma}}^{*}(H)-H    \Big)   \Big]= \frac{\sigma T }{ \lambda^{\gamma} }.
\end{align*}
  The two terms on the last line of~(\ref{Emma}) 
are both bounded by a constant multiple for $\lambda^{-1}$ for small $\lambda$, and I will show this only for the first  since the terms are similar.  The factor $\mathit{E}^{\frac{1}{2}}(\mathbf{p})\approx \mathbf{p}= \lambda^{-1}\mathbf{p}_{0}$ is $\mathit{O}(\lambda^{-1})$, and by the triangle inequality
\begin{align}\label{YoungHag}
\Big|\Tr\Big[ \check{\rho}_{\lambda}  \Big( H^{\frac{1}{2}}-E^{\frac{1}{2}}(\mathbf{p})   \Big)    \Big] \Big|\leq &\Big|\Tr\Big[ \check{\rho}_{\lambda}  \Big( H^{\frac{1}{2}}-(P^2)^{\frac{1}{2}}   \Big)    \Big] \Big|+ \Big|\Tr\Big[\check{\rho}_{\lambda}\Big((P^2)^{\frac{1}{2}} - \mathbf{p}\Big)\Big]\Big|\nonumber \\ &+ \sup_{p\in \R}\Big| |p| -E^{\frac{1}{2}}(p)\Big| .
\end{align} 
The  first and third terms on the right side of~(\ref{YoungHag}) are uniformly bounded  by    Parts (2) and (1) of Lem.~\ref{LemLittleEnergy}, respectively.  
The second term on the right side of~(\ref{YoungHag})  is uniformly bounded   for small $\lambda$ since 
\begin{align}\label{Commie}
\Big|\Tr\Big[\check{\rho}_{\lambda}\Big((P^2)^{\frac{1}{2}} - \mathbf{p} \Big)\Big]\Big|=\Big| \int_{\R}dp\, \big|\frak{h}_{0}(p)\big|^{2} \big|  \mathbf{p}+p   \big|-\mathbf{p} \Big| \leq \int_{\R}dp\, \big|\frak{h}_{0}(p)\big|^{2}|p| = \big\langle\frak{h}_{0} \big| |P| \frak{h}_{0}\big\rangle,
\end{align}
where the first equality uses that  $\check{\rho}_{\lambda}:= | \frak{h}\rangle \langle \frak{h}|$ for $\frak{h}:=e^{\textup{i}\mathbf{p} X}   \frak{h}_{0}$.   The right side of~(\ref{Commie}) is finite by our assumption~(\ref{Gauss}) that  $\|P^{2}\frak{h}_{0}\|_{2}<\infty $.

\end{proof}

\subsection{Proof of Lemma~\ref{StanToQuasi}}

The main ingredient for the proof of Lem.~\ref{StanToQuasi} is the bound for the difference between the Bloch functions $\psi_{p}$ and $\tilde{\psi}_{p}$ for $|p|\gg 1$ in Lem.~\ref{Ticks}. 
The upper bound  stated in Lem.~\ref{StanToQuasi} is weak for $p$ close to elements in $\frac{1 }{2}\Z$, and I apply Part (4) of Prop.~\ref{PropMoreFiber} to ensure that the momentum densities are bounded when contracted to the torus $\mathbb{T}$, and thus are not concentrated in the troublesome region around the lattice.   Other elements in the proof are the Cauchy-Schwarz-type inequalities of Sect.~\ref{SecDissect}.

\vspace{.4cm}

\begin{proof}[Proof of Lem.~\ref{StanToQuasi}]

By adding and subtracting $\big\langle p- k\big| \rho_{\lambda,\frac{T}{\lambda^{\gamma} } }\big| p+k \big\rangle_{\scriptscriptstyle{Q}}  $ and using the triangle inequality, I have the bound
\begin{align}\label{Bus}
\left\|   \big[\rho_{\lambda,\frac{T}{\lambda^{\gamma} } }  \big]^{(k)}-  \big[\rho_{\lambda,\frac{T}{\lambda^{\gamma} } } \big]^{(k)}_{\scriptscriptstyle{Q}}\right\|_{1}\leq & \int_{\R}dp\,\big| \big\langle p- k \big| \rho_{\lambda,\frac{T}{\lambda^{\gamma} } }\big(\big| p+ k \big\rangle- \big| p+k \big\rangle_{\scriptscriptstyle{Q}} \big) \big| \nonumber   \\
&+\int_{\R}dp\,\big| \big(\big\langle p-k\big| - \text{  }_{\scriptscriptstyle{Q}}\big\langle p-k \big|\big) \rho_{\lambda,\frac{T}{\lambda^{\gamma} } } \big| p+k \big\rangle_{\scriptscriptstyle{Q}} \big| .
\end{align}
The terms on the right side of~(\ref{Bus}) are similar, so I will treat only the first.   For $\phi\in \mathbb{T}, \kappa\in [-\frac{1}{4},\frac{1}{4})$ with 
$\phi=  p\,\text{mod}\,1$ and $\kappa=  k\,\text{mod}\,\frac{1}{2}$, translating from ket notation to Bloch functions yields 
\begin{align*}
 \big\langle p- k \big| \rho_{\lambda,\frac{T}{\lambda^{\gamma} } }\big(\big| p+k\big\rangle- \big| p+ k \big\rangle_{\scriptscriptstyle{Q}} \big) = \Big\langle \psi_{p-k} \Big| \ell_{\phi}^{(\kappa)}\big(\rho_{\lambda,\frac{T}{\lambda^{\gamma} } }\big) \Big(\psi_{p+k}-\widetilde{\psi}_{p+k} \Big)\Big\rangle . 
\end{align*} 
The first term on the right side of~(\ref{Bus}) is bounded by 
\begin{align}\label{AllTheDancers}
 \int_{\R}dp\,\Big| \Big\langle & \psi_{p- k } \Big| \ell_{\phi}^{(\kappa)}\big(\rho_{\lambda,\frac{T}{\lambda^{\gamma} } }\big) \Big(\psi_{p+k}-\widetilde{\psi}_{p+k} \Big)\Big\rangle  \Big|\nonumber \\  \leq & \int_{\R}dp\,\Big\langle \psi_{p-k}\Big|\,\ell_{\phi- \kappa }^{(0)}\big(\rho_{\lambda,\frac{T}{\lambda^{\gamma} } }\big) \psi_{p-k}\Big\rangle^{\frac{1}{2}}  \Big\langle \psi_{p+k}-\widetilde{\psi}_{p+k} \,\Big| \ell_{\phi+ \kappa }^{(0)}\big(\rho_{\lambda,\frac{T}{\lambda^{\gamma} } }\big) \Big|\psi_{p+k}-\widetilde{\psi}_{p+k}\Big\rangle^{\frac{1}{2}} \nonumber   \\  \leq & \Big(\int_{\R  }dp\,\Big\langle \psi_{p+k}-\widetilde{\psi}_{p+k} \,\Big| \ell_{\phi+ \kappa }^{(0)}\big(\rho_{\lambda,\frac{T}{\lambda^{\gamma} } }\big) \Big|\psi_{p+k}-\widetilde{\psi}_{p+k}\Big\rangle \Big)^{\frac{1}{2}}  .
\end{align}
The first inequality above  is by Part (1) of Prop.~\ref{MiscFiber}, and the second inequality uses  Cauchy-Schwarz  along with Part (1) of Prop.~\ref{PropFiberII} to    obtain 
\begin{align*}
\int_{\R}dp\,\Big\langle \psi_{p-k}\Big|\,\ell_{\phi- \kappa }^{(0)}\big(\rho_{\lambda,\frac{T}{\lambda^{\gamma} } }\big) \Big| \psi_{p-k}\Big\rangle=  \int_{\R}dp\, \big[\rho_{\lambda,\frac{T}{\lambda^{\gamma} } }  \big]^{(0)}\big(p-k\big) =\Tr\big[\rho_{\lambda,\frac{T}{\lambda^{\gamma} } }\big]=1.  
\end{align*}
To bound the bottom line of~(\ref{AllTheDancers}), I will treat the integrand separately for the domains $ |p|\in [\mathbf{p}-2\lambda^{\iota}\mathbf{p}  ,  \mathbf{p}+2\lambda^{\iota}\mathbf{p}]$ and  $ |p|\notin [\mathbf{p}-2\lambda^{\iota}\mathbf{p}  ,  \mathbf{p}+2\lambda^{\iota}\mathbf{p}]$ in (i) and (ii) below. \vspace{.3cm}
 
\noindent (i). For the domain $|p|\in [\mathbf{p}-2\lambda^{\iota}\mathbf{p}  ,  \mathbf{p}+2\lambda^{\iota}\mathbf{p}]$,
\begin{align}\label{Nike}
 \Big(\int_{|p|\in [\mathbf{p}-2\lambda^{\iota}\mathbf{p}  ,  \mathbf{p}+2\lambda^{\iota}\mathbf{p}]  }\,&\Big\langle \psi_{p+k}-\widetilde{\psi}_{p+k} \,\Big| \ell_{\phi+ \kappa}^{(0)}\big(\rho_{\lambda,\frac{T}{\lambda^{\gamma} } }\big)\Big|\psi_{p+k}-\widetilde{\psi}_{p+k}\Big\rangle \Big)^{\frac{1}{2}}\nonumber  \\   \leq &\Big(\sup_{\phi\in \mathbb{T}}\big\|\ell_{\phi}^{(0)}\big(\rho_{\lambda,\frac{T}{\lambda^{\gamma} } }\big) \big\|_{\infty}  \Big)^{\frac{1}{2}}  \Big( 8\lambda^{\iota}\mathbf{p}  \int_{\mathbb{T} }d\phi \, \sup_{\substack{p=\phi\,\textup{mod}\,1, \\ |p|\geq  \frac{1}{2}\mathbf{p} }  }  \big\|\psi_{p}-\widetilde{\psi}_{p}\big\|_{2}^{2} \Big)^{\frac{1}{2}},
\end{align}
 where I have bounded the number of $p\in [\mathbf{p}-2\lambda^{\iota}\mathbf{p}  ,  \mathbf{p}+2\lambda^{\iota}\mathbf{p}] $ with $p+k=\phi\,\textup{mod}\,1$ for a fixed $\phi\in \mathbb{T}$ by $8\mathbf{p}\lambda^{\iota}$. The left term on the second line of~(\ref{Nike}) is bounded independently of $T,\lambda >0$ since
 \begin{align}\label{Mule} \sup_{\phi\in \mathbb{T}}\big\|\ell_{\phi}^{(0)}\big(\rho_{\lambda,\frac{T}{\lambda^{\gamma} } }\big)\big\|_{\infty} \leq \sup_{\phi\in \mathbb{T}} \Tr\big[ \ell_{\phi}^{(0)}\big(\rho_{\lambda,\frac{T}{\lambda^{\gamma} } }\big)   \big] & \leq \sup_{\phi\in \mathbb{T}}\langle \check{\rho}_{\lambda}\rangle^{(0)}_{\phi}+ \frac{\varpi}{\mathcal{R}}\nonumber \\ &= \sup_{\phi\in \mathbb{T}}\sum_{n\in\Z}\big| \frak{h}_{0}(n+\phi)  \big|^{2}   + \frac{\varpi}{\mathcal{R}},
 \end{align}
 where the second inequality is by Part (4) of Prop.~\ref{PropMoreFiber}. The equality in~(\ref{Mule}) holds since $\check{\rho}_{\lambda}:=|\frak{h}\rangle \langle \frak{h}| $ for  $\frak{h}:=e^{\textup{i}\mathbf{p}X}\frak{h}_{0}$, and the second line of~(\ref{Mule}) is finite by the argument at the end of the proof of Lem.~\ref{LemLittleEnergy}.
 The right term on the second line of~(\ref{Nike}) is smaller than
\begin{align}\label{MuleII}
8\lambda^{\iota}\mathbf{p} \int_{\mathbb{T} }d\phi \, \sup_{\substack{p=\phi\,\textup{mod}\,1, \\ |p|\geq  \frac{1}{2}\mathbf{p} }  }  \big\|\psi_{p}-\widetilde{\psi}_{p}\big\|_{2}^{2} & \leq  8\lambda^{\iota}\mathbf{p} \int_{[-\frac{1}{4}, \frac{1}{4}] }d\theta \,\frac{C}{(1+|\frac{\mathbf{p}}{2}  \theta|)^2}\nonumber  \\  &\leq  8C\lambda^{\iota} \int_{[- \frac{\mathbf{p}}{8}, \frac{\mathbf{p}}{8}] }dy \frac{1}{(1+|y|)^2}\nonumber \\ & \leq      16C\lambda^{\iota}  ,  
\end{align}
where the first inequality is for some $C>0$ by Lem.~\ref{Ticks}.  The inequalities~(\ref{Mule}) and~(\ref{MuleII})  yield that the right side of~(\ref{Nike}) is $\mathit{O}(\lambda^{\iota})$.  \vspace{.4cm}

\noindent (ii).  As a preliminary, notice that the above analysis implies 
\begin{align*}
\int_{|p|\in [\mathbf{p}-\lambda^{\iota}\mathbf{p}  ,  \mathbf{p}+\lambda^{\iota}\mathbf{p}] }\Big| \big[\rho_{\lambda,\frac{T}{\lambda^{\gamma} } }  \big]^{(0)}(p)-  \big[\rho_{\lambda,\frac{T}{\lambda^{\gamma} } } \big]^{(0)}_{\scriptscriptstyle{Q}}(p)\Big|=\mathit{O}(\lambda^{\iota}).
\end{align*}
Moreover, since $\int_{\R}dp\,\big[\rho_{\lambda,\frac{T}{\lambda^{\gamma} } }  \big]^{(0)}(p)=\int_{\R}dp\, \big[\rho_{\lambda,\frac{T}{\lambda^{\gamma} } } \big]^{(0)}_{\scriptscriptstyle{Q}}(p)=1$ by Part (1) of Prop.~\ref{PropFiberII},
\begin{align}\label{Gyro}
\Big|\int_{|p|\notin [\mathbf{p}-\lambda^{\iota}\mathbf{p}  ,  \mathbf{p}+\lambda^{\iota}\mathbf{p}] } \big[\rho_{\lambda,\frac{T}{\lambda^{\gamma} } }  \big]^{(0)}(p)-  \int_{|p|\notin [\mathbf{p}-\lambda^{\iota}\mathbf{p}  ,  \mathbf{p}+\lambda^{\iota}\mathbf{p}] }\big[\rho_{\lambda,\frac{T}{\lambda^{\gamma} } } \big]^{(0)}_{\scriptscriptstyle{Q}}(p)\Big|=\mathit{O}(\lambda^{\iota}).
\end{align}

 For the integration over the domain $|p|\notin [\mathbf{p}-2\lambda^{\iota}\mathbf{p}  ,  \mathbf{p}+2\lambda^{\iota}\mathbf{p}]$, I have the following inequalities:
\begin{align*}
&\int_{|p|\notin [\mathbf{p}-2\lambda^{\iota}\mathbf{p}  ,  \mathbf{p}+2\lambda^{\iota}\mathbf{p}]  }\,\Big\langle \psi_{p+k}-\widetilde{\psi}_{p+k} \,\Big| \ell^{(0)}_{\phi}\big(\rho_{\lambda,\frac{T}{\lambda^{\gamma} } }\big)\Big|\psi_{p+k}-\widetilde{\psi}_{p+k}\Big\rangle\\ & \leq 2 \int_{|p|\notin [\mathbf{p}-2\lambda^{\iota}\mathbf{p}  ,  \mathbf{p}+2\lambda^{\iota}\mathbf{p}]  }\,\Big(\Big\langle \psi_{p+k}\,\Big|  \ell^{(0)}_{\phi}\big(\rho_{\lambda,\frac{T}{\lambda^{\gamma} } }\big)\Big|\psi_{p+k}\Big\rangle + \Big\langle \widetilde{\psi}_{p+k} \,\Big|  \ell^{(0)}_{\phi}\big(\rho_{\lambda,\frac{T}{\lambda^{\gamma} } }\big)\Big|\widetilde{\psi}_{p+k}\Big\rangle \Big)\\ & \leq  2 \int_{|p|\notin [\mathbf{p}-\lambda^{\iota}\mathbf{p}  ,  \mathbf{p}+\lambda^{\iota}\mathbf{p}]  }\,\Big(  \big[\rho_{\lambda,\frac{T}{\lambda^{\gamma} } }  \big]^{(0)}(p)+  \big[\rho_{\lambda,\frac{T}{\lambda^{\gamma} } } \big]^{(0)}_{\scriptscriptstyle{Q}}(p)\Big) \\ & \leq  4 \int_{|p|\notin [\mathbf{p}-\lambda^{\iota}\mathbf{p}  ,  \mathbf{p}+\lambda^{\iota}\mathbf{p}]  }\, \big[\rho_{\lambda,\frac{T}{\lambda^{\gamma} } } \big]^{(0)}_{\scriptscriptstyle{Q}}(p)+\mathit{O}(\lambda^{\iota})\\  & \leq C'\lambda^{2-\gamma-2\iota}+\mathit{O}(\lambda^{\iota})=\mathit{O}(\lambda^{\iota}).
\end{align*}
The  second inequality uses the definitions of $[\rho_{\lambda,\frac{T}{\lambda^{\gamma} } }  ]^{(0)}$, $[\rho_{\lambda,\frac{T}{\lambda^{\gamma} } }]^{(0)}_{\scriptscriptstyle{Q}}$ and the assumption $|k|\leq \lambda^{\iota}\mathbf{p}$. The third inequality follows from~(\ref{Gyro}), and the last inequality holds for some $C'>0$ by Lem.~\ref{LemMomDist}.  Finally, the order equality uses that $\iota=2-\gamma-2\iota$.

\end{proof}

\section{The adiabatic approximation}\label{SecFreidlin}

In this section I prove Thm.~\ref{ThmSemiClassical}.  The analysis in the proof of Thm.~\ref{ThmSemiClassical} is an extension of the analysis for the proof of~\cite[Thm.2.1]{Dispersion}.  The previous result only  characterized the limiting  autonomous dynamics for the diagonals of the time-evolved density matrices in the extended-zone scheme representation whereas the treatment here includes a region of off-diagonals. 

\subsection{Preliminary estimates for the adiabatic approximation  }

Recall that the function $\mathbf{n}:\R\rightarrow \Z$ is defined such that $\mathbf{n}(p)=2\big(p-\theta\big)$ for $\theta\in [-\frac{1}{4},\frac{1}{4})$ with $p=\theta\,\textup{mod}\,\frac{1}{2}$.  Given $p,v\in \R $, define the set $I(p,v)\subset \Z$ to be 
$$I(p,v):=\big\{ 0,-\mathbf{n}(p),-\mathbf{n}(p+v),\mathbf{n}(p) -\mathbf{n}(p+v)  \big\} . $$
The following technical lemma is from~\cite[Lem.2.2]{Dispersion}.

\begin{lemma}\label{BadTerms}
  There exists a $C>0$ such that the following inequalities hold:
\begin{enumerate}
\item For all $p,v\in \R$ with $|v|\leq \frac{1}{2}|p|$,
$$ \sum_{n\notin I(p,v)}|\kappa_{v}(p,n)|^{2}\leq \frac{C}{1+| p|^{2}}. $$

\item For all  $m,n\in \Z$ with $m\neq -\mathbf{n}(p)$, $n\neq 0$, and $\big|\frac{1}{2}m-n-p\big|\leq \frac{1}{2}|p| $, 
$$ \int_{-\frac{1}{4}  }^{\frac{1}{4} }d\theta\,\big|\kappa_{\theta+\frac{1}{2}m-n-p}(p,n)\big|^{2}  \leq  \frac{C }{1+\big|p\big|}.  $$

\end{enumerate}

\end{lemma}

The following lemma is similar to Part (4) of Prop.~\ref{PropMoreFiber}, and I will neglect the proof.  

\begin{lemma}\label{HorseMeister}
Let the maps $\Phi_{\lambda, \xi, t}$ and the times $t_{n}$ be defined as in Lem.~\ref{Trivial}.  The following inequality holds for all $\rho\in \mathcal{B}_{1}\big(L^{2}(\R)   \big)$:  
$$\big\|\big\langle \Phi_{\lambda, \xi, t_{n}}  (  \rho ) \big\rangle^{(\kappa)}\big\|_{\infty}\leq \delta_{0,n}\big\|\big\langle  \rho \big\rangle^{(\kappa)}\big\|_{\infty}+(1-\delta_{0,n})\frac{\varpi }{\mathcal{R}}.  $$
\end{lemma}

   The bounds from Lem.~\ref{BadTerms} will be applied in the proof of Lem.~\ref{Fructose}.     Define  the function $\mathbf{Q}:\R\rightarrow \{0,1\}$ as  $\mathbf{Q}(p)= 1-\sum_{n\in \Z} 1_{[\frac{1}{2}n -\lambda^{2}, \frac{1}{2}n +\lambda^{2}]}(p)$.  I introduce the factor  $ \mathbf{Q}(p)$ in the statement of Lem.~\ref{Fructose} and  the proof of Lem.~\ref{Grim} to ensure that $p+N $ and $p+ 2k+N $ live on the same energy band for $p\in \textup{Supp}(\mathbf{Q})$ and $|k|\leq \frac{1}{2}\lambda^{2}$.  Throughout the analysis of this section, the reader should remember that $ k$ for $|k|\leq \frac{1}{2}\lambda^{2}$ is negligible compared to the length $\frac{1}{2}$ between  neighboring momenta satisfying the Bragg condition.  In the proof of Lem.~\ref{Fructose}, I rely mainly on decay that arises from the term  $| E(p)-E(p+ 2k+N )|^{-1}$ for large $|N|$.  However, it can occur that $|N|\gg 1$ but the  energies $ E(p)$  and   $E(p+ 2k+N )$ are not far apart, in which case I use Lem.~\ref{BadTerms} to extract some additional decay from the sum $ \sum_{n\in \Z }\big|\kappa_{v}(p,n)\big|\,\big|\kappa_{v}\big(p+ 2k+N,n-N\big)\big|$.

\begin{lemma}\label{Fructose}
Let $|k|\leq \frac{1}{2}\lambda^{2}$ and $\rho \in \mathcal{B}_{1}( \mathcal{H})$ be  positive. There is a $C>0$ such that for $\lambda < 1$,
\begin{align*}
 \sum_{N\neq 0} \int_{\R} dp\, \mathbf{Q}(p) &\int_{\R}dv\, j(v)\Big| \big[\rho\big]^{(k+\frac{1}{2}N)}_{\scriptscriptstyle{Q}}\big(p+k+\frac{1}{2}N\big)\Big|    \\ & \times \frac{ \sum_{n\in \Z }\big|\kappa_{v}(p,n)\big|\,\big|\kappa_{v}\big(p+ 2k+N,n-N\big)\big|}{\big| E(p)-E(p+2 k+N ) \big| }\leq C\|\rho\|_{\mathbf{1}}.  
\end{align*}

\end{lemma}

\begin{proof}
By splitting the integration $\int_{\R}dp$ into parts  $|p+\frac{1}{2}N|\leq 1$ and $|p+\frac{1}{2}N|> 1$, I have the inequality
\begin{align}\label{Bouncer}
 \sum_{N\neq 0}  \int_{\R}& dp\, \mathbf{Q}(p) \int_{\R}dv\, j(v)\Big| \big[\rho\big]^{(k+\frac{1}{2}N)}_{\scriptscriptstyle{Q}}\big(p+k+\frac{1}{2}N\big)\Big|   \frac{ \sum_{n\in \Z }\big|\kappa_{v}(p,n)\big|\,\big|\kappa_{v}\big(p+2 k+N,n-N\big)\big|}{ \big| E(p)-E(p+2 k+N ) \big| }\nonumber \\  \leq & \Big(\frac{\mathcal{R}}{\inf_{n\in \mathbf{N}}g_{n}}   \Big)\sum_{N\neq 0} \int_{|p+\frac{1}{2}N|\leq 1} dp\,\int_{\R}dv\,\frac{j(v)}{\mathcal{R}} \Big| \big[\rho\big]^{(k+\frac{1}{2}N)}_{\scriptscriptstyle{Q}}\big(p+k+\frac{1}{2}N\big)\Big| \nonumber  \\ & \times \sum_{n\in \Z} \big|\kappa_{v}(p,n)\big|\, \big|\kappa_{v}\big(p+2 k+N,n-N\big)\big|+\|\rho\|_{\mathbf{1}}\sum_{N\neq 0}C_{N}   ,
\end{align}
where $g_{n}$ is the $n$th energy band gap, and the values $C_{N}>0$ are defined as
$$C_{N}=  \sup_{|p+\frac{1}{2}N|\geq 1 } \int_{\R}dv\, j(v)  \frac{\sum_{n\in \Z}\big|\kappa_{v}(p,n)\big|\,\big|\kappa_{v}\big(p+ 2k+N,n-N\big)\big|}{ |E(p)-E(p+ 2k+N) | }.$$
For the domain $|p+\frac{1}{2}N|\leq 1$, I have used that 
the momenta $p$ and $p+ 2k+N $ belong to different energy bands when $p\in \textup{Supp}(\mathbf{Q})$, $|k|\leq \frac{1}{2}\lambda^{2}$, and $N\neq 0$.  It follows that 
$\big| E(p)-E(p+ 2k+N ) \big| $ must be bounded from below by the infemum of the energy gaps $g_{n}$.     For the domain $|p+\frac{1}{2}N|\geq 1$, I have applied Holder's inequality and $\big\|\big[\rho\big]^{(k+\frac{1}{2}N)}_{\scriptscriptstyle{Q}}\big\|_{1}\leq \|\rho\|_{\mathbf{1}}$, where the latter follows by Part (3) of Prop.~\ref{PropFiberII}.  I will show that the first and second terms on the right side of~(\ref{Bouncer}) are bounded by multiples of  $  \|\rho\|_{\mathbf{1}}$ in parts (i) and (ii), respectively, below.  

\vspace{.4cm}
\noindent (i).\hspace{.2cm} For the first term on the right side of~(\ref{Bouncer}), 
 the integral has the bound 
\begin{align}\label{Urbane}
\sum_{N\neq 0} &\int_{|p+\frac{1}{2}N|\leq 1 } dp\,\int_{\R}dv\,\frac{j(v)}{\mathcal{R}} \Big| \big[\rho\big]^{(k+\frac{1}{2}N)}_{\scriptscriptstyle{Q}}\big(p+ k+\frac{1}{2}N\big)\Big|\nonumber \\ & \times\sum_{n\in \Z }\big|\kappa_{v}(p,n)\big|\,\big|\kappa_{v}\big(p+2 k+N,n-N\big)\big|  \nonumber \\  \leq  &\sum_{N\neq 0}\int_{|p+\frac{1}{2}N|\leq   1 } dp\,\Big(\frac{1}{2} \big[\rho\big]^{(0)}_{\scriptscriptstyle{Q}}\big(p\big)+ \frac{1}{2}\big[\rho\big]^{(0)}_{\scriptscriptstyle{Q}}\big(p+ 2k+N\big)\Big)\nonumber  \\ \leq &  2\int_{\R}dp\,\big[\rho\big]^{(0)}_{\scriptscriptstyle{Q}}(p)=2\|\rho\|_{\mathbf{1}}. 
\end{align}
For the first inequality above, I have used that $\int_{\R}dv\frac{j(v)}{\mathcal{R}}=1$ and applied  the Cauchy-Schwarz inequality to get
\begin{align}\label{Cauchy}
\sum_{n\in \Z}& \big|\kappa_{v}\big(p,n\big)\big|\,\big|\kappa_{v}\big(p+ 2k+N,n-N\big)\big| \nonumber   \\  &\leq  \Big(\sum_{n\in \Z}\big|\kappa_{v}\big(p,n\big)\big|^{2}\Big)^{\frac{1}{2}}\,\Big(\sum_{n\in \Z}\big|\kappa_{v}\big(p+2k+N,n-N\big)\big|^2 \Big)^{\frac{1}{2}}\leq 1. 
\end{align}
Also for the first inequality  in~(\ref{Urbane}), I have applied Part (3) of Prop.~\ref{MiscFiber} to $\big[\rho\big]^{(k+\frac{1}{2}N)}_{\scriptscriptstyle{Q}}$ in combination with the relation $|xy|\leq \frac{x^2}{2}+\frac{y^2}{2}$.

\vspace{.4cm}
\noindent (ii).\hspace{.2cm}   It is sufficient to prove that the sum of the $C_{N}$'s is finite and has a bound independent of $\lambda<1$.  I will show that the $C_{N}$'s decay on the order of $|N|^{-\frac{3}{2}}$.   A single $C_{N}$ can be bounded  independently of $\lambda<1$ by the same reasoning as in (i).   The  difference $|E(p)-E(p+2 k+N)|$  becomes large for large $|N|\gg 1$ except when $p+ 2k+N$ is close to $-p$.  By the restrictions $|p+\frac{1}{2}N|\geq 1$ and $|k|\leq \frac{1}{2}\lambda^{2} $, the momenta $p$ and $p+2k+N$ can not lie on the same or neighboring energy bands.  Thus the absolute value of the difference between the energies $E(p)$ and $E(p+2 k+N)$ must be at least the length $L_{|N|}$  for $L_{m}$ defined by 
$$\hspace{4cm} L_{m}:=  E\big(\frac{m}{2}\big)-E\big(\frac{m-1}{2}\big)=\frac{2m-1}{4} , \hspace{1.5cm} m\in \mathbb{N}.         $$
By the same reasoning, if $|p|\wedge |p+ 2k+N|\leq \frac{1}{4}|N|$, then the momenta $p$ and $p+ k+N$ must be separated by the energy bands with band index between $\frac{1}{4}|N|$ and $\frac{3}{4}|N|$:
\begin{align}\label{Joad}
\Big|E\big(p\big)-E\big(p+ 2k+N\big)\Big|\geq \sum_{ \frac{1}{4}|N|< m< \frac{3}{4}|N|}L_{m} \propto N^2,
\end{align}
 where the  asymptotic proportion is for $|N|\gg 1$.   
  
By the above remarks, $C_{N}< \frac{1}{ |N|}C^{\prime}_{N}+C_{N}^{\prime \prime}$, where $C^{\prime}_{N}$ and $C^{\prime \prime}_{N}$ are defined as
\begin{eqnarray*}  
  C_{N}^{\prime}&:=& \sup_{\substack{|p+\frac{1}{2}N|\geq 1, \\  |p|\wedge |p+2 k+N|\geq \frac{1}{4}| N| }  } \int_{\R}dv \,\frac{j(v)}{\mathcal{R}}  \sum_{n\in \Z}\big|\kappa_{v}\big(p,n\big)\big|\,\big|\kappa_{v}\big(p+2 k+N,n-N\big)\big|,  \\
   C_{N}^{\prime \prime}&:=& \sup_{ |p|\wedge |p+ 2k+N|\leq \frac{1}{4}|N|   } \int_{\R}dv\, j(v)  \frac{\sum_{n\in \Z}\big|\kappa_{v}(p,n)\big|\,\big|\kappa_{v}\big(p+2 k+N,n-N\big)\big|}{\big| E(p)-E(p+2k+N) \big| } .
  \end{eqnarray*}
By the same reasoning as in (i), I have the  inequality below
$$C_{N}^{\prime \prime} \leq  \sup_{ |p|\wedge |p+2 k+N|\leq \frac{1}{4}|N|   } \frac{ \mathcal{R}  }{\big| E(p)-E(p+ 2k+N) \big| }=\mathit{O}\big(N^{-2}) . $$
The order equality follows from~(\ref{Joad}). Thus, the $C_{N}^{\prime \prime}$'s decay quadratically and are  summable.   In the analysis below, I  show that the $C_{N}'$'s have order $\mathit{O}(|N|^{-\frac{1}{2}})$, which implies that the $C_{N}$'s are summable.

 Bounding the $C_{N}^{\prime}$'s is trickier than the $C_{N}^{\prime \prime}$'s, since I depend on some decay for large $|N|$ arising from the sum of the terms $|\kappa_{v}(p,n)|\,|\kappa_{v}(p+2 k+N,n-N)|$, and there are various cases in which   $|\kappa_{v}(p,n)|$ and $|\kappa_{v}(p+2 k+N,n-N)|$ may not both be small. As a preliminary, I will partition the integration over $v\in \R$ into the sets $|v|>\frac{1}{8} |N|$ and $|v|\leq\frac{1}{8} |N|$.  For the domain $|v|>\frac{1}{8} |N|$ there is quadratic decay since  Chebyshev's inequality and~(\ref{Cauchy}) imply that 
\begin{align*}\int_{|v|> \frac{1}{8} |N|}dv\, \frac{j(v)}{\mathcal{R}}  \sum_{n\in \Z}\big|\kappa_{v}(p,n)\big|\,\big|\kappa_{v}(p+ 2k+N,n-N)\big| & \leq \int_{|v| > \frac{1}{8}N}dv\,\frac{j(v)}{\mathcal{R}}\\ & \leq  \frac{64\sigma  }{ \mathcal{R} }  N^{-2},
 \end{align*}
 where $\sigma=\int_{\R}dv\, j(v)v^{2}$.  For the domain $|v|\leq  \frac{1}{8} |N|$, I will rely on the results from Lem.~\ref{BadTerms}.  Given $p,v\in \R$, the Cauchy-Schwarz inequality and~(\ref{Cauchy})  yield that
\begin{align}\label{Rabble}
\sum_{n\in \Z }\big| & \kappa_{v}\big(p,n\big)\big|\,\big|\kappa_{v}\big(p+ 2k+N,n-N\big)\big| \nonumber  \\  \leq & \sum_{ \substack{ n\in I(p,v), \\  n\in I(p+ 2k+N,v)+N}}\big|\kappa_{v}\big(p,n\big)\big|\,\big|\kappa_{v}\big(p+ 2k+N,n-N\big)\big| \nonumber  \\ &+\Big(\sum_{n\notin I(p,v)}\big|\kappa_{v}\big(p,n\big)\big|^{2}\Big)^{\frac{1}{2}}+\Big(\sum_{n\notin I(p+ 2k+N,v)+N}\big|\kappa_{v}\big(p+ 2k+N,n\big)\big|^{2}\Big)^{\frac{1}{2}}.      
\end{align}
Under the constraints $|p|\wedge |p+2 k+N|\geq \frac{1}{4}|N|$ and $|v|\leq  \frac{1}{8} |N|$, Part (1) of Lem.~\ref{BadTerms} implies that the  terms on the bottom line of~(\ref{Rabble}) are bounded by multiples of  $|p|^{-1}\leq 4|N|^{-1}$ and $|p+2 k+N|^{-1}\leq 4|N|^{-1}$, respectively.  Since the total weight of the integration $\int_{|v|\leq  \frac{1}{8} |N|} dv\, \frac{j(v)}{\mathcal{R}}$ is less than one,  these terms make contributions to $C_{N}^{\prime}$ that vanish with order $\mathit{O}(|N|^{-\frac{1}{2}})$.   

The final task is to bound the sum on the second line of~(\ref{Rabble}).  Let $p':=p+ 2k+N$.  Note that  $p'\approx p+N $ since $| k|\leq \frac{1}{2}\lambda^{2}$.  The inequalities $|p|\wedge |p'|\geq \frac{1}{4}|N|$ and $|v| \leq \frac{1}{8}|N|$ imply that the set $  (I(p',v)+N )\cup I(p,v)$ must be empty unless the momenta $p$ and  $p'$ have opposite signs. If $p$ and  $p'$  have opposite signs,  the matching possibilities for elements in  $I(p,v)$  and $I(p',v)+N $ are those in the same rows below: 
\begin{center}
\begin{tabular}{|p{3cm}|p{4.5cm}|}
\hline 
\hspace{1.1cm}$I(p,v)$   &\hspace{1.3cm} $I(p',v)+N $ \\ 
\hline
\hline
\hspace{.05cm}$0, \mathbf{n}(p)-\mathbf{n}(p+v) $ &  $ -\mathbf{n}(p' )+N, -\mathbf{n}(p'+v )+N $\\
\hline
 $- \mathbf{n}(p), -\mathbf{n}(p+v)$ &  $ N, -\mathbf{n}(p' ) +\mathbf{n}(p'+v )+N$  \\
\hline 
\end{tabular}

\end{center}
However, the inequalities $|p+\frac{1}{2}N|\geq 1$ and $| k|\leq \frac{1}{2}\lambda^{2}$ leave only the following possibilities:
 \begin{center}
\begin{tabular}{|p{2.8cm}|p{3.8cm}|}
\hline
\hspace{.9cm}$I(p,v)$   &\hspace{.8cm} $I(p',v)+N $   \\
\hline \hline
\hspace{1.2cm} $0 $ & \hspace{.4cm} $ -\mathbf{n}(p'+v )+N $\\
\hline
\hspace{.1cm}$ \mathbf{n}(p)-\mathbf{n}(p+v) $ &  \hspace{.7cm} $ -\mathbf{n}(p' )+N $\\
\hline
\hspace{.7cm} $- \mathbf{n}(p)$ &  $ -\mathbf{n}(p' ) +\mathbf{n}(p'+v )+N$  \\
\hline
\hspace{.4cm} $-\mathbf{n}(p+v)$ & \hspace{1.5cm} $ N$  \\
 \hline
\end{tabular}

\end{center}
For each case of $n\in (I(p',v)+N )\cup I(p,v)$, either $n\neq 0,-\mathbf{n}(p)$ or $n':=n-N  $ satisfies $n'\neq 0,-\mathbf{n}(p')$.  The cases listed above are similar, so I will take  $n=N=-\mathbf{n}(p+v)$:    
\begin{align}\label{ShizyMay}
\sup_{\substack{|p+\frac{1}{2}N|\geq  1, \\  |p|\wedge |p+2 k+N|\geq \frac{1}{4}| N| }  } &\int_{|v|\leq \frac{1}{8}|N| }dv\,\frac{j(v)}{\mathcal{R}}\, \big|\kappa_{v}\big(p,N\big)\big|\,\big|\kappa_{v}\big(p+2 k+N,0\big)\big| \chi\big(N=-\mathbf{n}(p+v)   \big)\nonumber \\ &\leq    \frac{\varpi}{\mathcal{R}2^{\frac{1}{2}} } \sup_{  |p| \geq  \frac{1}{4}| N|   } \Big(  \int_{-\frac{1}{4} }^{ \frac{1}{4} }d\theta \, \big|\kappa_{\theta+\frac{1}{2}N-p  }\big(p,N\big)\big|^2      \Big)^{\frac{1}{2}}\nonumber  \\ & \leq\frac{c}{|N|^{\frac{1}{2}}},
\end{align}
where the second inequality holds for some $c>0$  by Part (2) of Lem.~\ref{BadTerms}.  In the expression on the first line of~(\ref{ShizyMay}), the integrand has support over the set $v\in N-p+[-\frac{1}{4},\frac{1}{4}]$ due to the factor $\chi(N=-\mathbf{n}(p+v)  )$.  In the fist inequality of~(\ref{ShizyMay}), I have used that $j(v)\leq \varpi $ by assumption (2) of List~\ref{Assumptions}, the Cauchy-Schwarz inequality, and that $\big|\kappa_{v}\big(p+2k+N,0\big)\big|\leq 1$.

\end{proof}

The proof of Lem.~\ref{Grim} proceeds by subtracting-off small parts from the expressions
$$\int_{s_{1}}^{s_{2}}dr \,   \mathcal{R}^{-1}U_{\lambda, -r}^{(k)}\big[\Psi(e^{-\frac{\textup{i}r}{\lambda^{\varrho}}H}\rho e^{\frac{\textup{i}r}{\lambda^{\varrho}}rH})\big]^{(k)}_{\scriptscriptstyle{Q}}\quad \text{and} \quad \int_{s_{1}}^{s_{2}}dr \,  U_{\lambda, -r}^{(k)}T_{k} U_{\lambda, r}^{(k)}[\rho  ]^{(k)}_{\scriptscriptstyle{Q}}$$
such that the difference between the remaining expressions can be bounded by an application of Lem.~\ref{Fructose}.  The parts removed from the expressions are associated with momenta near the lattice $\frac{1}{2}\Z$, as usual, and also momenta that are ``too high".  For technical purposes, capping the momentum is necessary to maintain that the difference of energies $|E(p-k )-E(p+k)|$ is small compared to $\frac{\alpha}{\pi}$, which is the scale for the gaps between the energy bands; recall $\lim_{N\rightarrow \infty} g_{N}=\frac{\alpha}{\pi}$.  Although assuming that $p\in \textup{Supp}(\mathbf{Q})$ guarantees  $p\pm k $ are on the same energy band, there will still be linear growth $|E(p- k )-E(p+ k)|\approx  4| k|\,| p|  $ for  high momenta $|p|\gg 1$ bounded away from the lattice $\frac{1}{2}\Z$.  The linear rate of growth,  $4| k|$, is slow for $\lambda\ll 1 $ under the constraint  $|k|\leq \frac{1}{2}\lambda^{2}$.

Recall from Sect.~\ref{SecPseudoPoisson} that the  operator  $T_{k}:L^{1}(\R)$ is defined to have kernel
$T_{k}(p,p')= \mathcal{R}^{-1}J_{k}(p,p') $ and $U_{\lambda, t}^{(k)}:L^{1}(\R)$  acts as multiplication by the function $U_{\lambda, t}^{(k)}(p)=e^{-\frac{\textup{i}t}{\lambda^{\varrho}}\left(E(p -k)-E(p +k)\right)    }$.

\begin{lemma}\label{Grim}
Let $\rho \in \mathcal{B}_{1}(\mathcal{H}    )$ be positive.   There is  a $C>0$ such that for all $0<\lambda<1$, $|k|\leq \frac{1}{2}\lambda^{2}$, and $s_{1}\leq s_{2}$,
\begin{align*}
  \Big\| \int_{s_{1}}^{s_{2}}dr\,\Big(  &\mathcal{R}^{-1} U_{\lambda, -r}^{(k)}\big[\Psi(e^{-\frac{\textup{i}r}{\lambda^{\varrho}} H}\rho e^{\frac{\textup{i}r}{\lambda^{\varrho}} H})\big]^{(k)}_{\scriptscriptstyle{Q}} -U_{\lambda, -r}^{(k)}T_{k} U_{\lambda, r}^{(k)}[\rho]^{(k)}_{\scriptscriptstyle{Q}} \Big)  \Big\|_{1}\\ &\leq C(s_{2}-s_{1})\Big( \lambda^{3+\frac{\gamma}{2} }\Tr[H\rho] +\lambda^{2}\|\langle \rho \rangle^{(0)}\|_{\infty}+ \lambda^{2}\|\rho\|_{\mathbf{1}}  \Big)+C\lambda^{\varrho }\|\rho\|_{\mathbf{1}}.
 \end{align*}

\end{lemma}

\begin{proof}

Let $\mathbf{Q}$, $\mathbf{Q}'$, $\mathbf{Q}''$ be the projections on $L^{2}(\R)$, or alternatively  $L^{1}(\R)$, that act as multiplication by the functions
 $$\mathbf{Q}(p)= 1-\sum_{n\in \Z} 1_{[\frac{1}{2}n -\lambda^{2}   , \frac{1}{2}n +\lambda^{2}  ] }(p), \quad  \mathbf{Q}'(p)=1_{|p|\leq 2\lambda^{-\frac{3}{2}-\frac{\gamma}{4}  }   },\quad \mathbf{Q}''(p)=\mathbf{Q}(p)\mathbf{Q}'(p).   $$
   Also, denote $\widetilde{\rho}=\rho-(I-\mathbf{Q})\rho (I-\mathbf{Q})  $.   There is a $C>0$ such that  the following inequalities hold for all $\lambda<1$, $|k|\leq \frac{1}{2} \lambda^{2}$, and  $\rho\in \mathcal{B}_{1}(\mathcal{H}  )$:   

\begin{enumerate}[(i).]

\item  
 \begin{align*}
 \Big\| \int_{s_{1}}^{s_{2}}dr\,  &  \mathcal{R}^{-1}U_{\lambda, -r}^{(k)}\big[\Psi(e^{-\frac{\textup{i}r}{\lambda^{\varrho}}H}\rho e^{\frac{\textup{i}r}{\lambda^{\varrho}}H})\big]^{(k)}_{\scriptscriptstyle{Q}} - \mathbf{Q}''\int_{s_{1}}^{s_{2}}dr\,\mathcal{R}^{-1}U_{\lambda, -r}^{(k)}\big[\Psi(e^{-\frac{\textup{i}r}{\lambda^{\varrho}} H}\widetilde{\rho} e^{\frac{\textup{i}r}{\lambda^{\varrho}}  H})\big]^{(k)}_{\scriptscriptstyle{Q}}  \Big\|_{1}\\ & \leq C(s_{2}-s_{1})\Big( \lambda^{3+\frac{\gamma}{2}}\Tr[H\rho] +\lambda^{2}\|\langle \rho \rangle^{(0)}\|_{\infty}+ \lambda^{2}\|\rho\|_{\mathbf{1}}  \Big),
 \end{align*}

\item 

\begin{align*}
\Big\| \int_{s_{1}}^{s_{2}}dr\,  & U_{\lambda, -r}^{(k)}T_{k} U_{\lambda, r}^{(k)}[\rho]^{(k)}_{\scriptscriptstyle{Q}}  -\mathbf{Q}''\int_{s_{1}}^{s_{2}}dr\, U_{\lambda, -r}^{(k)}T_{k} U_{\lambda, r}^{(k)}[\widetilde{\rho}]^{(k)}_{\scriptscriptstyle{Q}}  \Big)  \Big\|_{1} \\ & \leq C(s_{2}-s_{1})\Big(\lambda^{3+\frac{\gamma}{2}}\Tr[H\rho] +\lambda^{2}\|\langle \rho \rangle^{(0)}\|_{\infty}+ \lambda^{2}\|\rho\|_{\mathbf{1}}  \Big),
\end{align*}

\item  

\begin{align*}
  \Big\| \mathbf{Q}'' \int_{s_{1}}^{s_{2}}dr\,\Big(  &  \mathcal{R}^{-1}U_{\lambda, -r}^{(k)}\big[\Psi(e^{-\frac{\textup{i}r}{\lambda^{\varrho}}H}\widetilde{\rho} e^{\frac{\textup{i}r}{\lambda^{\varrho}}H})\big]^{(k)}_{\scriptscriptstyle{Q}} -U_{\lambda, -r}^{(k)}T_{k} U_{\lambda, r}^{(k)}[\widetilde{\rho}]^{(k)}_{\scriptscriptstyle{Q}} \Big)  \Big\|_{1}\leq C\lambda^{\varrho}\|\rho\|_{\mathbf{1}}.
\end{align*}

\end{enumerate}

\vspace{.4cm}

\noindent (i).\hspace{.2cm} By the triangle inequality, $\mathbf{Q}''(p)\leq 1$, and $1-\mathbf{Q}''(p)\leq 1-\mathbf{Q}(p)+1-\mathbf{Q}'(p)$, the left side of (i) is smaller than
 \begin{align}\label{LikeARat}
 \Big\| (I-&\mathbf{Q}') \int_{s_{1}}^{s_{2}}dr \,   \mathcal{R}^{-1}U_{\lambda, -r}^{(k)}\big[\Psi(e^{-\frac{\textup{i}r}{\lambda^{\varrho}}H}\rho e^{\frac{\textup{i}r}{\lambda^{\varrho}}H})\big]^{(k)}_{\scriptscriptstyle{Q}}\Big\|_{1}+\Big\|(I-\mathbf{Q}) \int_{s_{1}}^{s_{2}}dr   \, \mathcal{R}^{-1}U_{\lambda, -r}^{(k)}\big[\Psi(e^{-\frac{\textup{i}r}{\lambda^{\varrho}}H}\rho e^{\frac{\textup{i}r}{\lambda^{\varrho}}H})\big]^{(k)}_{\scriptscriptstyle{Q}}\Big\|_{1}\nonumber \\ &+ \Big\| \int_{s_{1}}^{s_{2}}dr   \, \mathcal{R}^{-1}U_{\lambda, -r}^{(k)}\big[\Psi(e^{-\frac{\textup{i}r}{\lambda^{\varrho}}H}(I-\mathbf{Q})\rho (I-\mathbf{Q}) e^{\frac{\textup{i}r}{\lambda^{\varrho}}H})\big]^{(k)}_{\scriptscriptstyle{Q}}\Big\|_{1}\nonumber  \\  \leq  & \mathcal{R}^{-1}\int_{s_{1}}^{s_{2}}dr \, \Big\| (I-\mathbf{Q}')  \big[\Psi(e^{-\frac{\textup{i}r}{\lambda^{\varrho}}H}\rho e^{\frac{\textup{i}r}{\lambda^{\varrho}}H})\big]^{(k)}_{\scriptscriptstyle{Q}}\Big\|_{1} + \mathcal{R}^{-1}\int_{s_{1}}^{s_{2}}dr\, \Big\|(I-\mathbf{Q}) \big[\Psi(e^{-\frac{\textup{i}r}{\lambda^{\varrho}}H}\rho e^{\frac{\textup{i}r}{\lambda^{\varrho}}H})\big]^{(k)}_{\scriptscriptstyle{Q}}\Big\|_{1}\nonumber \\ &+ (s_{2}-s_{1})\big\| (I-\mathbf{Q})\rho (I-\mathbf{Q}) \big\|_{\mathbf{1}}   .
 \end{align}
The inequality above uses that $U_{\lambda, -r}^{(k)}$ is a multiplication operator with multiplication function bounded by one, i.e., $|U_{\lambda, -r}^{(k)}(p)|\leq 1$. Also, for the last term, I have applied Part (3) of Prop.~\ref{PropFiberII} to get the inequality below:
\begin{align*}
  \Big\| \big[\Psi(e^{-\frac{\textup{i}r}{\lambda^{\varrho}}H}(I-\mathbf{Q})\rho (I-\mathbf{Q}) e^{\frac{\textup{i}r}{\lambda^{\varrho}}H})\big]^{(k)}_{\scriptscriptstyle{Q}}\Big\|_{1} & \leq \big\| \Psi(e^{-\frac{\textup{i}r}{\lambda^{\varrho}}H}(I-\mathbf{Q})\rho (I-\mathbf{Q}) e^{\frac{\textup{i}r}{\lambda^{\varrho}}H}  \big)  \big\|_{\mathbf{1}}\\ &= \mathcal{R}\| (I-\mathbf{Q})\rho (I-\mathbf{Q})\|_{\mathbf{1}}. 
  \end{align*}

For the first term on the right side of~(\ref{LikeARat}),  
\begin{align}\label{Bashir}
\frac{1}{\mathcal{R}}\int_{s_{1}}^{s_{2}}dr\, \Big\| (I-\mathbf{Q}')  \big[\Psi\big(e^{-\frac{\textup{i}r}{\lambda^{\varrho}}H}\rho e^{\frac{\textup{i}r}{\lambda^{\varrho}}H}\big)\big]^{(k)}_{\scriptscriptstyle{Q}}\Big\|_{1}& \leq \frac{1}{ \mathcal{R} }\int_{s_{1}}^{s_{2}}dr\, \int_{|p|>\lambda^{-\frac{3}{2}-\frac{\gamma}{4 }}  } \big[\Psi\big(e^{-\frac{\textup{i}r}{\lambda^{\varrho}}H}\rho e^{\frac{\textup{i}r}{\lambda^{\varrho}}H}\big)\big]^{(0)}_{\scriptscriptstyle{Q}}(p )\nonumber  \\  & \leq \frac{\lambda^{ 3+\frac{\gamma}{2}}  }{ \mathcal{R}}\int_{s_{1}}^{s_{2}}dr\,\int_{\R }dp E(p)\big[\Psi\big(e^{-\frac{\textup{i}r}{\lambda^{\varrho}}H}\rho e^{\frac{\textup{i}r}{\lambda^{\varrho}}H}\big)\big]^{(0)}_{\scriptscriptstyle{Q}}(p )\nonumber \\ &=  \frac{\lambda^{ 3+\frac{\gamma}{2} } }{ \mathcal{R}  }\int_{s_{1}}^{s_{2}}dr\,\Tr\big[ H \Psi\big( e^{-\frac{\textup{i}r}{\lambda^{\varrho}}H}\rho e^{\frac{\textup{i}r}{\lambda^{\varrho}}H} \big)\big]\nonumber \\ &=\lambda^{ 3+\frac{\gamma}{2} }(s_{2}-s_{1})\Big(\Tr[H\rho]+\frac{\sigma}{\mathcal{R}}\Tr[\rho]\Big).
\end{align}
The first inequality above is by Part (3) of Prop~\ref{MiscFiber} and holds for $\lambda<1$.  The second inequality in~(\ref{Bashir}) is Chebyshev's combined with $E(p)\geq p^{2}$, and  the second equality is by the explicit form $\Psi^{*}(H)=\mathcal{R}H+\sigma I$.

To bound the second term on the right side of~(\ref{LikeARat}), notice that by Part (3) of Prop.~\ref{MiscFiber} and the inequality $|2xy| \leq x^{2}+y^{2}$,
\begin{align}\label{Jujubee}
\int_{s_{1}}^{s_{2}}dr\,\Big\|(I-\mathbf{Q}) \big[\Psi\big(e^{-\frac{\textup{i}r}{\lambda^{\varrho}}H}\rho e^{\frac{\textup{i}r}{\lambda^{\varrho}}H}\big)\big]^{(k)}_{\scriptscriptstyle{Q}}\Big\|_{1}\leq  & \frac{1}{2 }\int_{s_{1}}^{s_{2}}dr\,\int_{\R}dp\,\big(1-\mathbf{Q}(p)\big) \big[\Psi\big(e^{-\frac{\textup{i}r}{\lambda^{\varrho}}H}\rho e^{\frac{\textup{i}r}{\lambda^{\varrho}}H}\big)\big]^{(0)}_{\scriptscriptstyle{Q}}\big(p-k\big)\nonumber \\ &+\frac{1}{2 }\int_{s_{1}}^{s_{2}}dr\,\int_{\R}dp\big(1-\mathbf{Q}(p)\big) \big[\Psi\big(e^{-\frac{\textup{i}r}{\lambda^{\varrho}}H}\rho e^{\frac{\textup{i}r}{\lambda^{\varrho}}H}\big)\big]^{(0)}_{\scriptscriptstyle{Q}}\big(p+k\big)
\nonumber \\ \leq &  6\lambda^{2} \varpi (s_{2}-s_{1})\|\rho \|_{\mathbf{1}} .
\end{align}
To see the second inequality above,   notice that the first is bounded by
\begin{align*}
 \frac{1}{2}\int_{s_{1}}^{s_{2}}dr &\int_{ [ -\lambda^{2}  ,\lambda^{2}  ]  \cup [\frac{1}{2} -\lambda^{2}    , \frac{1}{2} ] \cup [-\frac{1}{2}   ,- \frac{1}{2}+\lambda^{2} ]  }d\phi \,\Big\langle  \Psi\big(e^{-\frac{\textup{i}r}{\lambda^{\varrho}}H}\rho e^{\frac{\textup{i}r}{\lambda^{\varrho}}H}\big)  \Big\rangle_{\phi- \kappa }^{(0)}\\ &\leq  2\lambda^{2}\int_{s_{1}}^{s_{2}}dr\, \Big\|  \Big\langle  \Psi\big(e^{-\frac{\textup{i}r}{\lambda^{\varrho}}H}\rho e^{\frac{\textup{i}r}{\lambda^{\varrho}} H}\big)  \Big\rangle^{(0)}   \Big\|_{\infty}\\ &\leq 2\lambda^{2}\varpi (s_{2}-s_{1}) \|\rho\|_{\mathbf{1}},
\end{align*}
where $\kappa\in [-\frac{1}{4},\frac{1}{4})$ with $\kappa=k \mod \frac{1}{2}$.  The second inequality above is by Part (3) of Prop.~\ref{PropMoreFiber} and the fact  that the trace norm is invariant of unitary conjugation.  The second term after the first inequality in~(\ref{Jujubee}) has the same bound.

For the third term on the right side of~(\ref{LikeARat}),
\begin{align*}
  \big\| (I-\mathbf{Q})\rho (I-\mathbf{Q}) \big\|_{\mathbf{1}} = & \sum_{n\in\Z}\int_{ [\frac{1}{2}n -\lambda^{2}   , \frac{1}{2}n +\lambda^{2}]  }dp\, [\rho]^{(0)}_{\scriptscriptstyle{Q}}(p)\\ = & \int_{ [ -\lambda^{2}   , \lambda^{2} ]  \cup [\frac{1}{2} -\lambda^{2}   , \frac{1}{2} ] \cup [-\frac{1}{2}   ,- \frac{1}{2}+\lambda^{2} ]  }d\phi \,\langle \rho\rangle_{\phi}^{(0)}\\ \leq  & 4\lambda^{2} \big\| \langle \rho\rangle^{(0)}\big\|_{\infty} .
  \end{align*}

\vspace{.4cm}

\noindent (ii).\hspace{.2cm} This follows by similar analysis as for  (i).

\vspace{.4cm}

\noindent (iii).\hspace{.2cm}  By an evaluation of the integral, I have the following equality: 
\begin{align}\label{Hoser}
\mathbf{Q}^{\prime \prime}(p)\int_{s_{1}}^{s_{2}}dr\, &\Big(    \mathcal{R}^{-1}U_{\lambda, -r}^{(k)}\big[\Psi(e^{-\frac{\textup{i}r}{\lambda^{\varrho}}H}\rho e^{\frac{\textup{i}r}{\lambda^{\varrho}}H})\big]^{(k)}_{\scriptscriptstyle{Q}}(p) -U_{\lambda, -r}^{(k)}T_{k} U_{\lambda, r}^{(k)}[\widetilde{\rho}]^{(k)}_{\scriptscriptstyle{Q}}(p) \Big) \nonumber  \\   = & \lambda^{\varrho}  1_{A}(p,v,n,m)\int_{\R}dv\,\frac{j(v)}{\mathcal{R}}\sum_{n\neq m}\kappa_{v}\big(p- k-n-v,n\big)\overline{\kappa}_{v}\big(p+ k-m-v,m\big)\nonumber \\ &\times [\rho]_{\scriptscriptstyle{Q}}^{( k+\frac{1}{2}n-\frac{1}{2}m   ) }\big(p-v-\frac{1}{2}(n+m)\big) \nonumber \\ & \times \textup{i}\Big(\frac{e^{-\frac{\textup{i}s_{2}}{\lambda^{\varrho}}\big(E(p-k-n-v)-E(p+ k-m-v)-E(p-k )+E(p+k)\big)}   }{E(p- k-n-v)-E(p+k-m-v)-E(p- k )+E(p+ k) }\nonumber \\ &  -\frac{ e^{-\frac{\textup{i}s_{1}}{\lambda^{\varrho}}\big(E(p- k-n-v)-E(p+k-m-v)-E(p-k)+E(p+ k)\big)}     }{E(p- k-n-v)-E(p+ k-m-v)-E(p- k)+E(p+ k) } \Big) , 
\end{align}
where $A\subset \R^{2}\times \Z^2 $ is the set of $p,v,n,m$ such that:
\begin{enumerate}[(I).]
 \item  Either $p-k-n-v$ or $p+ k-m-v$ is not in the set $\cup_{N\in \Z}  [\frac{1}{2}N -\lambda^{2}   , \frac{1}{2}N +\lambda^{2} ] $.
 
 \item The number $p$ is not in the set $\cup_{N\in \Z}  [\frac{1}{2}N -\lambda^{2}   , \frac{1}{2}N +\lambda^{2} ] $, and 
 $|p|\leq 2\lambda^{-\frac{3}{2}-\frac{\gamma}{4} }$.  
 
 \end{enumerate}
I will argue below that statements (I) and (II) guarantee the inequality 
\begin{align}\label{Terk}
\Big| E(p- k )-E(p+ k)\Big| \leq \frac{1}{2}\Big|E(p- k-n-v)-E(p+ k-m-v)\Big|,
\end{align}
which obviously implies that
\begin{align}\label{Bezerk}
\Big|E\big(p- k-n-v\big)-& E\big(p+k-m-v)-E(p- k \big)+E\big(p+k\big) \Big|^{-1}\nonumber \\
&\leq 2\Big|E\big(p- k-n-v\big)-E\big(p+k-m-v\big) \Big|^{-1}.
\end{align}
It is an advantage to have a simplified denominator in later analysis, and the purpose of introducing $\mathbf{Q}''$ and $\widetilde{\rho}$ earlier in the proof was to avoid some scenarios in which the denominator on the left side of~(\ref{Bezerk}) becomes small.  

To see~(\ref{Terk}) notice that statement (I) and $|k|\leq \frac{1}{2}\lambda^{2}$ imply that  $p- k-n-v$ and $p+k-m-v$ always lie on different energy bands for $n\neq m$.  It follows that 
\begin{align}\label{Stars}
\Big|E\big(p-k-n-v\big)- E\big(p+ k-m-v)\Big|\geq \inf_{n\in \mathbb{N}}g_{n}:= c',  
\end{align}
where $g_{n}$ is the $n$th gap between energy bands.  Moreover, statement (II) implies that $p-k$ and $p+k$ belong to the same energy band,  and I have the bound    
\begin{align}\label{Stripes}
\Big|E\big(p-k\big)-E\big(p+k\big)\Big|& \leq 2|k|\Big(\sup_{\pm} \big|\mathbf{q}(p\pm  k)\big|\Big)\Big(\sup_{p\in \R-\frac{1}{2}\Z   } \big| \mathbf{q}'(p)\big|  \Big)\nonumber \\ & \leq  c''  |k| \big\lceil 2|p|\big\rceil   \leq 4c''|k|\,  |p |\leq \frac{c'}{2},   
\end{align}
where $c'':= \sup_{p\in \R-\frac{1}{2}\Z   } \big| \mathbf{q}'(p)\big|$ is finite by Part (2) or Lem.~\ref{CritEstimates}.  
The first inequality in~(\ref{Stripes}) uses calculus and $E(p):=q^{2}(p)   $.  The last inequality in~(\ref{Stripes}) is for small enough $\lambda$ and  uses the constraints $|k|\leq \frac{1}{2} \lambda^{2}$, $|p|\leq 2\lambda^{-\frac{3}{2}-\frac{\gamma}{4} }$.
   The second inequality in~(\ref{Stripes}) holds since $\mathbf{q}:\R\rightarrow \R^{+}$ is monotonically increasing, the highest value on the energy band containing  $p\pm  k$ is $\frac{1}{2}\lceil 2|p|\rceil$, and there is the explicit evaluation $\mathbf{q}(\frac{1 }{ 2}\lceil  2 |p|\rceil)= \frac{1}{2}\lceil 2 |p| \rceil $. Combining~(\ref{Stars}) and~(\ref{Stripes}) yields~(\ref{Terk}).

Making a change of variables $p- k-n-v\rightarrow p$ and $N=n-m$,  the relations (\ref{Hoser}) and (\ref{Bezerk}) imply the first inequality below 
\begin{align*}
 \Big\| \mathbf{Q}'' &\int_{s_{1}}^{s_{2}}dr\,\Big(    \mathcal{R}^{-1}U_{\lambda, -r}^{(k)}\big[\Psi(e^{-\frac{\textup{i}r}{\lambda^{\varrho}}H}\widetilde{\rho} e^{\frac{\textup{i}r}{\lambda^{\varrho}}H})\big]^{(k)}_{\scriptscriptstyle{Q}} -U_{\lambda, -r}^{(k)}T_{k} U_{\lambda, r}^{(k)}[\widetilde{\rho}]^{(k)}_{\scriptscriptstyle{Q}} \Big)  \Big\|_{1}\nonumber \\ &\leq \frac{\lambda^{\varrho}}{\mathcal{R}} \sum_{N\neq 0} \int_{\R} dp\,\mathbf{Q}(p)\int_{\R}dv\,j(v) \sum_{n\in \Z }\frac{\big|\rho\big(p,p+ 2k+N\big)\big|\,\big|\kappa_{v}(p,n)\big|\,\big|\kappa_{v}\big(p+2 k+N,n-N\big)\big|}{\big| E(p)-E(p+2 k+N ) \big| } \\ & \leq c\lambda^{\varrho}\|\rho\|_{\mathbf{1} }.
\end{align*}
The second inequality is for some $c>0$ by Lem.~\ref{Fructose}.

\end{proof}

\subsection{Proof of Theorem~\ref{ThmSemiClassical}}

 The proof of Thm.~\ref{ThmSemiClassical} follows closely  from Lem.~\ref{Grim} after unraveling the maps $\Phi_{\lambda, t}:\mathcal{B}_{1}(\mathcal{H})$ and $\Phi_{\lambda, t}^{(k)}:L^{1}(\R)$ through the pseudo-Poisson representation of Sect.~\ref{SecPandL} and introducing a  telescoping sum of intermediary dynamics that evolve according to the original dynamics up to the $n$th Poisson time and the idealized dynamics afterwards.  There is a technical difficulty in the application of Lem.~\ref{HorseMeister} resulting from the presence of the factor $\Tr[H\rho]$ in the upper bound since $\Tr[H\Phi_{\lambda,\xi, t_{n}}(\check{\rho})]$ increases linearly with $n\in \mathbb{N}$ by Part (3) of Lem.~\ref{LemSubMartBasic}.  This small problem is resolved by considering a suitable time cut-off that avoids unmanageable energy growth  and by bounding the remainder through a simpler estimate.

\vspace{.4cm}

\begin{proof}[Proof of Thm.~\ref{ThmSemiClassical}]

  For $\mathcal{N}$ and $\xi$ defined as in Lem.~\ref{Trivial} and $0\leq r\leq t$, define   
$$\Phi_{\lambda,\xi,r, t}^{(k)}:= U_{\lambda, t-t_{\mathcal{N}}}^{(k)} T_{k} \cdots  U_{\lambda, t_{n+1}-t_{n}}^{(k)}T_{k} U_{\lambda, t_{n}-r}^{(k)},     $$ 
where $t_{n}$ is the first Poisson time $>r$.  By Lem.~\ref{Trivial} I have the first equality below: 
\begin{align*}
\big[\Phi_{\lambda,t }( \check{\rho}_{\lambda})\big]^{(k)}_{\scriptscriptstyle{Q}}-\Phi_{\lambda,t }^{(k)}[ \check{\rho}_{\lambda}]^{(k)}_{\scriptscriptstyle{Q}}& = \mathbb{E}\Big[  \big[\Phi_{\lambda,\xi, t} (  \check{\rho}_{\lambda})\big]^{(k)}_{\scriptscriptstyle{Q}}-   \Phi_{\lambda,\xi, t}^{(k)}[ \check{\rho}_{\lambda}]^{(k)}_{\scriptscriptstyle{Q}} \Big]\\ &=\mathbb{E}\Big[\sum_{n=1}^{\mathcal{N}_{t}(\xi)}      \Phi_{\lambda,\xi,t_{n}, t}^{(k)}\big[\Phi_{\lambda, \xi, t_{n}}  (  \check{\rho}_{\lambda})\big]^{(k)}_{\scriptscriptstyle{Q}}-\Phi_{\lambda,\xi,t_{n-1}, t}^{(k)}\big[\Phi_{\lambda,\xi,  t_{n-1}} (  \check{\rho}_{\lambda})\big]^{(k)}_{\scriptscriptstyle{Q}}          \Big].  
\end{align*}
For the second equality, I have inserted terms   $ \Phi_{\lambda,\xi,t_{n}, t}^{(k)}\big[\Phi_{\lambda, \xi, t_{n}}  ( \check{\rho}_{\lambda})\big]^{(k)}_{\scriptscriptstyle{Q}}$ in the form of a telescoping sum.    The difference between $\big[\Phi_{\lambda,t }(\check{\rho}_{\lambda} )\big]^{(k)}_{\scriptscriptstyle{Q}}$ and $\Phi_{\lambda,t }^{(k)}[\check{\rho}_{\lambda} ]^{(k)}_{\scriptscriptstyle{Q}}$ at time $t=\frac{T}{\lambda^{\gamma}}$ is smaller than
\begin{align}\label{Beaut}
\Big\| \big[&\Phi_{\lambda,\frac{T}{\lambda^{\gamma}} }( \check{\rho}_{\lambda} )\big]^{(k)}_{\scriptscriptstyle{Q}}-\Phi_{\lambda,\frac{T}{\lambda^{\gamma}} }^{(k)}[\check{\rho}_{\lambda} ]^{(k)}_{\scriptscriptstyle{Q}}\Big\|_{1}\leq  2e^{-\frac{\mathcal{R}T}{\lambda^{\gamma}}}\sum_{\mathcal{N}=\lfloor  \frac{\mathcal{R}T}{\lambda^{2}}\rfloor+1}^{\infty} \frac{1}{\mathcal{N}!}\big(\frac{\mathcal{R}T}{\lambda^{\gamma}}\big)^{\mathcal{N}}\nonumber \\ & + e^{-\frac{\mathcal{R}T}{\lambda^{\gamma}}}\sum_{\mathcal{N}=1}^{\lfloor  \frac{\mathcal{R}T}{\lambda^{2}}\rfloor}\mathcal{R}^{\mathcal{N}}\sum_{n=1}^{\mathcal{N}}\Big \|\int_{0\leq t_{1}\cdots \leq t_{\mathcal{N}}\leq \frac{T}{\lambda^{\gamma}} }    \Phi_{\lambda,\xi,t_{n}, \frac{T}{\lambda^{\gamma}} }^{(k)}\big[\Phi_{\lambda, \xi, t_{n}}  ( \check{\rho}_{\lambda} )\big]^{(k)}_{\scriptscriptstyle{Q}}-\Phi_{\lambda,\xi,t_{n-1}, \frac{T}{\lambda^{\gamma}} }^{(k)}\big[\Phi_{\lambda,\xi,  t_{n-1}} ( \check{\rho}_{\lambda} )\big]^{(k)}_{\scriptscriptstyle{Q}}     \Big \|_{1} .
\end{align}
In the above, I have applied the triangle inequality to the telescoping sums for the first $\lfloor  \frac{\mathcal{R} T}{\lambda^{2}}\rfloor$ terms.  For the remaining terms, I have used that  $ \Phi_{\lambda,\xi, t}^{(k)}$ is contractive in the $1$-norm, $\Phi_{\lambda,\xi,t}$ is contractive in the trace norm, and $\|[  \check{\rho}_{\lambda}]^{(k)}_{\scriptscriptstyle{Q}}\|_{1}\leq \| \check{\rho}_{\lambda}\|_{\mathbf{1}}=1$.  The first term on the second line of~(\ref{Beaut}) decays superpolynomially as $\lambda$ goes to zero.   

A single term from the sum on the second line of ~(\ref{Beaut}) is smaller than 
\begin{align}\label{Janet} \Big\| &\int_{0\leq t_{1}\cdots \leq t_{\mathcal{N}}\leq \frac{T}{\lambda^{\gamma}} }    \Phi_{\lambda,\xi,t_{n},\frac{T}{\lambda^{\gamma}} }^{(k)}\big[\Phi_{\lambda, \xi, t_{n}}  ( \check{\rho}_{\lambda} )\big]^{(k)}_{\scriptscriptstyle{Q}}-\Phi_{\lambda,\xi,t_{n-1}, \frac{T}{\lambda^{\gamma}} }^{(k)}\big[\Phi_{\lambda,\xi,  t_{n-1}} ( \check{\rho}_{\lambda} )\big]^{(k)}_{\scriptscriptstyle{Q}}     \Big\|_{1}\nonumber \\ \leq &\int_{0\leq t_{1}\cdots \leq t_{n-1}\leq t_{n+1}\leq \cdots  t_{\mathcal{N}}\leq \frac{T}{\lambda^{\gamma}}}  \Big\| \int_{t_{n-1}}^{t_{n+1}}dt_{n}\,\Big(    \mathcal{R}^{-1}U_{\lambda, -t_{n}+t_{n-1} }^{(k)}\big[\Psi(e^{-\frac{\textup{i} (t_{n}-t_{n-1} )}{\lambda^{\varrho}} H}\Phi_{\lambda, \xi, t_{n-1}}  ( \check{\rho}_{\lambda} ) e^{\frac{\textup{i} (t_{n}-t_{n-1} )}{\lambda^{\varrho}}H})\big]^{(k)}_{\scriptscriptstyle{Q}}\nonumber \\ &-U_{\lambda, -t_{n}+t_{n-1} }^{(k)}T_{\lambda}^{(k)} U_{\lambda, t_{n}-t_{n-1} }^{(k)}[\Phi_{\lambda, \xi, t_{n-1}}  (\check{\rho}_{\lambda} )]^{(k)}_{\scriptscriptstyle{Q}}  \Big) \Big\|_{1}\nonumber \\
\leq & \frac{C}{\mathcal{N}!}\big(\frac{T}{\lambda^{\gamma}}\big)^{\mathcal{N}}\Big(\lambda^{3+\frac{\gamma}{2} }  \Tr\big[H\Phi_{\lambda, \xi, t_{n-1}}  ( \check{\rho}_{\lambda} )\big] +\lambda^{2}\big\|\big\langle \Phi_{\lambda, \xi, t_{n-1}}  ( \check{\rho}_{\lambda} ) \big\rangle^{(0)}\big\|_{\infty}+ \lambda^{2} \Big)   +\frac{C\lambda^{\varrho}}{(\mathcal{N}-1)!}\big(\frac{T}{\lambda^{\gamma}}\big)^{\mathcal{N}-1}\nonumber  \\ \leq &    \frac{C'\lambda^{1+\frac{\gamma}{2}}}{\mathcal{N}!} \big(\frac{T}{\lambda^{\gamma}}\big)^{\mathcal{N}} +\frac{C\lambda^{\varrho}}{(\mathcal{N}-1)!}\big(\frac{T}{\lambda^{\gamma}}\big)^{\mathcal{N}-1} 
\end{align}   
for some constants $C,C'>0$, where  I  identify $t_{0}\equiv 0$ and $t_{\mathcal{N}+1}\equiv \frac{T}{\lambda^{\gamma}}$ for the boundary terms on the second line.  The first inequality above uses that $\Phi_{\lambda,\xi,r, t}^{(k)}$ and  $U_{\lambda, t}^{(k)}$ are contractive in the $1$-norm.    The second inequality in~(\ref{Janet}) is by Lem.~\ref{Grim} and $\|  \Phi_{\lambda, \xi, t_{n-1}}  ( \check{\rho}_{\lambda} )  \|_{\mathbf{1}}\leq 1$.   The first term of the fourth line can be bounded by an application of Part (3) from Lem.~\ref{LemSubMartBasic} to get
$$\sup_{0\leq n \leq \frac{T\mathcal{R}}{\lambda^{2} } } \lambda^{3+\frac{\gamma}{2} } \Tr\big[H\Phi_{\lambda, \xi, t_{n-1}}  ( \check{\rho}_{\lambda} )\big]\leq \lambda^{3+\frac{\gamma}{2} }\Big( \Tr\big[H \check{\rho}_{\lambda} \big]+\frac{\sigma T}{\lambda^{2}} \Big)=\mathit{O}\big(\lambda^{1+\frac{\gamma}{2}}\big).
$$ 
The expression $\|\langle \Phi_{\lambda, \xi, t_{n-1}}  ( \check{\rho}_{\lambda} ) \rangle^{(0)}\|_{\infty}$ from the second term of the fourth line of~(\ref{Janet}) is uniformly bounded for small $\lambda$ by Lem.~\ref{HorseMeister} and Part (3) of Prop.~\ref{PropFiberII}.

With the result~(\ref{Janet}), the second line of~(\ref{Beaut}) is smaller than
$$e^{-\frac{\mathcal{R}T}{\lambda^{\gamma}}}\sum_{\mathcal{N}=1}^{\lfloor  \frac{\mathcal{R}T}{\lambda^{2} }\rfloor}\mathcal{N}\mathcal{R}^{\mathcal{N}}\Big(\frac{C'\lambda^{1+\frac{\gamma}{2}}}{\mathcal{N}!} \big(\frac{T}{\lambda^{\gamma}}\big)^{\mathcal{N}} +\frac{C\lambda^{\varrho}}{(\mathcal{N}-1)!}\big(\frac{T}{\lambda^{\gamma}}\big)^{\mathcal{N}-1}\Big) \leq  C'\mathcal{R}T \lambda^{1-\frac{\gamma}{2} }+C(1+\mathcal{R} T)\lambda^{\varrho-\gamma}.  $$
To obtain the above inequality, I have replaced the upper bound of the sum, $\lfloor  \frac{\mathcal{R}T}{\lambda^{2}}\rfloor$, by infinity and applied elementary manipulations to the Taylor expansion of an exponential function.  Since $1<\gamma <2$ and $\gamma<\varrho$, the right side tends to zero  with order $\lambda^{\iota}$ for   $\iota=\textup{min}(\varrho -\gamma,1-\frac{\gamma}{2}) $ and small  $\lambda$.

\end{proof}

\section{Classical Markovian approximation}\label{SecTransition}

This section contains the proof of Thm.~\ref{ThmSemiClassical}.  The main estimates that I use to prove Thm.~\ref{ThmSemiClassical} are stated in the lemma below.   Part (1) of  Lem.~\ref{HillDog} essentially states that the dispersion relation $E(p)$ has derivative close to $2p $ for most $p\in \R$ with $|p|\gg1$, and Part (2) is related to the continuity in $p$ of the coefficients $\kappa_{v }(p,n)$.  Both estimates require that the momenta involved are not too close to the lattice $\frac{1}{2}\Z$.   Recall that the function $\Theta:\R\rightarrow [-\frac{1}{4},\frac{1}{4}   )$ contracts momenta modulo $\frac{1}{2}$.

\begin{lemma} \label{HillDog}
Let $|k|\leq \lambda^{\frac{\gamma}{2}+\frac{5}{4}+\varrho}$ and $A\subset \R^{2}$ be the set of all $(p,v)$ satisfying $|\Theta(p)|,|\Theta(p+v)|>\lambda^{\frac{\gamma}{2}+1} $.  
There is a $C>0$ and an $\iota >0$ such that for all $\lambda<1$ and $(p,v)\in A$, 
\begin{enumerate}
\item  $\big| E\big(p-k\big)-E\big(p+k\big)+4p  k\big| < C \lambda^{\frac{\gamma}{2}+1+\varrho+\iota} \Big(|\Theta(p)|^{-1}+1\Big) $,

\item $ \sum_{n\in \mathbb{Z} }\Big| \kappa_{v }\big(p-k,\,n\big) \overline{\kappa}_{v }\big(p+k,\,n\big)- \big|\kappa_{v }\big(p,\,n\big)\big|^{2}    \Big|\leq C\big( 1+v^{2}\big)  \lambda^{\frac{\gamma}{2}+1+\iota } 
$.
\end{enumerate}

\end{lemma}

The proofs for Parts (1) and (2) of Lem.~\ref{HillDog} are contained in Sects.~\ref{SecDisRel} and~\ref{SecCoe}, respectively.  

\subsection{Estimates for the dispersion relation}\label{SecDisRel}

The results of this section will require a closer examination of the function $\mathbf{q}:\R\rightarrow \R$ determined by the Kr\"onig-Penney relation~(\ref{KronigPenney}).  Recall that  $\mathbf{q}$ is anti-symmetric, increasing,  and satisfies
 $$\hspace{2cm}\frac{n}{2}=\mathbf{q}\big(\frac{n}{2}\big)=\lim_{\epsilon\searrow 0} \mathbf{q}\big(\frac{n}{2}-\epsilon\big)< \lim_{\epsilon\searrow 0} \mathbf{q}\big(\frac{n}{2}+\epsilon\big),\hspace{1cm}n\in \mathbb{N}.  $$ 
In words, the function $\mathbf{q}$ has jumps at points in $\frac{1}{2}\Z-\{0\}$ but is continuous from the direction of the origin.  It is convenient to view $\mathbf{q}(p)$ over bands $p\in\big(\frac{n-1}{2},\frac{n}{2} \big]$, $n\in \mathbb{N}$ over which the function is continuous.   For each $N\in \mathbb{N}$ define the functions $f_{N}:[0,\pi]\rightarrow [-1,\infty)$ and $g_{N}:[0,\pi]\rightarrow [0,\pi]$ as
\begin{eqnarray*}
   f_{N}(x)&:=&\cos(\pi-x)+\alpha \pi \frac{\sin(\pi-x)}{ \pi N-x }, \\  g_{N}(x)&:=&   f_{N}^{-1}\big( \cos(\pi-x)   \big).   
    \end{eqnarray*}
The function $\mathbf{q}:\R\rightarrow \R$  can be written in the form
\begin{align}\label{ReKronigPenney}
\hspace{5cm} \mathbf{q}(p)=\frac{1}{2} \lceil 2 p  \rceil - \frac{1}{2\pi}g_{N}\big( \pi \lceil 2p \rceil- 2\pi p  \big), \hspace{1cm} p> 0 .      
 \end{align}

The following proposition is a consequence of basic calculus, and I do not include the proof.   
\begin{proposition}\label{LemCalculus}
Set $\upsilon:= \alpha \pi   $.  
\begin{enumerate}
\item   The function $g_{N}:[0,\pi]\rightarrow [0,\pi]$ satisfies the differential equation 
$$  g_{N}'(x)=  \frac{ \sin(\pi-x)   }{-\frac{\upsilon\cos(\pi-x)}{\pi N-g_{N}(x)}+\sin(\pi-g_{N}(x))+\frac{\upsilon^{2} \sin(\pi-g_{N}(x)) }{(\pi N-g_{N}(x))^2  }+\frac{\upsilon\sin(\pi-g_{N}(x))}{(\pi N-g_{N}(x))^{2}  }    }  .     $$  

\item  The second derivative of $g_{N}$ can be written implicitly in the form
$$ g_{N}''(x)= -\frac{ \frac{\upsilon^{2}\cos(\pi-x)    }{ (\pi N-g_{N}(x))^{2}   }+r_{N}(x)   }{\big(-\frac{\upsilon\cos(\pi-x)}{\pi N-g_{N}(x)}+\sin(\pi-g_{N}(x))+\frac{\upsilon^{2} \sin(\pi-g_{N}(x)) }{(\pi N-g_{N}(x))^2  } +\frac{\upsilon\sin(\pi-g_{N}(x))}{(\pi N-g_{N}(x))^{2}  }  \big)^{3} } ,    $$
where $r_{N}:[0,\pi]\rightarrow \R$ is defined as
\begin{align*}
 r_{N}(x) :=  &\frac{ 2\upsilon^{2}\sin(\pi-g_{N}(x))}{(\pi N-g_{N}(x))^{3}  }-\frac{ 3\upsilon^{2}\sin^{2}(\pi-g_{N}(x))\cos(\pi-g_{N}(x))}{(\pi N-g_{N}(x))^{4}  }\\ &+\frac{ 2\upsilon \sin^{3}(\pi-g_{N}(x))}{(\pi N-g_{N}(x))^{3}  }-\frac{ \upsilon^{3} \sin^{3}(\pi-g_{N}(x))}{(\pi N-g_{N}(x))^{5}  } . 
 \end{align*}

\end{enumerate}
The first two terms of the denominators in (1) and (2) can be alternatively written with the equality 
\begin{align}\label{SmallRemark}
-\frac{\upsilon\cos(\pi-x)}{\pi N-g_{N}(x)}+\sin(\pi-g_{N}(x)) = \frac{\upsilon\cos(\pi-x)}{\pi N-g_{N}(x)}+\frac{ \sin^{2}(\pi-x)   }{ \sin(\pi-g_{N}(x))}.
\end{align}

\end{proposition}

For the statement of Lem.~\ref{CritEstimates}, recall that the map $\Theta:\R\rightarrow [-\frac{1}{4},\frac{1}{4})$ contracts values in $p\in \R$ modulo $\frac{1}{2}$.

\begin{lemma}\label{CritEstimates}
There is $C>0$ such that  for all  $p \in \R $, 
\begin{enumerate}

\item $ \big|\mathbf{q}(p )-p\big|\leq  \frac{C}{ 1+|p|     }$,

\item  $ \big|\mathbf{q}'(p )-1\big|\leq  C   \textup{min}\Big\{1,  \frac{ 1 }{ |\theta(p)| (1+ |p|) } \Big\}$,

\item     $ \big|\mathbf{q}''(p)\big|\leq  C\textup{min}\Big\{\frac{1}{1+ |p| },\,  \frac{ 1 }{ |\theta(p)|^{3}(1+ |p|)^{2}   }\Big\}. $

\end{enumerate}

\end{lemma}

\begin{proof}By the equality~(\ref{ReKronigPenney}), it is equivalent to show that there is a $C>0$ such that all $x\in [0,\pi)$ and $N\in \mathbb{N}$,   
\begin{enumerate}
\item $\big|g_{N}(x)-x\big|\leq C \frac{1 }{ 1+N     }$,

\item $   \big|g_{N}'(x)-1\big|\leq  C   \textup{min}\Big\{1, \frac{ 1  }{ \textup{min}\{x,\pi-x\}   (1+ N) } \Big\}$,

\item   $ \big|g_{N}''(x)\big|\leq  C   \textup{min}\Big\{\frac{1}{ 1+N}, \frac{ 1  }{ (\textup{min}\{x,\pi-x\})^{3}   (1+ N)^{2} } \Big\} $.

\end{enumerate}

\noindent Part (1):  Clearly $g_{N}(x)\leq x$, since $f_{N}(x)\geq \cos(\pi-x)$ over the interval $[0,\pi]$.  The definition of $g_{N}$ gives the first equality below:
\begin{align}\label{Hazzelrod}
\frac{\upsilon  \sin(\pi-g_{N}(x))     }{\pi N-g_{N}(x)    }=&\cos(\pi-x)-\cos(\pi-g_{N}(x))\nonumber \\ =&\int_{0}^{x-g_{N}(x)}dy\, \sin(\pi-g_{N}(x)-y)\nonumber \\ \geq & \frac{1}{2}\big(x-g_{N}(x)\big) \sin(\pi-g_{N}(x)) 
 . 
\end{align}
The inequality in~(\ref{Hazzelrod}) uses that the function $F(y)=\sin(\pi-g_{N}(x)-y)$ is positive and  concave down over the interval  $y\in [0,x-g_{N}(x)]$; the area between graph of  $F(y)$ and the $y$-axis from $y=0$ to $y=x-g_{N}(x)$ thus encloses a triangle with height $F(0)=\sin(\pi-g_{N}(x)) $ and width $x-g_{N}(x)$.  From~(\ref{Hazzelrod}), it follows that
$$x-g_{N}(x)\leq \frac{2\upsilon}{\pi N-g_{N}(x)}=\mathit{O}(N^{-1}).  $$

\vspace{.4cm}

\noindent Part (2): 
It is convenient to use the form $g_{N}'(x)$ from Part (1)  of Lem.~\ref{LemCalculus} for  the domain $x\in [0,\frac{\pi}{2})$, and the alternative form using the remark~(\ref{SmallRemark}) for the domain $x\in [\frac{\pi}{2},\pi)$.  The analysis for the domains are similar, so I will discuss only $[0,\frac{\pi}{2})$.  By Part (1) of Prop.~\ref{LemCalculus} and since the terms $\frac{\upsilon^{2} \sin(\pi-g_{N}(x)) }{(\pi N-g_{N}(x))^2  }$, $\frac{\upsilon\sin(\pi-g_{N}(x))}{(\pi N-g_{N}(x))^{2}  } $ in the denominator of the expression for $g_{N}'(x)$ are positive over the domain $x\in [0,\pi)$,  I have the first inequality below:
\begin{align}\label{Jengo}
 \Big|g_{N}'(x)-1\Big| & \leq  \frac{\big| \sin(\pi-g_{N}(x))    -\sin(\pi-x)\big|+\mathit{O}(N^{-2}) }{\big|-\frac{\upsilon\cos(\pi-x)}{\pi N-g_{N}(x)}+\sin(\pi-g_{N}(x)) \big|  } \nonumber   \\ & \leq \Big(\frac{C}{1+N}+\mathit{O}(N^{-2})\Big) \textup{min}\Big\{ \frac{N}{-\upsilon \textup{cos}(\pi-x)} , \frac{1}{ \sin(\pi-g_{N}(x))}       \Big\}  \nonumber 
\\ &  \leq \mathit{O}(N^{-1}) \textup{min}_{+}\Big\{ \frac{N}{\upsilon \textup{cos}(\pi-x)   } , \frac{1}{\sin(\pi-x)-\frac{C}{1+N } }    \Big\} ,
\end{align}
where  $\textup{min}_{+}$ refers to the minimum positive value. The second inequality uses that sine has  derivative bounded by one and  Part (1) to guarantee that there is a  $C>0$ such that $|x-g_{N}(x)|\leq \frac{C}{1+N}$.  The third inequality bounds the difference between $\sin(\pi-x)$ and  $\sin(\pi-g_{N}(x))$ by  $\frac{C}{1+N } $ again.  The result can be easily seen by using linear lower bounds for the trigonometric functions on the third line of~(\ref{Jengo}).

\vspace{.4cm}

\noindent Part (3): Similar to Part (2).

\end{proof}

\vspace{.4cm}

\begin{proof}[Proof of Part (1) from Lem.~\ref{HillDog}] By writing $E(p\pm  k)$ in terms of first-order Taylor's formulas and using that $E(p)=\mathbf{q}^{2}(p)$, I have the equality
\begin{align}\label{RunRun}
E(p+ k)-E(p- k)-4 kp= &2 k\big(\mathbf{q}'(p)\mathbf{q}(p)-p \big)\nonumber \\ &+2\int_{- k}^{k }dv\,\int_{0}^{v}dw\,\Big(\mathbf{q}(p+w)\mathbf{q}''(p+w)+   \big|\mathbf{q}'(p+w)\big|^{2} \Big).   
\end{align}
It is thus sufficient to  bound the terms on the right side of~(\ref{RunRun}).
  The first term on the right side of~(\ref{RunRun}) has the bound
\begin{align*}
 |k|\,\big|\mathbf{q}'(p)\mathbf{q}(p)-p \big|\leq &  |k|\, \big|\mathbf{q}(p)\big|\,\big|\mathbf{q}'(p)-1\big| +| k|\,\big|\mathbf{q}(p)  -p \big| \\ \leq &  |k|\big( C +   |p|  \big)\frac{ C|\Theta(p)|^{-1} }{ 1+|p|  }+ |k|\frac{ C}{1+|p| } , \\ \leq & C'\lambda^{\frac{\gamma}{2}+\frac{5}{4}+\varrho }\Big(|\Theta(p)|^{-1}+1\Big)  , 
\end{align*}
where  the second inequality follows   for some  constant $C>1$ by applications of Parts (1) and (2) of Lem.~\ref{CritEstimates}.  The third inequality holds for $C'=C^{2}$.

For the second term on the right side of~(\ref{RunRun}),  
\begin{align*}
\int_{-k}^{k }dv\int_{0}^{v}dw\,\Big(& \big|\mathbf{q}(p+w)\big|\,\big|\mathbf{q}''(p+w)\big|+   \big|\mathbf{q}'(p+w)\big|^{2} \Big) \\ &\leq k^{2}\sup_{|r|\leq |k| } \big( \big|\mathbf{q}(p+r)\big|\,\big|\mathbf{q}''(p+r)\big|+\big|\mathbf{q}'(p+r)\big|^{2}\big)  \\  & \leq  k^{2}\sup_{|r|\leq |k| }\Big(   \frac{ C\big(C+|p+r| \big) |\Theta(p+r)|^{-3}  }{ \big(1+|p+r| \big)^{2}  }+\frac{  2C |\Theta(p+r)|^{-2} }{ (1+| p+r| )^{2}  }\Big)\\
& = C'' \lambda^{\frac{1}{2}+2\varrho } \Big(|\Theta(p)|^{-1}+1\Big)   ,
\end{align*}
where the second inequality is by Parts 1-3 of Lem.~\ref{CritEstimates}.  The third inequality holds for some $C''>0$ since $ |\Theta(p)|\geq \lambda^{\frac{\gamma}{2}+1}$ and thus $ |\Theta(p)|\gg  |k|$.  

Since $\varrho>\gamma>1$ the above bounds imply that~(\ref{RunRun}) is $\mathit{O}(\lambda^{\frac{\gamma}{2}+1+\varrho+\iota})$ for small enough $\iota>0$.

\end{proof}

\subsection{Estimates for the coefficients $\kappa_{v}(p,n)$  }\label{SecCoe}

\begin{lemma}\label{LemKappaDer}

There is a $C>0$ such that for all $p,v\in \R$, and $n\in \Z$,     
\begin{enumerate}
\item  $ \big|\frac{\partial^{2}}{\partial^{2} p}\kappa_{v}(p,n)\big|\leq  C\big(\frac{1}{|\Theta(p)|^{2} }+ \frac{1}{|\Theta(p+v)|^{2} }\big)  
, $

\item  $ \big|\textup{Im}\big[\overline{\kappa}_{v}(p,n)\frac{\partial}{\partial p}\kappa_{v}(p,n)\big]\big|\leq  C$.
\end{enumerate}

\end{lemma}

\begin{proof}\text{  }\\
Part (1): The following formula for $\kappa_{v}(p,n)$ is equivalent to the definition~(\ref{TheKappas}):
\begin{align}\label{Shera}
 \kappa_{v}(p,n)=\Big\langle \widetilde{\psi}_{p+v+n }\,\Big|\, e^{\textup{i}vX_{\mathbb{T}}}  \widetilde{\psi}_{p}       \Big\rangle .  
\end{align} 
By the product rule, the second derivative of $ \kappa_{v}(p,n)$ can be written in the form 
\begin{align}\label{JackAndJill}
\frac{\partial^{2}}{\partial^{2} p}\kappa_{v}(p,n)= &\Big\langle \frac{\partial^{2}}{\partial^{2} p}\widetilde{\psi}_{p+v+n }\,\Big|\, e^{\textup{i}vX_{\mathbb{T}}}  \widetilde{\psi}_{p}       \Big\rangle +\Big\langle \widetilde{\psi}_{p+v+n }\,\Big|\, e^{\textup{i}vX_{\mathbb{T}}} \frac{\partial^{2}}{\partial^{2} p} \widetilde{\psi}_{p}       \Big\rangle \nonumber \\ & +2\Big\langle \frac{\partial}{\partial p}\widetilde{\psi}_{p+v+n }\,\Big|\, e^{\textup{i}vX_{\mathbb{T}}} \frac{\partial}{\partial p} \widetilde{\psi}_{p}       \Big\rangle .
\end{align}
Since the operator $e^{\textup{i}vX_{\mathbb{T}}}\in \mathcal{B}\big(L^{2}(\mathbb{T})\big)$ has norm bounded by one, 
\begin{align}\label{Blimey}
\Big|\frac{\partial^{2}}{\partial^{2} p}\kappa_{v}(p,n)\Big|\leq \Big\| \frac{\partial^{2}}{\partial^{2} p}\widetilde{\psi}_{p+v+n }\Big\|_{2}+\Big\| \frac{\partial^{2}}{\partial^{2} p} \widetilde{\psi}_{p}       \Big\|_{2}  +2\Big\| \frac{\partial}{\partial p}\widetilde{\psi}_{p+v+n }\Big\|_{2}\Big\| \frac{\partial}{\partial p} \widetilde{\psi}_{p}       \Big\|_{2} .
\end{align}
The Bloch functions $ \widetilde{\psi}_{p}:\mathbb{T}\rightarrow \C$ have the form~(\ref{BlochFunctions}), where the normalization constant $N_{p}$ is equal to
$$N_{p}=2\pi+2\pi\frac{1-\cos\big(2\pi(\mathbf{q}(p)-p) \big) }{ 1-\cos\big(2\pi(\mathbf{q}(p)+p) \big)   }+\frac{\big|\cos\big(2\pi\mathbf{q}(p)   \big)-\cos\big(2\pi p   \big)\big|^{2}}{  1-\cos\big(2\pi(\mathbf{q}(p)+p)\big)   }.$$
The first two derivatives of $\widetilde{\psi}_{p} $ have the forms
\begin{align}\label{June}  \frac{\partial}{\partial p} \widetilde{\psi}_{p}=  \frac{\psi^{(1)}_{p}}{e^{\textup{i}2\pi(\mathbf{q}(p)+p)}-1  }   \quad \text{and} \quad \frac{\partial^{2}}{\partial^{2} p} \widetilde{\psi}_{p}=  \frac{\psi^{(2)}_{p}  }{\big(e^{\textup{i}2\pi (\mathbf{q}(p)+p)}-1 \big)^{2}}+\frac{\mathbf{q}''(p) \psi^{(3)}_{p}   }{ e^{\textup{i}2\pi (\mathbf{q}(p)+p)}-1 }
\end{align}
   for some $\psi^{(1)}_{p},\psi^{(2)}_{p},\psi^{(3)}_{p}\in L^{2}(\mathbb{T})$ that are uniformly bounded in norm for all  $p\in \R$.  The forms~(\ref{June}) use 
that $N_{p}\geq 2\pi$ is bounded away from zero and  that $\mathbf{q}'(p)$ is bounded by Part (2) of Lem.~\ref{CritEstimates}. 
   
    The modulus of the expression $e^{\textup{i}2\pi (\mathbf{q}(p)+p)}-1$ has the lower bound 
\begin{align}\label{Delaware}
\big| e^{\textup{i}2\pi(\mathbf{q}(p)+p)}-1\big|\geq \Big| \Theta\Big( \frac{ \mathbf{q}(p)+p}{2} \Big)  \Big|\geq  \big|\Theta(p)   \big| ,
\end{align}
where the first inequality is a piecewise linear lower bound for $\big| e^{4\pi\textup{i}x}-1\big|$, and the second uses that $q(p)\geq p$.  
Applying~(\ref{June}) and~(\ref{Delaware}) in~(\ref{Blimey}), I have the bound
\begin{align*}
\Big|\frac{\partial^{2}}{\partial^{2} p}\kappa_{v}(p,n)\Big|\leq &  \frac{B+C B}{\big|\Theta(p+v)\big|^{2} }+ \frac{B+C B}{\big|\Theta(p)\big|^{2} }+\frac{C B}{\big|\Theta(p)\big|\,\big|\Theta(p+v)\big| }  \\ \leq & C'\Big(\frac{1}{\big|\Theta(p)\big|^{2} }+ \frac{1}{\big|\Theta(p+v)\big|^{2} }   \Big).  
\end{align*}
where the constant $C>0$ is from bounding $|\mathbf{q}''(p)|$ with Part (3) of Lem.~\ref{CritEstimates}, and  $B>0$ is the supremum over  $ \|\psi^{(j)}_{p}\|_{2} $ for $p\in \R$ and $j\in \{1,2,3\}$.  The second inequality is for some $C'$ after applying the relation $2xy\leq x^{2}+y^{2}$ to $x=\frac{1}{|\theta(p)|}$ and $y=\frac{1}{|\theta(p+v)|}$ to the last term in the first line.      

\vspace{.4cm}

\noindent Part (2):  Differentiating~(\ref{Shera}) gives
\begin{align*}
\frac{\partial}{\partial p}\kappa_{v}(p,n)= \Big\langle \frac{\partial}{\partial p}\widetilde{\psi}_{p+v+n }\,\Big|\, e^{\textup{i}v X_{\mathbb{T}}}  \widetilde{\psi}_{p}       \Big\rangle +\Big\langle \widetilde{\psi}_{p+v+n }\,\Big|\, e^{\textup{i}vX_{\mathbb{T}}} \frac{\partial }{\partial p} \widetilde{\psi}_{p}       \Big\rangle \nonumber  .
\end{align*}
By the formula~(\ref{BlochFunctions}), the Bloch function $ \widetilde{\psi}_{p}\in L^{2}(\mathbb{T})$ has the form  $ \widetilde{\psi}_{p}=\psi_{p}^{-}+\psi_{p}^{+}$ with
\begin{align}\label{Mongrel}
\frac{\partial}{\partial p}\psi_{p }^{-}= \frac{ \psi_{p }^{-}   }{ \textup{sin}\big(\pi(\mathbf{q}(p)+p)   \big)    }+\psi_{p }^{- ,\prime},
\end{align}
where $ \psi_{p}^{-}, \psi_{p}^{+},\psi_{p }^{-,\prime}, \frac{\partial}{\partial p}\psi_{p }^{+}\in L^{2}(\mathbb{T})$ are uniformly bounded in norm for all   $p\in \R$. Since the factor  $\textup{sin}\big(\pi (\mathbf{q}(p)+p)\big)$ is real, I have the following equality:   
\begin{align*}
 \big|\textup{Im}\big[\overline{\kappa}_{v}(p,n)\frac{\partial}{\partial p}\kappa_{v}(p,n)\big]\big|=&\Big| \textup{Im}\Big[ \overline{\kappa}_{v}(p,n)\Big\langle   \psi_{p+v+n }^{- ,\prime} +\frac{\partial}{\partial p}\psi_{p+v+n }^{+}\,\Big|\, e^{\textup{i}vX_{\mathbb{T}}} \widetilde{\psi}_{p}       \Big\rangle \\ &+  \overline{\kappa}_{v}(p,n)\Big\langle   \widetilde{\psi}_{p+v+n }\,\Big|\, e^{\textup{i}vX_{\mathbb{T}}} \big(   \psi_{p }^{- ,\prime} +\frac{\partial}{\partial p}\psi_{p }^{+}\big) \Big\rangle\Big] \Big| \\ 
\leq & 2\sup_{ p\in \R}\max\Big\{\| \psi_{p }^{-,\prime}\|_{2},  \|  \frac{\partial}{\partial p}\psi_{p }^{+}   \|_{2}\Big\}:=C.
\end{align*}
The inequality above uses that $\kappa_{v}(p,n)\in \C$, $ e^{\textup{i}vX_{\mathbb{T}}}\in \mathcal{B}\big(L^{2}(\mathbb{T})\big)$, and  $ \widetilde{\psi}_{p}\in L^{2}(\mathbb{T})$ have norms less than one.  Since $ \psi_{p }^{-,\prime}$, $\frac{\partial}{\partial p}\psi_{p }^{+}$ are uniformly bounded in norm, the constant $C$ is finite.

\end{proof}

\begin{lemma}\label{LemVariance}

The following variance formula holds for a.e.  $(p,v)\in \R^{2}$:
$$\sum_{n\in \Z}\Big(  E^{\frac{1}{2}}(p+v+n)\,   \big| \kappa_{v}(p,n)    \big|^{2} - \sum_{m\in \Z}E^{\frac{1}{2}}(p+v+m)\,   \big| \kappa_{v}(p,m)    \big|^{2}    \Big)^{2}\leq v^{2} .   $$

\end{lemma}

\begin{proof}
The formula below holds generically for Schr\"odinger Hamiltonians $H=P^{2}+V(X)$ and any $v\in \R$: 
$$2 v^{2}I=e^{-\textup{i}vX}H e^{\textup{i}v X}+ e^{\textup{i}vX}H e^{-\textup{i}vX}-2H . $$
In particular, this implies that the fiber Hamiltonians $H_{\phi}$, $\phi\in \mathbb{T}$ satisfy
\begin{align}\label{OnFiber}
2v^{2}I= e^{-\textup{i}vX_{\mathbb{T}}}H_{\phi_{+}} e^{\textup{i}vX_{\mathbb{T}}}+ e^{\textup{i}vX_{\mathbb{T}}}H_{\phi_{-}} e^{-\textup{i}vX_{\mathbb{T}}}-2H_{\phi}    
 \end{align}
for $\phi,\phi_{\pm}\in \mathbb{T}$ with $\phi=p \,\textup{mod}\,1$ and $\phi_{\pm}=p\pm v \,\textup{mod}\,1$.  Evaluating both sides of~(\ref{OnFiber}) by $|\widetilde{\psi}_{p}\rangle $ yields
\begin{align*}
2v^{2}= &\big\langle \widetilde{\psi}_{p}\big| e^{-\textup{i}vX_{\mathbb{T}}}H_{\phi_{+}} e^{\textup{i}vX_{\mathbb{T}}}\big| \widetilde{\psi}_{p} \big\rangle + \big\langle \widetilde{\psi}_{p}\big| e^{\textup{i}vX_{\mathbb{T}}}H_{\phi_{-}} e^{-\textup{i}vX_{\mathbb{T}}}\big| \widetilde{\psi}_{p} \big\rangle -2\big\langle \widetilde{\psi}_{p}\big|H_{\phi}\big| \widetilde{\psi}_{p} \big\rangle \\
= & \sum_{n\in \Z}  E(p+v+n)   \big| \kappa_{v}(p,n)    \big|^{2} +\sum_{n\in \Z}  E(p-v+n)   \big| \kappa_{-v}(p,n)    \big|^{2} - 2E(p) .
 \end{align*}
For $ \mathcal{E}_{p,v}:= \sum_{n\in \Z} E^{\frac{1}{2}}(p+v+n)\,   \big| \kappa_{v}(p,n)    \big|^{2}$,
\begin{align}\label{TheIsles}
2v^{2}= &\sum_{\pm}\sum_{n\in \Z}  E(p\pm v+n)   \big| \kappa_{\pm v}(p,n)    \big|^{2} - 2E(p)   \nonumber \\  =& \sum_{\pm } \sum_{n\in \Z}\Big(E^{\frac{1}{2}}(p\pm v+n)\,   \big| \kappa_{\pm v}(p,n)    \big|^{2}- \mathcal{E}_{p,\pm v}   \Big)^{2}\nonumber \\ &+\frac{1}{2}\Big(  \mathcal{E}_{p,v}- \mathcal{E}_{p,-v}   \Big)^{2}  +\frac{1}{2}\Big(  \mathcal{E}_{p,v}+ \mathcal{E}_{p,-v}   \Big)^{2}-2E(p)   . 
\end{align}
The sum of the last two terms on the right side of~(\ref{TheIsles}) is positive since 
\begin{align*}
 \frac{1}{2} \mathcal{E}_{p,v}+ \frac{1}{2}\mathcal{E}_{p,-v}-E^{\frac{1}{2}}(p)=&\frac{1}{2}\sum_{\pm}\big\langle \widetilde{\psi}_{p}\big| e^{\mp \textup{i}vX_{\mathbb{T}}}H_{\phi_{\pm}}^{\frac{1}{2}} e^{\pm \textup{i}vX_{\mathbb{T}}}\big| \widetilde{\psi}_{p} \big\rangle  -\big\langle \widetilde{\psi}_{p}\big|H_{\phi}^{\frac{1}{2}}\big| \widetilde{\psi}_{p} \big\rangle 
\\ =&\frac{1}{2}\sum_{\pm} {   }_{\scriptscriptstyle{Q}}\big\langle p\big| e^{\mp \textup{i}vX} H^{\frac{1}{2}} e^{\pm \textup{i}vX}   \big|p\big\rangle_{\scriptscriptstyle{Q}} -{   }_{\scriptscriptstyle{Q}}\big\langle p\big|  H^{\frac{1}{2}}    \big|p\big\rangle_{\scriptscriptstyle{Q}}\\  \geq & 0 .
\end{align*}
The last inequality holds because  $\frac{1}{2}\sum_{\pm}  e^{\mp \textup{i}vX} H^{\frac{1}{2}} e^{\pm \textup{i}vX} -H^{\frac{1}{2}}$ is a positive operator; see the proof of Part (2) from~\cite[Prop.4.1]{Dispersion}.  The above formal reasoning can be made rigorous by approximating the kets $\big|p\big\rangle_{\scriptscriptstyle{Q}} $ by elements in $L^{2}(\R)$.  The   operators  $\frac{1}{2}\sum_{\pm}  e^{\mp \textup{i}vX_{\mathbb{T}}}H_{\phi_{\pm}}^{\frac{1}{2}} e^{\pm \textup{i}vX_{\mathbb{T}}} -H_{\phi}^{\frac{1}{2}}$ correspond to the operation of $\frac{1}{2}\sum_{\pm}  e^{\mp \textup{i}vX} H^{\frac{1}{2}} e^{\pm \textup{i}vX} -H^{\frac{1}{2}}$ on the fiber spaces and thus must be positive. 

Since the bottom line of~(\ref{TheIsles}) is positive, it follows that $\sum_{n\in \Z}\big(E^{\frac{1}{2}}(p+ v+n)\,   \big| \kappa_{v}(p,n)    \big|^{2}- \mathcal{E}_{p, v}   \big)^{2}$ is bounded  by $2v^{2}$. 

\end{proof}

The proof Lem.~\ref{HillDog}  depends most essentially on bounding the difference between the terms $\kappa_{v}(p-k,n   ) \overline{\kappa}_{v}(p+k,n   )$ and  $|\kappa_{v}(p,n   )|^{2}$ through the derivative inequalities in Lem.~\ref{LemKappaDer}. Since there are an infinite number of terms in the sum over $n\in \Z$, I designate cut-offs   for the set of $n$ in which I apply the finer estimates  and control the remaining terms with the variance inequality of Lem.~\ref{LemVariance} and a Chebyshev bound.

\vspace{.4cm}

\begin{proof}[Proof of Part (2) from Lem.~\ref{HillDog}]
First, I will bound a single term from the sum.  By a first-order Taylor expansion around $p \in \R$, 
\begin{align}\label{Turek}
  \kappa_{v}\big(p-k&,n   \big) \overline{\kappa}_{v}\big(p+k,n   \big)-\big|\kappa_{v}\big(p,n   \big)\big|^{2}\nonumber \\ = &\textup{i}2 k \textup{Im}\Big[ \kappa_{v}\big(p,n   \big) \frac{\partial}{\partial p} \overline{\kappa}_{v}\big(p,n   \big) \Big] +    \int_{0}^{k}dw'\,\int_{0}^{w'}dw\frac{\partial^{2}}{\partial^{2} w} \Big(\kappa_{v}\big(p-w,n   \big)  \overline{\kappa}_{v}\big(p+w,n   \big)\Big).  
\end{align}
Applying Parts (1) and (2) of Lem.~\ref{LemKappaDer} to the absolute value of~(\ref{Turek}) yields constants $C_{1},C_{2}>0$ such that
\begin{align}\label{Kristiana}
 \Big| \kappa_{v}\big(p-k,n   \big) \overline{\kappa}_{v}\big(p+k,n   \big)&-\big|\kappa_{v}\big(p,n   \big)\big|^{2}\Big|\nonumber \\  \leq  & C_{1}| k| +\frac{C_{2} }{2}k^{2}\sup_{|w|\leq | k| } \Big(\frac{1}{|\theta(p+w)|^{2}}+\frac{1}{|\Theta(p+v+w )|^{2}}    \Big)\nonumber \\ & +C_{2}^{2} k^{4}\sup_{|w|\leq | k| } \Big(\frac{1}{|\theta(p+w)|^{2}}+\frac{1}{|\Theta(p+v+w )|^{2}}    \Big)^{2}\nonumber  \\ \leq &  C'\lambda^{ \frac{\gamma}{2}+1+\varrho},  
\end{align}
where the second inequality is for some $C'>0$ by the constraints $|k|\leq \lambda^{\varrho+\gamma+2}$ and $|\Theta(p)|, |\Theta(p+v)|\geq \lambda^{\frac{\gamma}{2}+1} $. 

For the full sum of terms, I have the bound 
\begin{align}\label{Jacobi}
\sum_{n\in \Z}\Big| \kappa_{v}\big(p- k,n   \big)  \overline{\kappa}_{v}\big(p &+k,n   \big)-\big|\kappa_{v}\big(p,n   \big)\big|^{2}\Big|\nonumber  \\  \leq  & \sum_{n\in A_{p,v}  } \Big| \kappa_{v}\big(p-k,n   \big) \overline{\kappa}_{v}\big(p+k,n   \big)-\big|\kappa_{v}\big(p,n   \big)\big|^{2}\Big|  \nonumber  \\ &+\Big(\sup_{\substack{p'\in \R , \\ n\notin A_{p',v}  }  }\frac{1}{\Big| E^{\frac{1}{2}}\big(p'+v+n \big)-\mathcal{E}_{p',v}    \Big|^{2}    }     \Big)\nonumber \\ & \times \sup_{p'\in \R} \sum_{n\in \Z  }\Big| E^{\frac{1}{2}}\big(p'+v+n \big)-\mathcal{E}_{p',v}    \Big|^{2}\big|\kappa_{v}\big(p',n   \big)\big|^{2},
\end{align}
where $\mathcal{E}_{p,v}\in \R^{+}$ and  $A_{p,v}\subset \Z$ are defined as
$$\mathcal{E}_{p,v}:= \sum_{n\in \Z} E^{\frac{1}{2}}(p+v+n)\,   \big| \kappa_{v}(p,n)    \big|^{2}\quad \text{and} \quad  A_{p,v}:=\Big\{ n\in \Z\,\Big|\,    \big| E^{\frac{1}{2}}(p+v+n)-\mathcal{E}_{p,v}\big|\leq  \lambda^{-\frac{\gamma}{6}-\frac{1}{3}-\frac{\varrho}{3}  }\Big\} .  $$   
In~(\ref{Jacobi}) I have applied Chebyshev's inequality and $2|xy|\leq |x|^{2}+|y|^{2}$ with $x=\kappa_{v}(p- k,n   )$ and $y=  \overline{\kappa}_{v}(p +k,n   )$ to bound the sum of terms with $n\in A_{p,v}^{c}$.

For the first term on the right side of~(\ref{Jacobi}), there are approximately $2 \lambda^{-\frac{\gamma}{6}-\frac{1}{3}-\frac{\varrho}{3}  }$ terms in the sum   since 
$E^{\frac{1}{2}}(p')\approx |p'|$ for $|p'|\gg 1$ by Part (1) of Lem.~\ref{CritEstimates}.   Moreover, I can apply~(\ref{Kristiana}) to bound each individual term in the sum.   For the second term on the right side of~(\ref{Jacobi}), I can apply Lem.~\ref{LemVariance} to bound the sum by $v^2$.   Putting these observations together, there is $C>0$ such that
$$ \sum_{n\in \Z}\Big| \kappa_{v}\big(p- k,n   \big)  \overline{\kappa}_{v}\big(p +k,n   \big)-\big|\kappa_{v}\big(p,n   \big)\big|^{2}\Big|\leq \big( C' +v^{2}\big) \lambda^{\frac{\gamma}{3}+\frac{2\varrho}{3}+\frac{2}{3}  }=\mathit{O}\big( \lambda^{\frac{\gamma}{2}+\frac{7}{6}}  \big) .    $$
 The order equality uses that $\varrho>\gamma>1$ .

\end{proof}

\subsection{Proof of Lemma~\ref{SemiToFull}}

The proof of Lem.~\ref{SemiToFull}  primarily involves bounding the difference between the semigroups $\Phi_{\lambda,t}^{(k)}:L^{1}(\R) $ and $\Upsilon_{\lambda,t}^{(k)}:L^{1}(\R)$.  For this task, it is convenient to introduce an intermediary semigroup $\Upsilon_{\lambda,t}^{(k),\prime}$ that has the same drift term as $\Phi_{\lambda,t}^{(k)}$ and the same jump term as $\Upsilon_{\lambda,t}^{(k)}$. Let $\Upsilon_{\lambda,t}^{(k),\prime}:L^{1}(\R)$ be the semigroup with generator $\mathcal{L}_{\lambda, k}^{\prime \prime}$ that acts on elements  $f\in \mathcal{T}:= \big\{ g\in L^{1}(\R)\,\big|\,\int_{\R}dp\, |p|\, |g(p)|<\infty \big\}    $ as
\begin{align}\label{FFFF}
\big(\mathcal{L}_{\lambda, k}^{\prime \prime} f\big)(p)= &-\frac{\textup{i}}{\lambda^{\varrho}}\Big(E\big(p-k\big)- E\big(p+ k\big)   \Big)f(p) -\mathcal{R}f(p)+\int_{\R}dp' J(p,p')f(p')   .
\end{align}
I will bound the difference $\Phi_{\lambda,t}^{(k)}-\Upsilon_{\lambda,t}^{(k),\prime}$ by means of a Duhamel equation and an application of Part (2) of Lem~\ref{HillDog}, and  I will bound the difference $\Upsilon_{\lambda,t}^{(k),\prime}-\Upsilon_{\lambda,t}^{(k)}$ though a pseudo-Poisson unraveling and an application of Part (1) of Lem.~\ref{HillDog}.   Since the inequalities in Lem.~\ref{HillDog} pertain to momenta bounded away from the lattice $\frac{1}{2}\Z$, I take  precautions though  Parts (3) and (4) of Prop.~\ref{PropMoreFiber} to ensure that momentum densities are not peaked in the region around the lattice.

\vspace{.4cm}

\begin{proof}[Proof of Lem.~\ref{SemiToFull}]
By the fiber decomposition for the classical dynamics, the density $\mathcal{P}^{(k)}_{\lambda, t}\in L^{1}(\R)$ is given by 
\begin{align}\label{Fasting}
\mathcal{P}^{(k)}_{\lambda, t}=\Upsilon_{\lambda,t}^{(k)}\mathcal{P}^{(k)}_{\lambda, 0} 
\end{align}
 for the semigroup $\Upsilon_{\lambda,t}^{(k)}$ with generator $\mathcal{L}_{\lambda,k}^{\prime}$ defined in~(\ref{FlowerBasket}).   Moreover, since the initial distribution $\mathcal{P}_{\lambda,0}(y,p)$ for the classical dynamics is $\delta_{0}(y)  |\frak{h}(p)|^{2}$,
$$\mathcal{P}^{(k)}_{\lambda, 0}(p) := \int_{\R^{2}}dxdp\,\mathcal{P}_{\lambda,0}(x,p)e^{i2xk}=|\frak{h}(p)|^{2}=[\check{\rho}_{\lambda}]^{(0)}(p),  $$
where the last equality uses that $[\check{\rho}_{\lambda}]^{(0)}(p):=\check{\rho}_{\lambda}(p,p)$ and $\check{\rho}_{\lambda}:=| \frak{h}\rangle\langle \frak{h}|$.   Since the maps $\Upsilon_{\lambda,t}^{(k)}:L^{1}(\R)$ are contractive, the difference in norm between $\mathcal{P}^{(k)}_{\lambda, t}=\Upsilon_{\lambda,t}^{(k)}[\check{\rho}_{\lambda}]^{(0)}$ and $\Upsilon_{\lambda,t}^{(k)}[\check{\rho}_{\lambda} ]^{(k)}_{\scriptscriptstyle{Q}}$ for all $|k|\leq \lambda^{\frac{\gamma}{2}+\frac{5}{4}+\varrho  }$, $\lambda<1$, and $t\in \R^{+}$ is smaller than
\begin{align}
 \left\| \mathcal{P}^{(k)}_{\lambda, t}- \Upsilon_{\lambda,t}^{(k)}[\check{\rho}_{\lambda} ]^{(k)}_{\scriptscriptstyle{Q}}  \right\|_{1}\leq  &\left\| [\check{\rho}_{\lambda}]^{(0)}- [\check{\rho}_{\lambda} ]^{(k)}_{\scriptscriptstyle{Q}} \right\|_{1}\nonumber  \\  \leq  &\left\| [\check{\rho}_{\lambda}]^{(0)}- [\check{\rho}_{\lambda} ]^{(k)} \right\|_{1}+ \left\| [\check{\rho}_{\lambda}]^{(k)}- [\check{\rho}_{\lambda} ]^{(k)}_{\scriptscriptstyle{Q}} \right\|_{1}\nonumber \\ = & \int_{\R}dp\,\Big| |\frak{h}(p)|^{2}-  \overline{\frak{h}(p-k)}  \frak{h}(p+k) \Big|       +\mathit{O}\big(\lambda^{\frac{1}{2}}\big)=\mathit{O}\big(\lambda^{\frac{1}{2}}\big). 
\end{align}
For the first equality above,  the approximation techniques from the proof of Lem.~\ref{StanToQuasi} can be used to show that the $L^{1}$ norm of $[\check{\rho}_{\lambda}]^{(k)}- [\check{\rho}_{\lambda} ]^{(k)}_{\scriptscriptstyle{Q}} $ is $\mathit{O}(\lambda^{\frac{1}{2}})$.   Moreover, the $ L^{1}$ difference between $|\frak{h}(p)|^{2}$ and  $ \overline{\frak{h}(p-k)}  \frak{h}(p+k) $ is $\mathit{O}(\lambda^{\frac{1}{2}})$ by the assumptions $\frak{h}:=e^{\textup{i}\mathbf{p}X}\frak{h}_{0}$  and $\| X \frak{h}_{0}\|_{2}=\int_{\R}dp\,\big| \frac{ d\frak{h}_{0}}{dp}(p)\big|^{2}<\infty $.

By the above remarks, it is  sufficient  to control the difference between the operation of the contractive semigroups $\Phi_{\lambda,t}^{(k)}$ and $\Upsilon_{\lambda,t}^{(k)}$ acting on $[\check{\rho}_{\lambda}]^{(k)}_{\scriptscriptstyle{Q}}$.  For the intermediary semigroup $\Upsilon_{\lambda,t}^{(k),\prime}:L^{1}(\R)$ with generator~(\ref{FFFF}),   the differences $\Phi_{\lambda,t}^{(k)}-\Upsilon_{\lambda,t}^{(k),\prime}$ and  $\Upsilon_{\lambda,t}^{(k),\prime}-\Upsilon_{\lambda,t}^{(k)}$ are bounded in parts (i) and (ii) below, respectively.  \vspace{.4cm}

\noindent (i).\hspace{.2cm} The difference between  $\Phi_{\lambda,t}^{(k)}$ and   $\Upsilon_{\lambda,t}^{(k),\prime}$  can be written in terms of the Duhamel equation
$$\Phi_{\lambda,t}^{(k)}-\Upsilon_{\lambda,t}^{(k),\prime}=\int_{0}^{t}dr \,  \Phi_{\lambda,t-r}^{(k)} (J_{k}-J)    \Upsilon_{\lambda,r}^{(k),\prime}.  $$
Let $A\subset \R^2$ be defined as in Lem.~\ref{HillDog} and $B\subset \R$ be defined as  $B=\{ p\in \R\,|\,|\theta(p)|> \lambda^{\frac{\gamma}{2}+1}    \} $.  Since $\Phi_{\lambda,t-r}^{(k)}$ is contractive in the $L^{1}$-norm, I have the first inequality below:
\begin{align}\label{Douche}
\Big\|  \Phi_{\lambda,t}^{(k)}&[\check{\rho}_{\lambda}]^{(k)}_{\scriptscriptstyle{Q}} -\Upsilon_{\lambda,t}^{(k),\prime}[\check{\rho}_{\lambda}]^{(k)}_{\scriptscriptstyle{Q}} \Big\|_{1}\nonumber \\  \leq & t\sup_{0\leq r\leq t}\int_{\R}dp   \sum_{n\in \Z}\int_{\R}dv\,j(v)\Big|\Big(\Upsilon_{\lambda,r}^{(k),\prime}[\check{\rho}_{\lambda}]^{(k)}_{\scriptscriptstyle{Q}}\Big)(p)\Big|\,  \Big| \kappa_{v }\big(p-k,\,n\big) \overline{\kappa}_{v }\big(p+k,\,n\big)- \big|\kappa_{v }\big(p,\,n\big)\big|^{2}    \Big|\nonumber  \\  \leq & 2t\sup_{0\leq r\leq t} \int_{\R}dp\int_{\R}dv\,1_{A^{c}}(p,v) j(v) \Big|\Big(\Upsilon_{\lambda,r}^{(k),\prime}[\check{\rho}_{\lambda}]^{(k)}_{\scriptscriptstyle{Q}}\Big)(p)\Big| \nonumber \\ & + \mathcal{R}t\big(1 +\frac{\sigma}{\mathcal{R}} \big)  \sup_{(p,v)\in A} \frac{1}{1  +v^2 }\sum_{n\in \Z}  \Big| \kappa_{v }\big(p-k,\,n\big) \overline{\kappa}_{v }\big(p+k,\,n\big)- \big|\kappa_{v }\big(p,\,n\big)\big|^{2}    \Big| .  
\end{align}
The second inequality above partitions the integration over $(p,v)\in \R^2$ into the domains $A$ and $A^c$, and for the domain $A^{c}$ applies the Cauchy-Schwarz inequality along with the fact that $\sum_{n\in \Z}\big|\kappa_{v }\big(p,\,n\big)\big|^{2}\leq 1$ for all $(p,v)\in \R^2$. For the domain $A\subset \R^{2}$, I have multiplied and divided by $1  +v^{2}$ and applied Holder's inequality in combination with $\big\|\Upsilon_{\lambda,r}^{(k),\prime}[\check{\rho}_{\lambda}]^{(k)}_{\scriptscriptstyle{Q}}\big\|_{1}\leq \big\|[\check{\rho}_{\lambda}]^{(k)}_{\scriptscriptstyle{Q}}\big\|_{1}\leq 1 $.  The last line of~(\ref{Douche}) is bounded by a constant multiple of $t\lambda^{\frac{\gamma}{2}+1}$ by Part (2) of Lem.~\ref{HillDog}.   
 I will bound the expression  on the third line of~(\ref{Douche}) in (I) below.

\vspace{.4cm}

\noindent (I). \hspace{.2cm} For the integrand on the third line of~(\ref{Douche}),
 \begin{align}\label{HayWire}
  \Big| \big(\Upsilon_{\lambda,r}^{(k),\prime}[\check{\rho}_{\lambda}]^{(k)}_{\scriptscriptstyle{Q}}\big)(p)\Big|\leq \frac{1}{2}\big(\Upsilon_{\lambda,r}^{(0),\prime}\mathcal{S}_{ k}[\check{\rho}_{\lambda}]^{(0)}_{\scriptscriptstyle{Q}}\big)(p)+\frac{1}{2}\big(\Upsilon_{\lambda,r}^{(0),\prime}\mathcal{S}_{-k}[\check{\rho}_{\lambda}]^{(0)}_{\scriptscriptstyle{Q}}\big)(p) ,  
  \end{align}
 where $\mathcal{S}_{q}:L^{1}(\R)$ is the shift operator by $q\in \R$: $\big(\mathcal{S}_{q}f\big)(p)=f(p-q)$.  The above inequality uses that $|(\Upsilon_{\lambda,r}^{(k),\prime}f)(p)|\leq  (\Upsilon_{\lambda,r}^{(0),\prime}|f|)(p)$ for every $f\in L^{1}(\R)$ and a.e. $p\in \R$ and Part (3) of Prop.~\ref{MiscFiber} to bound $[\check{\rho}_{\lambda}]^{(k)}_{\scriptscriptstyle{Q}}$ by the sum $\frac{1}{2}\mathcal{S}_{k}[\check{\rho}_{\lambda}]^{(0)}_{\scriptscriptstyle{Q}} +\frac{1}{2}\mathcal{S}_{-k}[\check{\rho}_{\lambda}]^{(0)}_{\scriptscriptstyle{Q}}$.
 With~(\ref{HayWire}) and the bound $1_{A^{c}}(p,v)\leq 1_{B^{c}}(p)+1_{B^{c}}(p+v)$, I have the first inequality below:
\begin{align*}
\int_{\R}dp\int_{\R}dv\,1_{A^{c}}(p,v) j(v)\Big| \big(\Upsilon_{\lambda,r}^{(k),\prime}[\check{\rho}_{\lambda}]^{(k)}_{\scriptscriptstyle{Q}}\big)(p)\Big|  \leq &  \frac{1}{2} \int_{\R}dp\int_{\R}dv\,\Big(1_{B^{c}}(p)+1_{B^{c}}(p+v)\Big) j(v)
\\  &  \times \Big(\big(\Upsilon_{\lambda,r}^{(0),\prime}\mathcal{S}_{k}[\check{\rho}_{\lambda}]^{(0)}_{\scriptscriptstyle{Q}}\big)(p)+\big(\Upsilon_{\lambda,r}^{(0),\prime}\mathcal{S}_{-k}[\check{\rho}_{\lambda}]^{(0)}_{\scriptscriptstyle{Q}}\big)(p)\Big)
\\  \leq &  \frac{1}{2} \int_{\R}dp\,\Big(\mathcal{R}1_{B^{c}}(p)+\varpi \lambda^{\frac{\gamma}{2}+1}\Big) 
\\  &  \times \Big(\big(\Upsilon_{\lambda,r}^{(0),\prime}\mathcal{S}_{k}[\check{\rho}_{\lambda}]^{(0)}_{\scriptscriptstyle{Q}}\big)\big(p\big)+\big(\Upsilon_{\lambda,r}^{(0),\prime}\mathcal{S}_{-k}[\check{\rho}_{\lambda}]^{(0)}_{\scriptscriptstyle{Q}}\big)\big(p\big)\Big)
\\ \leq &
 \lambda^{\frac{\gamma}{2}+1}\Big(\mathcal{R}e^{-\mathcal{R}r}\| \langle \check{\rho}_{\lambda}\rangle^{(0)}\|_{\infty}+2\varpi \Big).
\end{align*}
The second inequality follows by assumption (2) of List~\ref{Assumptions}.  To see the third inequality above, notice that $ \Upsilon_{\lambda,r}^{(0),\prime}= \Upsilon_{\lambda,r}^{(0)}$ is the Markovian semigroup with jump rate kernel $J$. 
 When contracted to the torus $\mathbb{T}=[-\frac{1}{2},\frac{1}{2})$, the process is still Markovian and has kernel $J_{\mathbb{T}}$. Thus, by Part (2) of Prop.~\ref{PropMoreFiber}, the density  $\Upsilon_{\lambda,r}^{(0),\prime}\mathcal{S}_{\pm k}[\check{\rho}_{\lambda}]^{(0)}_{\scriptscriptstyle{Q}} $ is equal to $ \big\langle  \Phi_{\lambda,t}(\mathcal{S}_{\pm k}\check{\rho}_{\lambda}\mathcal{S}_{\mp k})  \big\rangle^{(0)}  $ when contracted to the torus, and  I can then apply Part (4) of Prop.~\ref{PropMoreFiber} to obtain the bound.

\vspace{.4cm}

\noindent (ii).\hspace{.2cm}
Recall that the linear map $\widetilde{U}_{\lambda , t}^{(k)}:L^{1}(\R)$ is defined as  multiplication by the function $\widetilde{U}_{\lambda , t}^{(k)}(p):=e^{\frac{\textup{i}t}{\lambda^{\varrho}} 4 p k   } $ and define $T:=\mathcal{R}^{-1}J$.
 Also, let $\Upsilon_{\lambda,\xi,r, t}^{(k)}$ and $\Upsilon_{\lambda,\xi,r, t}^{(k),\prime}$ be defined analogously to $\Phi_{\lambda,\xi,r, t}^{(k)}$ in the proof of Thm.~\ref{ThmSemiClassical}   as the products
\begin{align*}
\Upsilon_{\lambda,\xi,r, t}^{(k)}(\rho):= &\widetilde{U}_{\lambda, t-t_{\mathcal{N}}}^{(k)} T \cdots  \widetilde{U}_{\lambda, t_{n+1}-t_{n}}^{(k)}T \widetilde{U}_{\lambda, t_{n}-r}^{(k)} ,  \\   \Upsilon_{\lambda,\xi,r, t}^{(k),\prime}(\rho):= & U_{\lambda, t-t_{\mathcal{N}}}^{(k)} T \cdots  U_{\lambda , t_{n+1}-t_{n}}^{(k)}T U_{\lambda, t_{n}-r}^{(k)},
\end{align*} 
where $\xi=(t_{1},t_{2},\cdots)\in (\R^{+})^{\infty} $  and $t_{n}\leq \cdots \leq t_{\mathcal{N}}$ are the values in the interval $(r,t]$.   The difference between  the maps  $\Upsilon_{\lambda, t}^{(k),\prime}$ and $\Upsilon_{\lambda, t}^{(k)}$ can be written in terms of  telescoping sums as
$$
\Upsilon_{\lambda, t}^{(k),\prime}-\Upsilon_{\lambda, t}^{(k)} =e^{-\mathcal{R}t}\sum_{\mathcal{N}=0}^{\infty}\mathcal{R}^{\mathcal{N}}\sum_{n=0}^{\mathcal{N}} \int_{0\leq t_{1} \dots  \leq t_{N}\leq t} \Big( \Upsilon_{\lambda, \xi, t_{n+1}, t}^{(k),\prime}\Upsilon_{\lambda,\xi, t_{n+1} }^{(k)}-\Upsilon_{\lambda, \xi, t_{n}, t}^{(k),\prime}\Upsilon_{\lambda,\xi, t_{n} }^{(k)} \Big),
$$
where I use the identifications $t_{0}:=0$ and $t_{\mathcal{N}+1}:=t$.  By the triangle inequality, I have the first inequality below:   
\begin{align}\label{ZigZag}
\Big\|\Upsilon_{\lambda,\xi, t}^{(k),\prime}[\check{\rho}_{\lambda}]^{(k)}_{\scriptscriptstyle{Q}}-\Upsilon_{\lambda,\xi, t}^{(k)}[\check{\rho}_{\lambda}]^{(k)}_{\scriptscriptstyle{Q}} \Big\|_{1} \leq & e^{-\mathcal{R}t} \sum_{\mathcal{N}=0}^{\infty}\mathcal{R}^{\mathcal{N}}\sum_{n=0}^{\mathcal{N}}\int_{0\leq t_{1} \dots \leq t_{N}\leq t}\nonumber\\ & \text{ }\quad \sup_{r\in[0, t_{n+1}-t_{n}]} \Big\| \Big( \widetilde{U}_{\lambda, r}^{(k)}- U_{\lambda, r}^{(k)} \Big)\Upsilon_{\lambda,\xi, t_{n} }^{(k)}[\check{\rho}_{\lambda}]^{(k)}_{\scriptscriptstyle{Q}} \Big\|_{1}\nonumber
\\  \leq & e^{-\mathcal{R}t}\sum_{\mathcal{N}=0}^{\infty}\mathcal{R}^{\mathcal{N}}\sum_{n=0}^{\mathcal{N}}\int_{0\leq t_{1} \dots  \leq t_{N}\leq t}  \lambda^{\frac{\gamma}{2}+1}\big(c_{1}+c_{2}\mathcal{R}(t_{n+1}-t_{n})  \big)\nonumber  \\  = &   \lambda^{\frac{\gamma}{2}+1}\mathcal{R}t \big(c_{1}+c_{2} \big),
\end{align}
for some $c_{1}, c_{2}>0$ determined implicitly below.  
The expression in the second line of~(\ref{ZigZag}) is bounded through the following inequalities:
\begin{align*}
 &\sup_{r\in[0, t_{n+1}-t_{n}]} \Big\| \Big( \widetilde{U}_{\lambda, r}^{(k)}- U_{\lambda, r}^{(k)} \Big)\Upsilon_{\lambda,\xi, t_{n} }^{(k)}[\check{\rho}_{\lambda}]^{(k)}_{\scriptscriptstyle{Q}} \Big\|_{1} \\ &\leq \int_{\R}dp\,\Big|\Upsilon_{\lambda,\xi, t_{n} }^{(k)}[\check{\rho}_{\lambda}]^{(k)}_{\scriptscriptstyle{Q}}(p)\Big|\sup_{r\in[0, t_{n+1}-t_{n}]}\Big|  e^{\frac{\textup{i}r}{\lambda^{\varrho}}4pk}-e^{-\frac{\textup{i}r}{\lambda^{\varrho}}(E(p-k)-E(p+k)    )}    \Big| 
\\   &\leq  2 \int_{|\theta|\leq \lambda^{\frac{\gamma}{2}+1} } \sum_{p\in \frac{1}{2}\Z+\theta }  \Big|\Upsilon_{\lambda,\xi, t_{n} }^{(k)}[\check{\rho}_{\lambda}]^{(k)}_{\scriptscriptstyle{Q}}(p)\Big|\\  &\hspace{.3cm}+\frac{t_{n+1}-t_{n}}{\lambda^{\varrho}}\int_{ \lambda^{\frac{\gamma}{2}+1}< |\theta|\leq  \frac{1}{4} }\Big(\sum_{p\in \frac{1}{2}\Z+\theta}  \Big|\Upsilon_{\lambda,\xi, t_{n} }^{(k)}[\check{\rho}_{\lambda}]^{(k)}_{\scriptscriptstyle{Q}}(p ) \Big|\Big)\Big(\sup_{p\in \frac{1}{2}\Z+\theta } \Big| E\big(p-k\big)-E\big(p+k\big) +4pk\Big| \Big)
\\  &  \leq  \big(\varpi+\|\langle \check{\rho}_{\lambda}  \rangle^{(0)}\|_{\infty}\big) \Big( \lambda^{\frac{\gamma}{2}+1+\iota}+ c\mathcal{R}(t_{n+1}-t_{n}) \lambda^{\frac{\gamma}{2}+1+\iota}\int_{ \lambda^{\frac{\gamma}{2}+1}< |\theta|\leq  \frac{1}{4} }\frac{1}{|\theta|}  \Big)\\  &  \leq  \lambda^{\frac{\gamma}{2}+1} \big(c_{1}+c_{2}\mathcal{R}(t_{n+1}-t_{n}) \big).
\end{align*}
For the second inequality above, I bounded the expression  $|  e^{\frac{\textup{i}r}{\lambda^{\varrho}}4pk}-e^{-\frac{\textup{i}r}{\lambda^{\varrho}}(E(p-k)-E(p+k)    )}   | $ by $2$ for  $p\in B^{c}$ and by $\frac{t_{n+1}-t_{n}}{\lambda^{\varrho}}| E(p-k)-E(p+k) +4pk| $  for  $p\in B$.  The third inequality uses Part (1) of Lem.~\ref{HillDog} to bound the supremum in the second term for some $c>0$ and the bound
\begin{align}\label{ShyLo}
 \sup_{\theta\in \mathbb{T}}\sum_{p\in \frac{1}{2}\Z+\theta }  \Big|\Upsilon_{\lambda,\xi, t_{n} }^{(k)}[\check{\rho}_{\lambda}]^{(k)}_{\scriptscriptstyle{Q}}(p)\Big| & \leq  \frac{1}{2}\sup_{\theta\in \mathbb{T}}\sum_{N\in \Z}  \sum_{\pm}\Big|\Upsilon_{\lambda,\xi, t_{n} }^{(0)} \mathcal{S}_{\pm k}[\check{\rho}_{\lambda}]^{(0)}_{\scriptscriptstyle{Q}}\big(\theta+\frac{1}{2}N \big)\Big|\nonumber \\  & \leq \frac{1}{2} \sum_{\pm} \big\|\big\langle  \Phi_{\lambda,\xi,t_{n}}\big(\mathcal{S}_{\pm k}\check{\rho}_{\lambda}\mathcal{S}_{\mp k} \big) \big\rangle \big\|_{\infty}\nonumber \\ &\leq \sum_{\pm}\|\langle \mathcal{S}_{\pm k}\check{\rho}_{\lambda}\mathcal{S}_{\mp k} \rangle \|_{\infty}+\frac{\varpi}{\mathcal{R}}  =\|\check{\rho}_{\lambda}\|_{\infty}+\frac{\varpi}{\mathcal{R}}.  
 \end{align}
The first two inequalities in~(\ref{ShyLo}) follow by the reasoning in (i), and the third follows by Lem.~\ref{HorseMeister}.

\end{proof}

\section{Central limit theorem for the classical process}\label{SecClassical}

This section concerns only the classical stochastic process $(Y_{t},K_{t})$ discussed in Sect.~\ref{SecPreClassical}.  It will be convenient to change the definition of the spatial component from  $Y_{t}= \frac{2}{\lambda^{\varrho}} \int_{0}^{t}dr\, K_{r}  $ to  $Y_{t}= 2 \int_{0}^{t}dr\, K_{r}  $.  For simplicity, I will assume that the initial state of the momentum process is $K_{0}=\mathbf{p}$ rather than the distribution with density  $|\frak{h}(p)|^{2}$ for $\frak{h}\in L^{2}(\R)$ defined above~(\ref{Gauss}).

  The analysis appearing here is a simplification of that in~\cite[Sect.5]{Dispersion}, which focused on the classical dynamics over arbitrarily long time scales.     The component $K_{t}$ is an autonomous Markov process with jump rates $J(p,p')$ from $p'$ to $p$, and the component $Y_{t}$ is an integral functional $Y_{t}=2\int_{0}^{t}dr\, K_{r}$.  The jump rates $J(p,p')$ have constant escape rates $\mathcal{R}:=\int_{\R}dp\, J(p,p')$, and I refer to the jump times as the \textit{Poisson times}.  The Poisson times are denoted by $t_{n}$ with the convention $t_{0}=0$, and $\mathcal{N}_{t}$ denotes the number of non-zero Poisson times up to time $t$.

 \subsection{Definitions and general discussion of the classical Markov process}

 I must introduce a number of technical definitions, which I will summarize in the list below. Let $S:\R\rightarrow \{\pm 1\}$ be the sign function.  A \textit{sign-flip} is  a Poisson time $t_{n}$  such that  $S(K_{t_{n}})=S(K_{t_{n+1}})$ and there are an odd number $m$ of sign changes leading up to $t_{n}$: $S(K_{t_{n-r}})=-S(K_{t_{n-r+1}})$ for $r\in [1,m]$ and $S(K_{t_{n-m-1}})=S(K_{t_{n-m}})$.  Note that, under this definition, a sign-flip time is not a hitting time with respect to the filtration generated by the process $K_{t}$, since the identification of a sign-flip time depends on a verification that the sign does not change again at the next Poisson time.  This awkward definition is formed to avoid counting occurrences in which the  momentum changes sign at successive pairs of Poisson times, which a detailed examination of the jump rates $J(p,p')$ shows is likely.  The double-flipping is a minor impediment to finding a more stable characterization for the sign behavior of the momentum process, and  I have discussed this issue in detail at the beginning of~\cite[Sect.5]{Dispersion}.  Define the $\tau_{m}$, $m\geq 0$ inductively to be the sequence of times such that $\tau_{0}=0$ and $\tau_{m+1}$ is the first time $t\in \R^{+}$ following $\tau_{m}$ for which  $t$ is a sign-flip or $|K_{t}| \notin \big[\frac{1}{2}|K_{ \tau_{m}}|,\frac{3}{2}|K_{ \tau_{m}}|  \big]$.  Introducing the cutoff for the deviation of the absolute value of the momentum over the interval $[\tau_{m},\tau_{m+1})$ is a technical precaution, which I use because the $\tau_{m}$'s are less frequent over time intervals in which the momentum is high $|K_{t}|\gg 1$.  I denote the number of non-zero $\tau_{m}$'s to have occurred up to time $t\in \R^{+}$ by $\mathbf{N}_{t}$.  Pick $\epsilon\in (0, \frac{2-\gamma}{2} )   $, and define $\varsigma$ to be the hitting time that $|K_{t}|$ jumps out of the interval $[\mathbf{p}-\lambda^{\epsilon}\mathbf{p} ,\mathbf{p}+\lambda^{\epsilon}\mathbf{p}] $.   The standard filtration generated by the process $K_{t}$  is denoted by $\mathcal{F}_{t}:=\sigma\big(K_{r}:\, 0\leq r\leq t\big)$.   Let  $\widetilde{\mathcal{F}}_{t}$ be the filtration given by
$$\widetilde{\mathcal{F}}_{t}= \sigma\Big (\tau_{m+1},\, K_{r}\,:\, 0\leq r\leq \tau_{m+1} \text{ for the $m\in \mathbb{N}$ with } t\in [\tau_{m},\tau_{m+1} ) \Big)  .       $$
When $t \in [\tau_{m}, \tau_{m+1})$ for some $m$,  the $\sigma$-algebra $\widetilde{\mathcal{F}}_{t}$ contains knowledge of the time $\tau_{m+1}$ and all information about the process $K_{t}$ up to time $\tau_{m+1}$.   For $\widetilde{\mathcal{F}}_{\lambda, s}:=\widetilde{\mathcal{F}}_{\frac{s}{\lambda^\gamma} } $ and $\Delta \tau_{m}:=\tau_{m+1}-\tau_{m}$, define the $\widetilde{\mathcal{F}}_{\lambda, s}$-adapted martingale

$$
\mathbf{m}_{\lambda, s}:= \lambda^{\frac{\gamma+3}{2}}\sum_{m=1}^{\mathbf{N}_{\frac{s}{\lambda^\gamma}}}\chi\big(\tau_{m}< \varsigma    \big)\,K_{\tau_{m}}\Big(\Delta \tau_{m}- \mathbb{E}\big[\Delta \tau_{m}\,\big|\, \widetilde{\mathcal{F}}_{\tau_{m}^{-}}\big]   \Big).
$$
At a glance, the above definitions are given by the following: 
\begin{eqnarray*}
& Y_{\lambda, s}  & \text{The normalized integral functional: $ Y_{\lambda, s}:=2 \lambda^{\frac{\gamma+3}{2}}\int_{0}^{\frac{s}{\lambda^\gamma} }dt\, K_{t}$}\\
&  t_{n}   & \text{$n$th Poisson time}\\
&   \mathcal{N}_{t}       & \text{Number of non-zero Poisson times up to time $t$       }\\ & \varsigma & \text{First time that $|K_{t}|$ jumps out of the interval $[\mathbf{p}-\lambda^{\epsilon}\mathbf{p} ,\mathbf{p}+\lambda^{\epsilon}\mathbf{p}] $ }\\
&  \tau_{m}     &   \text{Time of the $m$th sign-flip} \\
&  \Delta\tau_{m}    &   \text{Time elapsed between the $m$th and $m+1$th sign-flip: $\Delta\tau_{m}=\tau_{m+1}-\tau_{m}$  } \\
 &\mathbf{N}_{t}     & \text{Number of $\tau_{m}$'s up to time $t\in \R^{+}$ }\\
&\mathcal{F}_{t}    &    \text{Information up to time $t$ } \\ 
&\widetilde{\mathcal{F}}_{t}    &    \text{Information up to the time of the sign-flip following $t$  }\\
& \mathbf{m}_{\lambda,s}& \text{Martingale with respect to $\widetilde{\mathcal{F}}_{\lambda, s}=\widetilde{\mathcal{F}}_{\frac{s}{\lambda^\gamma} }$ that approximates $Y_{\lambda, s}$ for $\lambda\ll 1$  }
\end{eqnarray*}

\subsection{ Preliminary Analysis }

Define $\nu:=\frac{\alpha \mathcal{R}}{4}$.   The following proposition is from~\cite[Prop.6.3]{Dispersion} and provides a few inequalities related to the time periods $[\tau_{m},\tau_{m+1})$ between successive times $\tau_{m}$.  The assumption in the statement of Prop.~\ref{TimeFlip} that $K_{t}$ does not change signs at the first Poisson time corresponds to the information known if $\tau_{m}$ is a sign-flip.  Part (1) of Prop.~\ref{TimeFlip} states that the moments for $\frac{\nu(\tau_{m+1}-\tau_{m})  }{ |K_{\tau_{m}}|  }$ are approximately equal to those of a mean $1$ exponential, and Part (2) bounds the amount of time that the momentum process spends with the opposite sign before making a sign-flip.

\begin{proposition} \label{TimeFlip}
Let $\zeta>0$, $K_{0}=p$ for $|p|\gg 1$, and $K_{t}$ be conditioned not to make a sign change at the first Poisson time (i.e. $ S(K_{0})=S(K_{t_{1}})$).  Define $\tau$ to be  the first time that either $K_{t}$ has a sign-flip or  $|K_{t}|$ jumps out of the set $[\frac{1}{2}|p|,\,\frac{3}{2}|p|]$ depending on what occurs first.  For fixed $\zeta,\, n>0$, there exists a $C>0$  such that the following inequalities hold for all $p$:

\begin{enumerate}
\item $\Big|\mathbb{E}\big[ \big(\frac{\tau}{|p|}\big)^{n}  \big]-n!\nu^{-n}   \Big|\leq C |p|^{\zeta-1},  $

\item 
$ \mathbb{E}\Big[\int_{0}^{\tau}dr\,\chi\big(S(K_{r})\neq S(p)   \big)   \Big] \leq C. $

\end{enumerate}

\end{proposition}

The following lemma bounds the moments for the longest time interval $\Delta\tau_{m}$ to occur up to time $\min\{\frac{T}{\lambda^{\gamma}},\varsigma\}$, and the proof follows easily from Part (1) of Prop.~\ref{TimeFlip}.

\begin{lemma}\label{LemIntLen}
For any $n\geq 1$ and $\iota>0$, there is a $C$ such that for all $\lambda<1$, 
$$ \mathbb{E}\Big[   \Big( \sup_{1\leq m\leq \mathbf{N}_{\frac{T}{\lambda^\gamma}}}\chi(\tau_{m}<\varsigma) \Delta\tau_{m} \Big)^{n}      \Big]\leq C\lambda^{-n-\iota} .$$
\end{lemma}

\begin{proof}
Pick $\iota\in (0, \gamma )$.  Since $|K_{t}|\leq 2\mathbf{p}=\frac{2  \mathbf{p}_{0}}{\lambda  } $  for $t\leq \varsigma$, I have the following inequality:   
\begin{align} \label{Libs}
  \mathbb{E}\Big[   \Big( \sup_{1\leq m\leq \mathbf{N}_{\frac{T}{\lambda^\gamma}}}\chi(\tau_{m}<\varsigma) \Delta\tau_{m}\Big)^{n}      \Big]\leq \big(\frac{2  \mathbf{p}_{0}}{\lambda  } \big)^{n}  \mathbb{E}\Big[   \Big( \sup_{1\leq m\leq \mathbf{N}_{\frac{T}{\lambda^\gamma}}}\chi(\tau_{m}<\varsigma)\frac{\Delta\tau_{m}}{|K_{\tau_{m}}|} \Big)^{n}      \Big].
  \end{align}
   For $u > \frac{\gamma}{\iota}$, the expectation on the right side of~(\ref{Libs}) can be bounded by standard techniques: 
  \begin{align}\label{Jingle}
 \mathbb{E}\Big[   \Big( \sup_{1\leq m\leq \mathbf{N}_{\frac{T}{\lambda^\gamma}}}\chi(\tau_{m}<\varsigma)\frac{\Delta\tau_{m}}{|K_{\tau_{m}}|} \Big)^{n}      \Big]   \leq & \mathbb{E}\Big[   \sum_{m=1}^{ \mathbf{N}_{\frac{T}{\lambda^\gamma}}}  \chi(\tau_{m}<\varsigma) \big( \frac{\Delta\tau_{m}}{|K_{\tau_{m}}|} \big)^{n u}    \Big]^{\frac{1}{u}}\nonumber \\ & = \mathbb{E}\Big[   \sum_{m=1}^{ \mathbf{N}_{\frac{T}{\lambda^\gamma}}} \chi(\tau_{m}<\varsigma)\mathbb{E}\Big[\big(\frac{\Delta\tau_{m}}{|K_{\tau_{m}}|}  \big)^{nu}\,\Big|\,\widetilde{\mathcal{F}}_{\tau_{m}}   \Big]      \Big]^{\frac{1}{u}} \nonumber \\ & \leq  C_{nu}^{\frac{1}{u}}\mathbb{E}\big[   \mathbf{N}_{\frac{T}{\lambda^\gamma}}     \big]^{\frac{1}{u}}\nonumber \\ &\leq   C_{nu}^{\frac{1}{u}}\mathbb{E}\big[  \mathcal{N}_{\frac{T}{\lambda^\gamma}}     \big]^{\frac{1}{u}}= C_{nu}^{\frac{1}{u}}\big(\frac{T\mathcal{R}}{\lambda^{\gamma}})^{\frac{1}{u}}=\mathit{O}(\lambda^{-\iota}), 
 \end{align}   
 where  the second inequality holds for some $C_{nu}>0$ by Part (1) of Prop.~\ref{TimeFlip}. The first inequality in~(\ref{Jingle}) bounds the supremum by a sum and applies Jensen's inequality, and the third inequality  uses that the sign-flip times occur at Poisson times.   The second equality holds since the Poisson times occur with rate $\mathcal{R}$.

\end{proof}

The following proposition is from~\cite[Prop.4.2]{Dispersion}.  Notice that Parts (1) and (2) are analogous to (1) and (2) of Prop.~\ref{LemSubMartBasic}.   For Part (2), recall that the square of a positive submartingale is also a submartingale.  

\begin{proposition}\label{SubMart} Define $\varsigma:=\int_{\R}dv\,\frac{j(v)}{\mathcal{R}}v^{4}$.
\begin{enumerate}
\item  The square root   energy process,  $E^{\frac{1}{2}}(K_{t})$, is  a submartingale.

\item The increasing part of the  Doob-Meyer decomposition for the submartingale $E(K_{t})$ has the form   $E(K_{0})+\sigma t $.  

\item The second moment of  $E(K_{t})$ has the bound 
$$  \mathbb{E}\big[ E^2(K_{t})   \big]\leq   \mathbb{E}\big[ E^2(K_{0})   \big]+3\sigma t \mathbb{E}\big[ E(K_{0})   \big]+\varsigma\mathcal{R}t +\frac{3}{2}\sigma^2 t^2  .   $$

\end{enumerate}

\end{proposition}

The lemma below states that the probably of the process $|K_{t}|$ leaving the interval $[\mathbf{p}-\lambda^{\epsilon} \mathbf{p},\mathbf{p}+\lambda^{\epsilon} \mathbf{p}]$  before time $\frac{T}{\lambda^{\gamma}}$ is small for $\epsilon\in [0,\frac{2-\gamma}{2})$ and $\lambda\ll 1$.  In particular, this implies that the $\tau_{m}$'s  before time $\frac{T}{\lambda^{\gamma}}$ are likely to be all sign-flips.

\begin{lemma}\label{LemTerm}
For small enough $\lambda>0$,  
$$ \mathbb{P}\Big[\sup_{0\leq t\leq \frac{T}{\lambda^{\gamma} } }  \big| |K_{t}|-|\mathbf{p}|\big|>\lambda^{\epsilon} \mathbf{p}  \Big] \leq  \frac{64\sigma T}{\mathbf{p}_{0}^2}\lambda^{2-\gamma-2\epsilon} . $$

\end{lemma}

\begin{proof}

Since $K_{0}=\mathbf{p}$ the difference $|K_{t}|-|\mathbf{p}|$ can be written as
\begin{align*} \Big(|K_{t} |- E^{\frac{1}{2}}(K_{t})\Big)   + \Big( E^{\frac{1}{2}}(K_{t})-  E^{\frac{1}{2}}(K_{0}) \Big) &+\Big(  E^{\frac{1}{2}}(K_{0})-|K_{0}|  \Big)   , 
\end{align*}
I have the inequality
$$ \Big| |K_{t}|-|\mathbf{p}|   \Big| \leq 2 B+\Big| E^{\frac{1}{2}}(K_{t})-  E^{\frac{1}{2}}(K_{0}) \Big|    , $$
where $B:=\sup_{p\in \R}\big| E^{\frac{1}{2}}(p) -|p|   \big|$.
The above inequality implies that
\begin{align}\label{Tracy}
\mathbb{P}\Big[\sup_{0\leq t\leq \frac{T}{\lambda^{\gamma} } }  \big| |K_{t}|-|\mathbf{p}|\big|>\lambda^{\epsilon} \mathbf{p}  \Big]\leq  \mathbb{P}\Big[\sup_{0\leq t\leq \frac{T}{\lambda^{\gamma} } }  \big| E^{\frac{1}{2}}(K_{t})-E^{\frac{1}{2}}(K_{0})\big|>\lambda^{\epsilon} \frac{ \mathbf{p} }{2 }\Big]  
\end{align}
since $B$ is smaller than  $\lambda^{\epsilon} \frac{\mathbf{p} }{2}=\lambda^{\epsilon-1} \frac{\mathbf{p}_{0} }{2}  $ for small enough $\lambda$.  By Chebyshev inequality, I have the first  inequality below:
\begin{align}\label{Splash}
\lambda^{2\epsilon} \frac{ \mathbf{p}^{2} }{4}\mathbb{P}\Big[\sup_{0\leq t\leq \frac{T}{\lambda^{\gamma} } }  \big| E^{\frac{1}{2}}(K_{t})-E^{\frac{1}{2}}(K_{0})\big|>\lambda^{\epsilon} \frac{ \mathbf{p} }{2  }\Big]\leq & \mathbb{E}\Big[\sup_{0\leq t\leq \frac{T}{\lambda^{\gamma} } }  \big| E^{\frac{1}{2}}(K_{t})-E^{\frac{1}{2}}(K_{0})\big|^{2}\Big]\nonumber  \\ \leq & \mathbb{E}\Big[\sup_{0\leq t\leq \frac{T}{\lambda^{\gamma} } } g_{t}^{2}+\sup_{0\leq t\leq \frac{T}{\lambda^{\gamma} } }  \big| E^{\frac{1}{2}}(K_{t})-E^{\frac{1}{2}}(K_{0})\big|_{+}^{2}\Big]\nonumber \\ \leq & 4 \mathbb{E}\Big[g_{\frac{T}{\lambda^{\gamma}}}^{2}+ \big| E^{\frac{1}{2}}(K_{\frac{T}{\lambda^{\gamma}}})-E^{\frac{1}{2}}(K_{0})\big|_{+}^{2}\Big]\nonumber  \\ \leq & 8 \mathbb{E}\left[  E(K_{\frac{T}{\lambda^{\gamma} }})-E(K_{0})\right]\nonumber \\ =&\frac{8\sigma T}{\lambda^{\gamma} },
\end{align}
where $g_{t}$ is the martingale part in the Doob-Meyer decomposition for $E^{\frac{1}{2}}(K_{t})$, and $|x|_{+}:=x1_{x\geq 0}$ for $x\in \R$.  The second inequality follows by neglecting the increasing part in the Doob-Meyer decomposition when $ E^{\frac{1}{2}}(K_{t})-E^{\frac{1}{2}}(K_{0}) $ is negative.  Two applications of Doob's maximal inequality gives  the third inequality in~(\ref{Splash}), and   the equality is by Part (2) of Prop.~\ref{SubMart}.   The fourth inequality in~(\ref{Splash}) uses the definition for $g_{t}$ and that $ E^{\frac{1}{2}}(K_{t})-E^{\frac{1}{2}}(K_{0})$ is a submartingale.  Thus, I have that
$$ \mathbb{P}\Big[\sup_{0\leq t\leq \frac{T}{\lambda^{\gamma} } }  \big| E^{\frac{1}{2}}(K_{t})-E^{\frac{1}{2}}(K_{0})\big|>\lambda^{\epsilon} \frac{ \mathbf{p} }{2  }\Big] \leq  \frac{32\sigma T}{\mathbf{p}_{0}^2}\lambda^{2-\gamma-2\epsilon},  $$
and for small enough $\lambda$, the first term on the right side of~(\ref{Tracy}) is smaller than  $\frac{32\sigma T}{\mathbf{p}_{0}^2}\lambda^{2-\gamma-2\epsilon}$ also.  

\end{proof}

\begin{lemma} \label{Crave}
  As $\lambda\searrow 0 $ there is convergence in probability 
  $$\sup_{0\leq s\leq T}\Big| Y_{\lambda, s}- 2 \lambda^{\frac{\gamma+3}{2}} \sum_{m=1 }^{\mathbf{N}_{\frac{T}{\lambda^\gamma}}}\chi(\tau_{m}<\varsigma) S(K_{\tau_{m}})\int_{\tau_{m}}^{\tau_{m+1}  }dr\, |K_{r}|\Big| \Longrightarrow 0 .    $$

 \end{lemma}
 
\begin{proof} 
By Lem.~\ref{LemTerm} the probability of the event $\varsigma\geq \frac{T}{\lambda^{\gamma} }$ is close to one for $\lambda \ll 1$, and thus with probability close to one, the following inequality holds: 
 $$\sup_{0\leq s\leq T}\Big| Y_{\lambda, s}- 2 \lambda^{\frac{\gamma+3}{2}} \sum_{m=1 }^{\mathbf{N}_{\frac{s}{\lambda^\gamma}}} \chi(\tau_{m}<\varsigma)\int_{\tau_{m}}^{\tau_{m+1}  }dt\, K_{t}\Big|\leq   4 \lambda^{\frac{\gamma+3}{2}}\sup_{0\leq m\leq \mathbf{N}_{\frac{T}{\lambda^\gamma}}}\chi(\tau_{m}<\varsigma) \int_{\tau_{m}}^{\tau_{m+1}  }dt\, |K_{t}|. $$
The right side goes to zero by the argument in the proof of Lem.~\ref{LemLindberg}.  
Over an interval $t\in [\tau_{m},\tau_{m+1})$, the process $K_{t}$ maintains the same sign except for isolated Poisson times at which it jumps to the opposite sign and back again at the next Poisson time, and thus 
 \begin{align*}
 \sup_{0\leq s\leq T}\Big|   2 \lambda^{\frac{\gamma+3}{2}} \sum_{m=1 }^{\mathbf{N}_{\frac{s}{\lambda^\gamma}}} \chi(\tau_{m}<\varsigma)\int_{\tau_{m}}^{\tau_{m+1}  }dt\, K_{t}&-2 \lambda^{\frac{\gamma+3}{2}}\sum_{m=1 }^{\mathbf{N}_{\frac{s}{\lambda^\gamma}}}\chi(\tau_{m}<\varsigma) S(K_{\tau_{m}}) \int_{\tau_{m}}^{\tau_{m+1}  }dt\, |K_{t}| \Big| \\  & \leq  4 \lambda^{\frac{\gamma+3}{2}}\sum_{m=1}^{\mathbf{N}_{\frac{T}{\lambda^\gamma}}} \chi(\tau_{m}<\varsigma) \int_{\tau_{m}}^{\tau_{m+1}}dt\,\chi\big(K_{t}\neq K_{\tau_{m}}\big)|K_{t}|  .  
 \end{align*}

 The inequalities below bound the total duration of the premature sign changes.  
 The first inequality below holds since $|K_{t}|\in \big[\frac{1}{2}|K_{\tau_{m}}|,\,\frac{3}{2}|K_{\tau_{m}}|\big]$ for $t\in [\tau_{m},\tau_{m+1})$: 
\begin{align}\label{KindaBlue}
2 \lambda^{\frac{\gamma+3}{2}}\mathbb{E}\Big[\sum_{m=1}^{\mathbf{N}_{\frac{T}{\lambda^\gamma}}} \chi(\tau_{m}<\varsigma) &\int_{\tau_{m}}^{\tau_{m+1}}dt\, \chi\big(K_{t}\neq K_{\tau_{m}}\big)|K_{t}|\,\Big]\nonumber  \\ &\leq  3 \lambda^{\frac{\gamma+3}{2}} \mathbb{E}\Big[\sum_{m=1}^{\mathbf{N}_{\frac{T}{\lambda^\gamma}}}\chi(\tau_{m}<\varsigma)|K_{\tau_{m}}|\,\mathbb{E}\Big[\int_{\tau_{m}}^{\tau_{m+1}}dt\,\chi\big(K_{t}\neq K_{\tau_{m}}\big)\,\Big|\,\widetilde{\mathcal{F}}_{\tau_{m}^{-}}\Big]\Big]\nonumber   \\ &\leq 3C \lambda^{\frac{\gamma+3}{2}} \mathbb{E}\Big[ \sum_{m=1}^{\mathbf{N}_{\frac{T}{\lambda^\gamma}}}\chi(\tau_{m}<\varsigma)|K_{\tau_{m}}| \Big]\nonumber \\ &\leq  6C\nu \lambda^{\frac{\gamma+3}{2}}\mathbb{E}\Big[\sum_{m=1}^{\mathbf{N}_{\frac{T}{\lambda^\gamma}}}\chi(\tau_{m}<\varsigma)\mathbb{E}\big[\Delta\tau_{m}   \,\big|\,\widetilde{\mathcal{F}}_{\tau_{m}^{-}}\big]  \Big].
\end{align}
The second inequality in~(\ref{KindaBlue}) holds for some $C>0$ by Part (2) of Prop.~\ref{TimeFlip}.  For the third inequality, I have used that $\mathbb{E}\big[\Delta\tau_{m}   \,\big|\,\widetilde{\mathcal{F}}_{\tau_{m}^{-}}\big]\approx \nu^{-1}|K_{\tau_{m}} |$ for large  $| K_{\tau_{m}}  |$ by Part (1) of Prop.~\ref{TimeFlip}, and I have  doubled the bound to cover the error. 
By removing the nested conditional expectations, the bottom line of~(\ref{KindaBlue}) is equal to the following:
\begin{align}
6C\nu  \lambda^{\frac{\gamma+3}{2}}\mathbb{E}\Big[\sum_{m=1}^{\mathbf{N}_{\frac{T}{\lambda^\gamma}}}\chi(\tau_{m}<\varsigma)\Delta\tau_{m}    \Big]  \leq & 6C\nu \lambda^{\frac{\gamma+3}{2}}\mathbb{E}\Big[\sum_{m=1}^{\mathbf{N}_{\frac{T}{\lambda^\gamma}}-1}\Delta\tau_{m}    +\sup_{1\leq m\leq \mathbf{N}_{\frac{T}{\lambda^\gamma}} }\chi(\tau_{m}<\varsigma)\Delta\tau_{m}    \Big]\nonumber  \\  \leq & 6C\nu T\Big(\lambda^{\frac{3}{2}-\frac{\gamma}{2} }+\mathit{O}(\lambda^{\frac{\gamma}{2}+\frac{1}{2}-\iota } ) \Big) \longrightarrow 0,
\end{align}
where the convergence to zero as $\lambda\searrow 0$ holds for any $\iota \in (0,1)$. The last inequality uses that the $\Delta\tau_{m}$'s sum up to less than $\frac{T}{\lambda^{\gamma}}$ for the first term and Lem.~\ref{LemIntLen} for the second term.

 \end{proof}

Recall that $\nu:=\frac{\alpha \mathcal{R}}{4}$ and  $\vartheta:=\frac{16\mathbf{p}^{3}_{0}}{\alpha \mathcal{R} }=\frac{4\mathbf{p}^{3}_{0}}{\nu}$.  The process  $\big[\mathbf{m}_{\lambda},\mathbf{m}_{\lambda}\big]_{s}$ refers to the quadratic variation of the martingale $\mathbf{m}_{\lambda}$.   
  
\begin{lemma}\label{LemMartApprox}
In the limit $\lambda\searrow 0$ there are the following convergences in probability:
\begin{enumerate}
 \item $\sup_{0\leq s\leq T}\big| \mathbf{m}_{\lambda, s}-Y_{\lambda,s}   \big|\Longrightarrow 0 $,

\item $\sup_{0\leq s\leq T}\Big|   \big[\mathbf{m}_{\lambda},\mathbf{m}_{\lambda}\big]_{s}-s\vartheta    \Big| \Longrightarrow 0$.
\end{enumerate}
 
 \end{lemma}

\begin{proof}\text{  }\\  Part (1):  I can approximate $Y_{\lambda, s}$ by the expression $$ \lambda^{\frac{\gamma+3}{2}}\sum_{m=1 }^{\mathbf{N}_{\frac{s}{\lambda^\gamma}}} \chi(\tau_{m}<\varsigma    )S(K_{\tau_{m}})\int_{\tau_{m}}^{\tau_{m+1} }dt\, |K_{t}|$$ for $\lambda\ll 1$ by Lem.~\ref{Crave}.  The result follows by showing the convergences in probability (i)-(iii) below.   The differences in (i) and (ii) involve coarse-graining approximations in which the random time intervals $\Delta\tau_{m}$ are parsed into shorter intervals $\Delta_{m,n}$ with duration on the order $\mathit{O}\big(|K_{\tau_{m}}|^{\beta}\big)$ for some $\beta$ chosen from  the interval $(0,\frac{1}{3})$.   The proofs of the convergences (i) and (ii)   do not simplify much from the proof of~\cite[Lem.3.2]{Dispersion}, so I do not include them.   I introduce the following notations:
\begin{eqnarray*}
\omega(m)&:=& \left\lfloor \big| K_{\tau_{m}}\big|^{\beta }\right\rfloor, \\
L_{m}&:=&   \left\lfloor \frac{\mathcal{N}_{\tau_{m+1} } -\mathcal{N}_{\tau_{m}} }{\omega(m) }\right\rfloor , \\
\Gamma_{m,n}&:=&t_{\mathcal{N}_{\tau_{m}}+nL_{m}},\\
\Delta_{m,n}&:=&\Gamma_{m,n+1} -\Gamma_{m,n}.
\end{eqnarray*}

There are the following convergences to zero in probability as $\lambda\searrow 0$:
\begin{enumerate}[(i).]
\item $ \sup_{0\leq s\leq T}\Big|  2 \lambda^{\frac{\gamma+3}{2}}\sum_{m=1}^{\mathbf{N}_{\frac{s}{\lambda^\gamma}}} \chi(\tau_{m}<\varsigma    )S(K_{\tau_{m}})\Big(\int_{\tau_{m}}^{\tau_{m+1} }dt\,|K_{t}|- \sum_{n=0}^{L_{m}-1}\Delta_{m,n}|K_{\Gamma_{m,n} } |\Big)\Big|\Longrightarrow 0, $

\item $ \sup_{0\leq s\leq T}\Big|2 \lambda^{\frac{\gamma+3}{2}}\sum_{m=1   }^{\mathbf{N}_{\frac{s}{\lambda^\gamma}}}\chi(\tau_{m}<\varsigma    )\Big( S(K_{\tau_{m}})\sum_{n=0}^{L_{m}-1}\Delta_{m,n}|K_{\Gamma_{m,n} }|-  K_{\tau_{m}}\Delta\tau_{m}  \Big)   \Big| \Longrightarrow 0, $

\item $  \sup_{0\leq s\leq T}\Big| 2 \lambda^{\frac{\gamma+3}{2}}\sum_{m=1 }^{\mathbf{N}_{\frac{s}{\lambda^\gamma}}}\chi(\tau_{m}<\varsigma    ) K_{\tau_{m}}\mathbb{E}\big[ \Delta\tau_{m} \,|\,\widetilde{\mathcal{F}}_{\tau_{m}^{-}}  \big] \Big| \Longrightarrow 0.  $
\end{enumerate}
The expression $ \mathbf{m}_{\lambda,s}= 2 \lambda^{\frac{\gamma+3}{2}}\sum_{m=1   }^{\mathbf{N}_{\frac{s}{\lambda^\gamma}}}\chi(\tau_{m}<\varsigma    )K_{\tau_{m}}(\Delta\tau_{m} -\mathbb{E}\big[ \Delta\tau_{m} \,|\,\widetilde{\mathcal{F}}_{\tau_{m}^{-}}  \big] ) $ is obtained by the  right term in  (ii) minus the expression in (iii).\vspace{.3cm}

\noindent (iii).\hspace{.15cm}  I will first show that there is a vanishing error in replacing the terms $\mathbb{E}\big[ \Delta\tau_{m}\,\big|\,\widetilde{\mathcal{F}}_{\tau_{m}^{-}}\big]$ in the expression by $\nu^{-1} |K_{\tau_{m}}|$.     To bound the difference, I apply  Part (1) of Prop.~\ref{TimeFlip} to get the first and second inequalities below for $C>0$ depending on my choice of $0<\zeta<\frac{1}{2}$:  
\begin{align}\label{Course}
\mathbb{E}\Big[ \sup_{0\leq s\leq T} 2 \lambda^{\frac{\gamma+3}{2}} \sum_{m=1}^{\mathbf{N}_{\frac{s}{\lambda^{\gamma}}}} \chi(& \tau_{m}<\varsigma    ) |K_{\tau_{m}}|\, \big|\mathbb{E}\big[ \Delta\tau_{m}\,\big|\,\widetilde{\mathcal{F}}_{\tau_{m}^{-}}\big]-\nu^{-1} |K_{\tau_{m}}   |\big| \Big]\nonumber \\ \leq & 2C \lambda^{\frac{\gamma+3}{2}}\mathbb{E}\Big[\sum_{m=1}^{\mathbf{N}_{\frac{T}{\lambda^{\gamma}}}}\chi(\tau_{m}<\varsigma    )\big|K_{\tau_{m}}\big|\, \big|K_{\tau_{m}}\big|^{\zeta}  \Big]\nonumber  \nonumber \\ \leq &  4C\nu  \lambda^{\frac{\gamma+3}{2}}\mathbb{E}\Big[\sum_{m=1}^{\mathbf{N}_{\frac{T}{\lambda^{\gamma}}}}\chi(\tau_{m}<\varsigma    ) \mathbb{E}\big[ \Delta\tau_{m}\,\big|\,\widetilde{\mathcal{F}}_{\tau_{m}^{-}}\big]\big|K_{\tau_{m}}\big|^{\zeta}    \Big]\nonumber  \\ \leq & 4C\nu \big| \mathbf{p} \big|^{\zeta} \lambda^{\frac{\gamma+3}{2}}\Big(\mathbb{E}\Big[\sum_{m=1}^{\mathbf{N}_{\frac{T}{\lambda^{\gamma}}}-1 }\Delta\tau_{m}   \Big]+\mathbb{E}\Big[ \sup_{0\leq m\leq \mathbf{N}_{\frac{T}{\lambda^{\gamma}}}}\chi(\tau_{m}<\varsigma    )  \Delta \tau_{m}   \Big]\Big)\nonumber  \\ \leq & 4C\nu \big| \mathbf{p}\big|^{\zeta} \lambda^{\frac{\gamma+3}{2}}\Big(\frac{T}{\lambda^{\gamma}} +C'\lambda^{-1-\iota} \Big).
\end{align}
The third inequality follows by removing the nested conditional expectations, using that  $|K_{t}|\leq 2\mathbf{p}$ for $t< \varsigma$, and bounding the last term in the sum by the largest.  The first and second expressions in the fourth line of~(\ref{Course}) are bounded respectively using that $ \sum_{m=1}^{\mathbf{N}_{t}-1} \Delta\tau_{m}\leq t $ and for some $C'$ given $\iota \in (0, \frac{\gamma-1}{2})$ by Lem.~\ref{LemIntLen}.  The bottom line of~(\ref{Course}) decays with order $\lambda^{\frac{\gamma-1}{ 2}  } $ for small $\lambda$.

I am left to bound the expression
$$ \sup_{0\leq s\leq T}2 \lambda^{\frac{\gamma+3}{2}}\Big|\sum_{m=1}^{\mathbf{N}_{\frac{s}{\lambda^{\gamma}}}} \chi(\tau_{m}<\varsigma    ) K_{\tau_{m}} |K_{\tau_{m}}| \Big|.$$ 
By Lem.~\ref{LemTerm}, the probability of the event $\varsigma\geq \frac{T}{\lambda^{\gamma}}$  converges to one for small $\lambda$.    For $\varsigma> t$, the values $K_{\tau_{m}}$ change sign for each $\tau_{m}\leq t$, and  the sum of terms $\chi(\tau_{m}<\varsigma    )  K_{\tau_{m}} |K_{\tau_{m}}|$ can be written as
\begin{align}\label{HakiSak}
\sum_{m=1}^{\mathbf{N}_{t}} \chi(\tau_{m}<\varsigma    ) K_{\tau_{m}} |K_{\tau_{m}}| = & \sum_{m=1}^{\lfloor \frac{1}{2}\mathbf{N}_{t} \rfloor } -\chi(\tau_{2m-1}<\varsigma    ) S(K_{\tau_{2m-1}  })\Big( |K_{\tau_{2m}}|^{2}-|K_{\tau_{2m-1}}|^{2}     \Big)\nonumber \\ &+\chi\big(\mathbf{N}_{t} \text{ odd}\big)\chi(\tau_{\mathbf{N}_{t}}<\varsigma    )  K_{\mathbf{N}_{t}} |K_{\mathbf{N}_{t}}| .
\end{align}
The remainder term on the second line is bounded by $\chi(\tau_{\mathbf{N}_{t}}<\varsigma    )   |K_{\mathbf{N}_{t}}|^2< 4\mathbf{p}^{2}$ and tends to zero as $\lambda\searrow 0$ when multiplied by  $2\lambda^{\frac{\gamma+3}{2} }$.   Moreover,  by writing 
$$|K_{\tau_{m+1}}|^{2}-|K_{\tau_{m}}|^{2}=\Big( E(K_{\tau_{m+1}})-E(K_{\tau_{m}}) \Big)-\Big( E(K_{\tau_{m+1}})-|K_{\tau_{m+1}}|^{2}\Big)+\Big( E(K_{\tau_{m}})-|K_{\tau_{m}}|^{2}\Big)$$
and using the triangle inequality, 
\begin{align}\label{HighFlyin}
 \sup_{0\leq s\leq T} 2 \lambda^{\frac{\gamma+3}{2}}\Big|\sum_{m=1}^{\mathbf{N}_{\frac{s}{\lambda^{\gamma}}}}& -\chi(\tau_{m}<\varsigma    ) S(K_{\tau_{m}  })\Big( |K_{\tau_{m+1}}|^{2}-|K_{\tau_{m}}|^{2}     \Big) \Big|
\nonumber \\ \leq  &  \sup_{0\leq s\leq T} 2  \lambda^{\frac{\gamma+3}{2}}\Big|\sum_{m=1}^{\mathbf{N}_{\frac{T}{\lambda^{\gamma}}}} -\chi(\tau_{m}<\varsigma    ) S(K_{\tau_{m}  })\Big( E(K_{\tau_{m+1}})-E(K_{\tau_{m}})     \Big) \Big|\nonumber \\ &+2 \lambda^{\frac{\gamma+3}{2}}\sup_{p\in \R   }\big|  E(p)-p^{2}\big| \sum_{m=1}^{\mathbf{N}_{\frac{T}{\lambda^{\gamma}}} } \chi(\tau_{m}<\varsigma    ) .
 \end{align}
The supremum of $| E(p)-p^{2}|$ for  $p\in \R$ is bounded, and the expectation for the sum of terms $ \chi(\tau_{m}<\varsigma    )  $ is $\mathit{O}\big(\lambda^{-1-\gamma} \big)$, since 
\begin{align}\label{Jobba}
\mathbb{E}\Big[\sum_{m=1}^{\mathbf{N}_{\frac{T}{\lambda^{\gamma}}}} \chi(\tau_{m}<\varsigma    ) \Big] & \leq \frac{2}{\mathbf{p}}\mathbb{E}\Big[ \sum_{m=1}^{\mathbf{N}_{\frac{T}{\lambda^{\gamma}}}} \chi\big(  \tau_{m} <\varsigma \big)\big| K_{\tau_{m}}  \big|  \Big]\leq  \frac{4\nu}{\mathbf{p}}\mathbb{E}\Big[ \sum_{m=1}^{\mathbf{N}_{\frac{T}{\lambda^{\gamma}}}} \mathbb{E}\big[ \Delta\tau_{m}\,\big|\, \widetilde{\mathcal{F}}_{\tau_{m}^{-}}  \big]  \Big]\nonumber \\  & \leq \frac{4\nu}{\mathbf{p}} \mathbb{E}\Big[ \sum_{m=1}^{\mathbf{N}_{\frac{T}{\lambda^{\gamma}}}-1}  \Delta\tau_{m}+\sup_{1\leq m\leq \mathbf{N}_{\frac{T}{\lambda^{\gamma}} } }\chi\big(  \tau_{m} <\varsigma \big)\Delta\tau_{m} \Big]\nonumber \\ & \leq  \frac{4\nu }{\mathbf{p}}\Big(\frac{T}{\lambda^{\gamma}} +\mathit{O}(\lambda^{-1-\iota} ) \Big)=\mathit{O}\big(\lambda^{1-\gamma}\big),
\end{align}
where   $\iota \in (0,\frac{\gamma-1}{2})$ and the fourth inequality holds since $\sum_{m=1}^{\mathbf{N}_{t}-1}\leq t  $ and by Lem.~\ref{LemIntLen}.  The first inequality above uses that $|K_{t}|\geq \frac{1}{2}\mathbf{p}$ for $t<\varsigma$, and the second inequality applies Part (1) of Prop.~\ref{TimeFlip}.  Thus, the third line of~(\ref{HighFlyin}) decays as $\mathit{O}(\lambda^{\frac{5}{2}- \frac{\gamma}{2}} )$.

To bound the first term on the right side of~(\ref{HighFlyin}), it is convenient to write the summand as
 \begin{align}\label{Marianna}
 E(K_{\tau_{m+1}})-E(K_{\tau_{m}}) = & \Big(E(K_{\overline{\tau}_{m+1} })-E(K_{\overline{\tau}_{m}})-(\overline{\tau}_{m+1}-\overline{\tau}_{m}  ) \Big)-\Big(  E(K_{\overline{\tau}_{m+1} })-E(K_{\tau_{m+1} })\Big)\nonumber \\ &+\Big(  E(K_{\overline{\tau}_{m} })-E(K_{\tau_{m} })\Big)+ (\overline{\tau}_{m+1}-\overline{\tau}_{m} ) ,
\end{align}
where $\overline{\tau}_{m}$ is the Poisson time following $\tau_{m}$.  The times  $\overline{\tau}_{m}$ are hitting times, and the sum
\begin{align}\label{Thorns}
\sum_{m=1}^{\mathbf{N}_{\frac{s}{\lambda^{\gamma}}}}& -\chi(\tau_{m}<\varsigma    ) S(K_{\tau_{m}  }) \Big(E(K_{\overline{\tau}_{m+1} })-E(K_{\overline{\tau}_{m}})-\big( \overline{\tau}_{m+1}-\overline{\tau}_{m}  \big) \Big) 
\end{align}
is a martingale, since  the terms $E(K_{\overline{\tau}_{m+1} })-E(K_{\overline{\tau}_{m}})$ have mean $\overline{\tau}_{m+1}-\overline{\tau}_{m}    $  when conditioned on the information known at time $\overline{\tau}_{m}$ by Part (2) of Prop.~\ref{SubMart}.  The supremum of the absolute value for~(\ref{Thorns}) over $s\in [0,T]$ can be bounded through Doob's maximal inequality and techniques used previously.  The sums associated the other three terms on the right side of~(\ref{Marianna}) are treated using the counting techniques in~(\ref{Jobba}).

\vspace{.5cm} 
\noindent Part (2):  The quadratic variation process for  $\mathbf{m}_{\lambda}$ has the following form:
 $$
  [\mathbf{m}_{\lambda},\mathbf{m}_{\lambda}]_{s}= 4\lambda^{\gamma+3}\sum_{m=1}^{\mathbf{N}_{\frac{s}{\lambda^\gamma }}}\chi(\tau_{m}<\varsigma)  K_{\tau_{m}}^{2}\big(\Delta\tau_{m}- \mathbb{E}\big[\Delta\tau_{m}\,\big|\, \widetilde{\mathcal{F}}_{\tau_{m}^{-}}\big]   \big)^{2} .$$  
  I will show the following convergences in probability:
\begin{enumerate}[(i).]  
\item  $\sup_{0\leq s\leq T}\Big|  4 \lambda^{\gamma+3 }\sum_{m=1}^{\mathbf{N}_{\frac{s}{\lambda^\gamma }}}  \chi(\tau_{m}<\varsigma)\Big(K_{\tau_{m}}^{2}\big(\Delta\tau_{m}- \mathbb{E}\big[\Delta\tau_{m}\,\big|\, \widetilde{\mathcal{F}}_{\tau_{m}^{-}}\big]   \big)^{2}- \nu^{-1} |K_{\tau_{m} }|^{3}\Delta\tau_{m}  \Big)\Big| \Longrightarrow 0,  $
\item $ \sup_{0\leq s\leq T}\Big|  \frac{4\lambda^{\gamma+3 } }{\nu}\sum_{m=1}^{\mathbf{N}_{\frac{s}{\lambda^\gamma }}} \chi(\tau_{m}<\varsigma) |K_{\tau_{m} }|^{3}\Delta\tau_{m}  - s\vartheta  \Big|\Longrightarrow 0 .$
\end{enumerate}

  \vspace{.4cm}

\noindent (i).  \hspace{.15cm}  The difference in the supremum can be written as $\mathbf{W}_{\lambda, s}^{(1)} +\mathbf{W}_{\lambda, s}^{(2)} +\mathbf{W}_{\lambda, s}^{(3)}$ for 
 \begin{align*}
 \mathbf{W}_{\lambda, s}^{(1)} :=&  4\lambda^{\gamma+3}\sum_{m=1}^{\mathbf{N}_{\frac{s}{\lambda^\gamma }}} \chi(\tau_{m}<\varsigma) K_{\tau_{m}}^{2}\Big(\big(\Delta\tau_{m}- \mathbb{E}\big[\Delta\tau_{m}\,\big|\, \widetilde{\mathcal{F}}_{\tau_{m}^{-}}\big]   \big)^{2}- \mathbb{E}\Big[\big(\Delta\tau_{m}- \mathbb{E}\big[\Delta\tau_{m}\,\big|\, \widetilde{\mathcal{F}}_{\tau_{m}^{-}}\big]   \big)^{2}\,\Big|\, \widetilde{\mathcal{F}}_{\tau_{m}^{-}}\Big] \Big), \\
\mathbf{W}_{\lambda, s}^{(2)}: =&  4\lambda^{\gamma+3}\sum_{m=1 }^{\mathbf{N}_{\frac{s}{\lambda^\gamma }}}  \chi(  \tau_{m}<\varsigma ) K_{\tau_{m}}^{2}\Big( \mathbb{E}\Big[\big(\Delta\tau_{m}- \mathbb{E}\big[\Delta\tau_{m}\,\big|\, \widetilde{\mathcal{F}}_{\tau_{m}^{-}}\big]   \big)^{2}\,\Big|\, \widetilde{\mathcal{F}}_{\tau_{m}^{-}}\Big]-\nu^{-1}|K_{\tau_{m}}|  \mathbb{E}\big[\Delta\tau_{m}\,\big|\, \widetilde{\mathcal{F}}_{\tau_{m}^{-}}\big]   \Big),\\
\mathbf{W}_{\lambda, s}^{(3)}: =& \frac{ 4\lambda^{\gamma+3}}{\nu}\sum_{m=1}^{\mathbf{N}_{\frac{s}{\lambda^\gamma }}} \chi(  \tau_{m}<\varsigma )  |K_{\tau_{m}}|^{3}\Big(\mathbb{E}\big[\Delta\tau_{m}\,\big|\, \widetilde{\mathcal{F}}_{\tau_{m}^{-}}\big] - \Delta\tau_{m}\Big).
\end{align*}
The processes $\mathbf{W}_{\lambda, s}^{(1)} $ and $\mathbf{W}_{\lambda, s}^{(3)}$ are martingales with respect to the filtration $\widetilde{\mathcal{F}}_{\lambda, s}$, 
The expectation $\mathbb{E}[(\mathbf{W}_{\lambda, \frac{T}{\lambda^{\gamma}} }^{(1)})^{2}]$ is equal to the expression in the top line of~(\ref{Lambert}).  Using Part (1) of Prop.~\ref{TimeFlip} for the first two inequalities below,  the second moment of $\mathbf{W}_{\lambda, \frac{T}{\lambda^{\gamma}} }^{(1)}$ is bounded through the following inequalities: 
\begin{align}\label{Lambert}
\mathbb{E}\Big[ 16\lambda^{2\gamma+6}\sum_{\substack{m=1, \\ \tau_{m}<\varsigma }  }^{\mathbf{N}_{\frac{T}{\lambda^\gamma }}}  K_{\tau_{m}}^{4}\mathbb{E}\Big[\Big(&\big(\Delta\tau_{m}- \mathbb{E}\big[\Delta\tau_{m}\,\big|\, \widetilde{\mathcal{F}}_{\tau_{m}^{-}}\big]   \big)^{2}- \mathbb{E}\big[\big(\Delta\tau_{m}- \mathbb{E}\big[\Delta\tau_{m}\,\big|\, \widetilde{\mathcal{F}}_{\tau_{m}^{-}}\big]   \big)^{2}\,\big|\, \widetilde{\mathcal{F}}_{\tau_{m}^{-}}\big] \Big)^{2}\,\Big|\,\widetilde{\mathcal{F}}_{\tau_{m}^{-}}\Big]   \Big]\nonumber  \\ & \leq \frac{768}{ \nu^{4}} \mathbb{E}\Big[ \lambda^{2\gamma+6} \sum_{m=1}^{\mathbf{N}_{\frac{T}{\lambda^\gamma }}}\chi(\tau_{m}<\varsigma) | K_{\tau_{m}}|^{8}   \Big]\nonumber \\ &\leq  \frac{1536}{ \nu^{3}}\mathbb{E}\Big[\lambda^{2\gamma+6} \sum_{m=1}^{\mathbf{N}_{\frac{T}{\lambda^\gamma }}}\chi(\tau_{m}<\varsigma)  |K_{\tau_{m}}|^{7}\mathbb{E}\big[\Delta\tau_{m}  \,\big|\,\widetilde{\mathcal{F}}_{\tau_{m}^{-}} \big]   \Big]
\nonumber \\ &\leq \frac{1536(2\mathbf{p}_{0})^{7} }{\nu^{3} } \lambda^{2\gamma -1 } \Big(\mathbb{E}\Big[ \sum_{m=1}^{\mathbf{N}_{\frac{T}{\lambda^\gamma }}-1}\Delta\tau_{m}     \Big]  +\mathbb{E}\Big[ \sup_{1\leq m\leq \mathbf{N}_{\frac{T}{\lambda^\gamma }} }\chi(\tau_{m}<\varsigma) \Delta\tau_{m}     \Big]  \Big)\nonumber  \\ &\leq \frac{1536(2\mathbf{p}_{0})^{7} }{\nu^{3}  }\lambda^{2\gamma-1}\Big(\frac{T}{ \lambda^{\gamma  } }+\mathit{O}(\lambda^{-1-\iota} )  \Big)=\mathit{O}(\lambda^{\gamma-1}), 
\end{align}
where the last inequality uses that $\sum_{m=1}^{\mathbf{N}_{t}-1}\Delta\tau_{m}<t$  and Lem.~\ref{LemIntLen} for $\iota \in (0, \gamma-1)$.   The first inequality in~(\ref{Lambert}) uses that
\begin{align}\label{Timbo}
\mathbb{E}\Big[\Big(\big(\Delta\tau_{m}- \mathbb{E}\big[\Delta\tau_{m}\,\big|\, \widetilde{\mathcal{F}}_{\tau_{m}^{-}}\big]   \big)^{2}- \mathbb{E}\big[\big(\Delta\tau_{m}- \mathbb{E}\big[\Delta\tau_{m}\,\big|\, \widetilde{\mathcal{F}}_{\tau_{m}^{-}}\big]   \big)^{2}\,\big|\, \widetilde{\mathcal{F}}_{\tau_{m}^{-}}\big] \Big)^{2}\,\Big|\,\widetilde{\mathcal{F}}_{\tau_{m}^{-}}\Big] &\leq  \mathbb{E}\big[|\Delta\tau_{m}|^{4}\big|\,\big|\, \widetilde{\mathcal{F}}_{\tau_{m}^{-}}\big]\nonumber \\ &\leq 48   |K_{\tau_{m}}|^{4},
\end{align}
 where the second inequality in~(\ref{Timbo}) holds for small enough $\lambda$ by approximating  $\mathbb{E}\big[|\Delta\tau_{m}|^{4}\,\big|\, \widetilde{\mathcal{F}}_{\tau_{m}^{-}}\big]$ with $\frac{4!}{\nu^4} |K_{\tau_{m}}|^{4}   $ using Part (1) of Prop.~\ref{TimeFlip} and multiplying $\frac{4!}{\nu^4}|K_{\tau_{m}}|^{4}   $  by $2$ to cover the error.   I use a similar trick in the third inequality of~(\ref{Lambert}) by applying Part (1) of Prop.~\ref{TimeFlip} to approximate $|K_{\tau_{m}}|$  by $\nu\mathbb{E}\big[|\Delta\tau_{m}|\,\big|\, \widetilde{\mathcal{F}}_{\tau_{m}^{-}}\big]$  and multiplying $\nu\mathbb{E}\big[|\Delta\tau_{m}|\,\big|\, \widetilde{\mathcal{F}}_{\tau_{m}^{-}}\big]$   by $2$ to cover the error again.  For the third inequality, I removed the nested conditional expectations and used that $|K_{\tau_{m}}|\leq 2\mathbf{p}=\frac{2\mathbf{p}_{0}}{\lambda}$ for $\tau_m\leq \varsigma$.  
 Finally, by Doob's maximal inequality, I have that   
$ \mathbb{E}\big[\sup_{0\leq s\leq T} \big|\mathbf{W}_{\lambda, s}^{(1)}\big|^{2} \big]$  converges to zero. 

 The expression $\mathbb{E}\big[\sup_{0\leq s\leq T} \big|\mathbf{W}_{\lambda, s}^{(2)}\big| \big]$ is bounded by similar applications of Part (1) of Prop.~\ref{TimeFlip} as above, and $ \mathbb{E}\big[\sup_{0\leq s\leq T} \big|\mathbf{W}_{\lambda, s}^{(3)}\big|^{2} \big]$ is bounded using Doob's maximal inequality followed by the standard techniques involving  Part (1) of Prop.~\ref{TimeFlip}.

 \vspace{.4cm} 

\noindent (ii). By the triangle inequality, I have the first inequality below:
\begin{align}\label{GasProm}
 \sup_{0\leq s\leq T}\Big| &  \frac{4 \lambda^{\gamma+3} }{\nu}\sum_{m=1}^{\mathbf{N}_{\frac{s}{\lambda^\gamma }}} \chi(\tau_{m}<\varsigma) |K_{\tau_{m} }|^{3}\Delta\tau_{m}  - s\vartheta  \Big|\nonumber \\ &\leq   \frac{4\lambda^{\gamma+3} }{\nu}\sum_{m=1}^{\mathbf{N}_{\frac{T}{\lambda^\gamma }} } \chi(\tau_{m}<\varsigma)\big| |K_{\tau_{m} }|^{3}-\mathbf{p}^{3}\big|\Delta\tau_{m} + \vartheta \Big| T- \lambda^{\gamma}\sum_{m=1}^{\mathbf{N}_{\frac{T}{\lambda^\gamma }}} \chi(\tau_{m}<\varsigma)\Delta\tau_{m}\Big|\nonumber  \\ &\leq 6 T\vartheta \lambda^{\epsilon}+ 2\vartheta \lambda^{\gamma}\sup_{1\leq m\leq  \mathbf{N}_{\frac{T}{\lambda^\gamma }}  }  \chi(\tau_{m}<\varsigma)\Delta\tau_{m} .
 \end{align}
For the first term in the second inequality,  I have used that $ | |K_{\tau_{m} }|-\mathbf{p}|\leq \lambda^{\epsilon}\mathbf{p}$ ~ for $\tau_{m}<\varsigma$.  The expectation for the supremum of $\chi(\tau_{m}<\varsigma)\Delta\tau_{m}$ for  $m\leq  \mathbf{N}_{\frac{T}{\lambda^\gamma }}$ increases with order $\mathit{O}(\lambda^{-1-\iota  })$ as $\lambda\searrow 0$ for any  $\iota \in (0,\gamma-1 )$ by  Lem.~\ref{LemIntLen}.   Thus, the rightmost expression in~(\ref{GasProm}) decays with order $\mathit{O}(\lambda^{\gamma-1-\iota}  )$.

\end{proof}

 \begin{lemma}[Lindberg condition]\label{LemLindberg}
As $\lambda\searrow 0$ there is  convergence
$\mathbb{E}\big[\sup_{0\leq s\leq T}\big| \mathbf{m}_{\lambda,s}-\mathbf{m}_{\lambda, s^{-} }      \big|\big]\rightarrow 0$.    
 \end{lemma}

\begin{proof}[Proof of \ref{LemLindberg}]
The largest jump for the martingale $\mathbf{m}_{\lambda,s}$ over the interval $s\in [0,T]$ is bounded by
\begin{align*}
\sup_{0\leq s\leq T}\big| \mathbf{m}_{\lambda,s}-\mathbf{m}_{\lambda, s^{-} }      \big| & \leq      2 \lambda^{\frac{\gamma+3}{2}}\sup_{0\leq m\leq \mathbf{N}_{\frac{T}{\lambda^\gamma}}}  \chi(\tau_{m}<\varsigma)\int_{\tau_{m}}^{\tau_{m+1}  }dt\, |K_{t}|\\ &\leq   \frac{2\lambda^{\frac{\gamma-1}{2}} }{\nu}\Big(\sup_{0\leq t\leq \frac{T}{\lambda^\gamma}}\lambda |K_{t}|\Big)^{2}\Big( \sup_{1\leq m\leq \mathbf{N}_{\frac{T}{\lambda^\gamma}}} \chi(\tau_{m}<\varsigma)\frac{\nu \Delta\tau_{m}}{|K_{\tau_{m}}| } \Big). 
\end{align*}
 By the Cauchy-Schwarz inequality, 
\begin{align}\label{JunkInTrunk}
\mathbb{E}\Big[\sup_{0\leq s\leq T}\big| \mathbf{m}_{\lambda,s}-\mathbf{m}_{\lambda, s^{-} }      \big|     \Big]  \leq  & \frac{2\lambda^{\frac{\gamma-1}{2}} }{\nu} \mathbb{E}\Big[ \Big(\sup_{0\leq t\leq \frac{T}{\lambda^\gamma}}\lambda|K_{t}|\Big)^{4}    \Big]^{\frac{1}{2}}  \mathbb{E}\Big[   \Big( \sup_{1\leq m\leq \mathbf{N}_{\frac{T}{\lambda^\gamma}}}\frac{\nu \Delta\tau_{m}}{|K_{\tau_{m}}|} \Big)^2      \Big]^{\frac{1}{2}}\\  =& \mathcal{O}\big(\lambda^{\frac{\gamma-1-\iota}{2}}\big).
 \end{align}
 The first expectation on the right side is uniformly finite for small $\lambda$ by  the inequality $|K_{t}|^{4}\leq E_{t}^2$, Doob's maximal inequality, and the bound on the second moments of $E_{t}$ from Part (3) of Prop.~\ref{SubMart}.  The second expectation on the right side has order $\mathit{O}(\lambda^{-\iota})$ for arbitrary  $\iota>0$ by the proof of Lem.~\ref{LemIntLen}.  Thus, I can choose $\iota \in (0,\gamma-1)$ so that~(\ref{JunkInTrunk}) tends to zero for small $\lambda$.

\end{proof}
 
\subsection{Proof of Theorem~\ref{ThmClassical}}

\begin{proof}[Proof of Thm.~\ref{ThmClassical}]

By Part (1) of Lem.~\ref{LemMartApprox}, I can approximate the process $\big(Y_{\lambda,s},\,s \in[0,T]\big)$ by the martingale $\big(\mathbf{m}_{\lambda, s},\,s\in[0,T]\big)$ in the limit $\lambda\searrow 0$.  By~\cite[Thm.VIII.2.13]{Pollard}, the martingale  $\big(\mathbf{m}_{\lambda, s},\,s\in[0,T]\big)$  converges in law to a Brownian motion with diffusion rate $\vartheta$ over the interval $s\in [0,T]$, if the following hold:
\begin{itemize}
\item  The random variables  $ \big|   \big[\mathbf{m}_{\lambda},\mathbf{m}_{\lambda}\big]_{s}-s\vartheta    \big|  $ converge in probability to zero as $\lambda\searrow 0$ for $s\in [0,T]$.

\item The random variables $\sup_{0\leq s\leq T}\big| \mathbf{m}_{\lambda,s}-\mathbf{m}_{\lambda, s^{-} }\big| $ converge in probability to zero as $\lambda\searrow 0$. 

\end{itemize}
The above statements are implied by Part (2) of Lem.~\ref{LemMartApprox} and Lem.~\ref{LemLindberg}, respectively.  The convergence  is with respect to the uniform metric on paths.

\end{proof}

\section*{Acknowledgments}
  This work is supported by the European Research Council grant No. 227772  and  NSF grant DMS-08446325.

\end{document}